\documentclass[a4paper,draft,11pt]{article}

\usepackage{pb-diagram,amsmath,amssymb}

\usepackage{amsthm}
\usepackage[T2A]{fontenc}
\usepackage[cp1251]{inputenc}
\usepackage[active]{srcltx}

\usepackage{amsthm}

\newtheorem{theorem}{Theorem}

\newtheorem{lemma}[theorem]{Lemma}

\newtheorem{proposition}[theorem]{Proposition}
\newtheorem{corollary}[theorem]{Corollary}

\newcommand{\FF}{\mathbb F}

\newcommand{\NN}{\mathbb N}

\newcommand{\PP}{\mathbb P}

\newcommand{\ZZ}{\mathbb Z}
\newcommand{\cA}{\mathcal A}

\newcommand{\cC}{\mathcal C}
\newcommand{\cD}{\mathcal D}
\newcommand{\cE}{\mathcal E}

\newcommand{\cG}{\mathcal G}
\newcommand{\cH}{\mathcal H}
\newcommand{\cI}{\mathcal I}

\newcommand{\cO}{\mathcal O}
\newcommand{\cP}{\mathcal P}

\newcommand{\cS}{\mathcal S}

\newcommand{\cW}{\mathcal W}

\def\myarrow{\ \hbox to 2em{\leaders
\hbox to 0.5ex{\hss\raise 0.55ex\hbox to 0.3ex{\hrulefill}\hss}
\hfill\,\llap{$>$}}\ }

\title{  Tangent Codes  \footnote{
 \protect\\ } }
\author{        }
\date{      }

\begin{document}
\maketitle

\centerline{\scshape Azniv Kasparian }
\medskip
{\footnotesize
 \centerline{Section of Algebra, Department of Mathematics and Informatics}
   \centerline{Kliment Ohridski University of Sofia}
   \centerline{ 5 James Bouchier Blvd., Sofia 1164, Bulgaria}
    \centerline{ {\bf email:} kasparia@fmi.uni-soifa.bg }}

\medskip

\centerline{\scshape Evgeniya Velikova   }
\medskip
{\footnotesize
 \centerline{ Section of Algebra, Department of Mathematics and Informatics }
   \centerline{Kliment Ohridski University of Sofia }
   \centerline{ James Bouchier Blvd., Sofia 1164, Bulgaria }
     \centerline{ {\bf email:} velikova@fmi.uni-sofia.bg }
}

\thispagestyle{empty}

\begin{abstract}
The present article  studies the finite Zariski tangent spaces to an affine variety $X$ as linear codes, in order to characterize their typical or exceptional properties by global geometric conditions on $X$.
The discussion concerns the generic minimum distance of a tangent code to $X$, its lower semi-continuity under a  deformation of $X$, as well as the existence of Zariski tangent spaces to $X$ with exceptional minimum distance.
Tangent codes are shown to admit simultaneous decoding.
The duals of the tangent codes to $X$ are realized by gradients of polynomials from the ideal of $X$.
We provide constructions of affine varieties with near MDS, cyclic or Hamming tangent codes.
Puncturing, shortening and extending finite Zariski tangent spaces are related to the corresponding operations on affine varieties.
The $(u \vert u+v)$ construction of tangent codes is associated with a fibered product of varieties.
Explicit constructions realize   linear Hamming isometries as differentials of morphisms of affine varieties.
\end{abstract}

\section{ Introduction  }

Codes with additional structure are usually equipped with a priori properties, which facilitate their characterization and decoding.
For instance, algebro-geometric Goppa codes allowed Tsfasman, Vl\v{a}dut and Zink to improve the asymptotic  Gilbert-Var\-sha\-mov bound  on the information rate for a fixed relative minimum distance (cf. \cite{TVZ}).
Justesen, Larsen, Elbr{\o}nd, Jensen, Havemose, Hoholdt, Skorobogatov, Vl\v{a}dut, Krachkovskii, Porter, Duursma, Feng, Rao and others developed efficient algorithms for decoding Goppa codes after obtaining the support of the error of the received word
 (Pellikaan's \cite{P} is a survey on these results.)
Duursma's considerations from \cite{D1} imply that the averaged homogeneous weight enumerator of Goppa codes, associated with a complete set of representatives of the linear equivalence classes of divisors of fixed degree is related to the $\zeta$-polynomial of the underlying curve.
The realizations of codes by points of a Grassmannian, a determinantal variety or a  modification of an arc provide other examples for exploiting "an extra structure" on the objects under study.
The interpretation of the finite Zariski tangent spaces to an affine variety $X$, defined over a finite field $\FF_q$ promisees to be useful for construction of extremal codes (see Corollary \ref{MinimizingLength}) and for simultaneous  decoding  of a family of codes, after recognizing the error support  of the received word (Corollary  \ref{SimultaneousDecoding}).
By grouping linear codes in families we acquire a "dynamic" point of  view, which is a natural prerequisite for studying  optimization problems.
Besides, our families are "integrable"   and  "geometric", so that various properties  of the tangent codes are characterized by global geometric conditions on the corresponding affine variety $X$ (see Proposition \ref{MinimumDistance}, Proposition \ref{ShorteningTangentCodes}, Proposition \ref{ExtendingDirectSumUV}).
Explicit constructions of affine varieties realize as finite Zariski tangent spaces near MDS codes (Proposition \ref{NearMDS}), cyclic codes (Corollary \ref{CyclicTangentCodes}), Hamming codes (Proposition \ref{HammingTangentCodes}).

Here is a synopsis of the paper.
Section 2 comprises some preliminaries on the Zariski topology and the Zariski tangent spaces $T_a (X, \FF_{q^m})$ of an affine variety $X$.

Our research starts  in section 3  by studying  the  minimum distance $d(T_a (X, \FF_{q^m}))$  of $T_a(X, \FF_{q^m})$.
Proposition \ref{MinimumDistance} (i) from subsection 3.1  establishes  that   lower bounds on  $d(T_a (X, \FF_{q^m}))$  hold "almost everywhere"
(i.e., on a Zariski  open subset of $X$)  if true at some point $a \in X( \FF_{q^m}) := X \cap \FF_{q^m}$.
It provides explicit equations of the exceptional locus
  $\{ b \in X \, \vert \, d(T_b (X, \FF_{q^m})) < d \,  \mbox{  for some } \,   m \in \NN \, \mbox{ with  } \,  \FF_{q^m}^n \ni b  \}$.
The deleting $\Pi _{\gamma} : X \rightarrow \Pi _{\gamma}(X) \subseteq \overline{\FF_q} ^{n - |\gamma|}$ of the components, labeled by
 $\gamma \subseteq \{ 1, \ldots , n \}$ is called the  puncturing of $X$ at $\gamma$.
The presence of a non-finite puncturing $\Pi _{\gamma} : X \rightarrow \Pi _{\gamma}(X)$ at $|\gamma| =d$ coordinates is shown to  require  all the tangent codes to $X$ to be of minimum distance $\leq d$ (Proposition \ref{MinimumDistance} (ii)).
The last part (iii) of Proposition \ref{MinimumDistance} verifies that $d(T_a (X, \FF_{q^m})) \geq d+1$ at "almost all"  the points of $X$ whenever all the puncturings $\Pi _{\gamma} : X \rightarrow \Pi _{\gamma} (X)$ at $|\gamma| = d$ variables are finite and separable (i.e., induce finite separable extensions $\overline{\FF_q}( \Pi _{\gamma} (X)) \subset \overline{\FF_q}(X)$ of the corresponding function fields).
Corollary \ref{MinimizingLength} from Proposition \ref{MinimumDistance} provides a sufficient condition for a puncturing $\Pi _i : X \rightarrow \Pi _i (X)$  of $X$ at a single variable $x_i$ to preserve the dimension $k$  and the minimum distance $d$ of a generic tangent space.
We hope that  it could serve as a reasonable base for choice of equations of affine varieties  $X$, whose puncturings realize extremal codes as their Zariski tangent spaces.
Corollary \ref{StabilizationMinDist} from subsection 3.2 constructs an "exotic"  embedding of $\overline{\FF_q}^k$ in $\overline{\FF_q}^n$, whose generic  Zariski tangent spaces are $[n,k,d]$-codes.
Any family  $\cC$ of linear codes   $\cC(a) \subset \FF_q^n$, parameterized by a subset of $\FF_q^n$ is interpolated by the union of the  Zariski tangent bundles of the irreducible components of an affine variety $X$, given by explicit equations  (cf. Proposition \ref{DestabilizationMinDist} from subsection 3.3).

Proposition \ref{ErrorSupport} from subsection 4.1  proves  that the set of the received words $w \in \FF_{q^m} ^n$  with a  $T_a (X, \FF_{q^m})$-error
$e \in \FF_{q^m} ^n$, supported by $i = \{ i_1, \ldots , i_t \} \subset \{ 1, \ldots , n \}$ coincides with  the Zariski tangent space
$T_a ( \Pi _i ^{-1} ( \overline{\Pi _i (X)}), \FF_{q^m})$ to the cylinder $ \overline{\Pi _i (X)}) \simeq \overline{\Pi _i (X)} \times \overline{\FF_q}^t$ over the Zariski closure $\overline{\Pi _i (X)}$ of the puncturing $\Pi _i (X)$ of $X$ at $i$.
That reduces the calculation of the error $e$ to solving a linear system of $t <n$ variables.
Moreover, Corollary \ref{SimultaneousDecoding}    provides an algorithm for simultaneous decoding of finite Zariski tangent spaces to $X$ of minimum distance $\geq 2t +1$.
The procedure  supplies several polynomial matrices  by the means of Groebner bases computations.
For an arbitrary received word $w \in \FF_{q^m} ^n$   with a $T_a (X, \FF_{q^m})$-error  $e \in \FF_{q^m} ^n$  of weight $\leq t$, one evaluates a part of the constructed polynomial matrices  at $a$ and calculates their products with the transposed $w^t$ of the received word  $w$, in order to recognize a $t$-tuple $i$, containing the support of $e$ and to obtain the components of $e$, labeled by $i$.
Subsection 4.2 discusses the dependence of the generic minimum distance of a  tangent code on the equations of $X$.
More precisely, Proposition \ref{LowerSemicontinuityMinDist} shows that a generic deformation of $X$ through a fixed point $a \in X$ with tangent code
 $T_a (X, \FF_{q^m})$ of minimum distance $d$ has generic minimum distance $\geq d$ of the finite Zariski tangent spaces.
 The last subsection 4.3 of section 4 is devoted to the dual codes $T_a (X, \FF_{q^m})^{\perp}$ of the finite Zariski tangent spaces.
By Lemma \ref{GradientBundleIsDualToTangentBundle}, $T_a (X, \FF_{q^m})^{\perp}$ consists of the gradients of the polynomials from $I(X, \FF_{q^m})$ at $a$.
Let $\beta \subset \{ 1, \ldots , n \}$ be a $d$-tuple of indices with complement $\neg \beta := \{ 1, \ldots , n \} \setminus \beta$.
Proposition \ref{DominantPuncturing} establishes that if "almost all" the points of $\overline{\FF_q}^d$ are from  the image of the  puncturing
 $\Pi _{\neg \beta} : X \rightarrow  \overline{\FF_q}^d$ and the absolute ideal $I(X, \overline{\FF_q})$ of $X$ contains a non-zero  polynomial in
  $x_{\beta} = \{ x_{\beta _1}, \ldots , x_{\beta _d} \}$ then the generic minimum distance of a gradient code to $X$ is $\leq d$.

Section 5 discusses tangent codes of special type.
Proposition \ref{NearMDS} from subsection 5.1 establishes that if $X$ admits a non-finite puncturing $\Pi _{\alpha} : X \rightarrow \Pi _{\alpha} (X)$ at
$|\alpha| = n-k$ variables and at least one tangent code to $X$ at a smooth point is a near MDS code then "almost all" the finite Zariski tangent spaces to $X$ are near MDS.
In a contrast, the locus of the cyclic tangent codes is shown to be exceptional in subsection 5.2.
Corollary \ref{CyclicTangentCodes} (i) provides the equations of this locus in the set   $X^{\rm smooth}$  of the smooth points of $X$.
Let $p$ be a prime integer, relatively prime to   $n \in \NN$ and $\FF_{p^m}$ be the splitting field of $t^n-1 \in \FF_p[t]$ over $\FF_p$.
We say that $T_a (X, \FF_{p^{s}})$ with $a \in X( \FF_{p^{s}})$ is a cyclic tangent code   if there is a cyclic code $C \subset \FF_{p^m} ^n$ with
$T_a (X, \FF_{p^{s}}) \otimes _{\FF_{p^s}} \FF_{p^{ms}} = C \otimes _{\FF_{p^{m}}}  \FF_{p^{ms}}$.
All cyclic codes of length $n$ over $\overline{\FF_p}$ are realized as finite Zariski tangent spaces to $X$ if for any cyclic code $C \subset \FF_{p^m} ^n$ there exists a  point $a \in X(\FF_{p^s})$, such that $C \otimes _{\FF_{p^m}} \FF_{p^{ms}} = T_a (X, \FF_{p^s}) \otimes _{\FF_{p^s}} \FF_{p^{ms}}$.
Explicit constructions provide   affine varieties, whose finite Zariski tangent spaces realize all the cyclic codes of length $n$ over $\overline{\FF_p}$.
Among them, there  is an example, whose all tangent codes are cyclic (cf. Corollary \ref{CyclicTangentCodes} (ii)).
 For any natural number $M$ we provide an  affine variety, whose finite Zariski tangent spaces realize all cyclic codes of length $n$ over $\overline{\FF_p}$ and contain at least $M$ non-cyclic tangent codes (Corollary  \ref{CyclicTangentCodes} (iii)).
 For an arbitrary finite field $\FF_q$ and an arbitrary natural number $r$, let us put $n := \frac{q^r-1}{q-1}$.
 Proposition \ref{HammingTangentCodes} from subsection 5.3  constructs an exotic embedding of $\overline{\FF_q}^{n-r}$ in $\overline{\FF_q}^n$, which has $q(q-1)^{n-r-1}$ Hamming tangent codes in the case of an odd characteristic ${\rm char} \FF_q$ or $q^{n-r}$ Hamming tangent codes for ${\rm char} \FF_q =2$.

 Section 6 draws parallels between operations on tangent codes and operations on the corresponding affine varieties.
 If  the    tangent codes to $X$ are of generic minimum distance $d$ and the differentials of a puncturing $\Pi _{\gamma} : X \rightarrow \Pi _{\gamma}(X)$ at $|\gamma| < d$ coordinates are generically  surjective   then the generic tangent vectors to $\Pi _{\gamma} (X)$ of minimum weight $d - |\gamma|$ are
shown in  Corollary \ref{PropertiesPuncturedTangentCodes} to be the puncturings of the generic tangent vectors to $X$ of weight $d$, containing $\gamma$ in its support.
For an appropriate index set $\gamma$ of cardinality $|\gamma| \leq \dim X$, Proposition \ref{ShorteningTangentCodes}  establishes the coincidence of the shortenings of the tangent codes to $X$ on $\gamma$ with the tangent codes to the shortening $X \cap V( x_i \, \vert \, i \in \gamma)$ of $X$ on $\gamma$.
Similarly,  the puncturings of the gradient codes to $X$ at $\gamma$  are exactly  the gradient codes to the shortening
 $X \cap V( x_i \, \vert \, i \in \gamma)$ of $X$ on $\gamma$.
 Another set of assumptions guarantee the coincidence of the shortenings of the  gradient codes to $X$  on $\gamma$ with the gradient codes to the puncturing
 $\Pi _{\gamma}(X)$ of $X$ at $\gamma$.
 The extension of a tangent code to $X$ is proved to be a Zariski tangent space to the extension of $X$ in subsection 6.3.
 The direct sum of finite Zariski tangent spaces to affine varieties $X,Y$ turns to be a Zariski tangent space to the direct product $X \times Y$.
 The $(u \, \vert \, u+v)$ construction of tangent codes is realized as a Zariski tangent space to a fibered product of appropriate affine varieties.

 The final, seventh section constructs a morphism $\overline{\FF_q}^n \rightarrow \overline{\FF_q}^n$, whose differentials are linear Hamming isometries of "almost all"  Zariski tangent spaces to "almost all" affine varieties $X \subset \overline{\FF_q}^n$ (Proposition \ref{ExistenceMorphismIsometricDifferentials}).
 An arbitrary family of linear Hamming isometries $\FF_q ^n \rightarrow \FF_q ^n$, parameterized by $\FF_q^n$ is interpolated by differentials of an explicit morphism $\overline{\FF_q} ^n \rightarrow \overline{\FF_q}^n$.

We conclude with an immediate consequence of results of Duursma, which relates the $\zeta$-polynomials of an appropriate family of Goppa codes with the $\zeta$-polynomial of the underlying curve.
This is one more evidence that the algebraic geometry provides a reasonable grouping of linear codes in families.

{\bf Acknowledgements:}   The authors are grateful  to  the referee of Finite Fields and their Applications for the useful remarks and suggestions.
The research is partially supported by    Contract 015/2014 and Contract  144/2015  with the Scientific Foundation of Kliment Ohridski University of Sofia.

\section{   Algebraic geometry preliminaries  }

Let $\overline{\FF_q} = \cup _{m=1} ^{\infty} \FF_{q^m}$ be the algebraic closure of the finite field $\FF_q$ with $q$ elements  and $\overline{\FF_q}^n$ be the $n$-dimensional affine space over $\overline{\FF_q}$.
An affine variety $X \subset \overline{\FF_q}^n$ is the common zero set
\[
X = V( f_1, \ldots , f_m) = \{ a \in \overline{\FF_q} ^n \ \ \vert  f_1(a) = \ldots = f_m(a) =0 \}
\]
of polynomials $f_1, \ldots , f_m \in \overline{\FF_q} [ x_1, \ldots , x_n]$.
We say that   $X \subset \overline{\FF_q}^n$ is defined over ${\mathbb F}_q$ and denote  $X / \FF_q \subset \overline{\FF_q}^n$ if the absolute ideal
\[
I(X, \overline{\FF_q}) := \{ f \in \overline{\FF_q} [ x_1, \ldots , x_n] \ \ \vert \ \  f(a) =0, \ \ \forall a \in X \}
\]
of $X$ is generated by polynomials $f_1, \ldots , f_m \in \FF_q [ x_1, \ldots , x_n]$ with coefficients from $\FF_q$.

The affine subvarieties of $X$ form  a  family of closed subsets.
The corresponding topology is referred to as the  Zariski topology on $X$.
The Zariski closure $\overline{M}$ of a subset $M \subseteq X$ is defined as the intersection of the Zariski closed subsets $Z$ of $X$, containing $M$.
It is easy to observe that $\overline{M} = VI(M, \overline{\FF_q})$ is the affine variety of the absolute ideal
$I(M, \overline{\FF_q}) \triangleleft \overline{\FF_q} [ x_1, \ldots , x_n] $ of $M$.
A subset $M \subseteq X$ is Zariski dense if its Zariski closure $\overline{M} = X$ coincides with $X$.
A property $\cP (a)$, depending on a point $a \in \overline{\FF_q}^n$ holds at a generic point of an affine variety $X \subset \overline{\FF_q} ^n$ if there is a Zariski dense subset $M \subseteq X$, such that $\cP (a)$ is true for all $a \in M$.

An affine variety $X \subset \overline{\FF_q}^n$ is irreducible if any decomposition $X = Z_1 \cup Z_2$ into a union of Zariski closed subsets $Z_j \subseteq X$ has $ Z_1 =X$ or $Z_2 =  X$.
This holds  exactly when the absolute  ideal $I(X, \overline{\FF_q}) \triangleleft \overline{\FF_q} [ x_1, \ldots , x_n]$ of $X$ is prime, i.e.
 $fg \in I(X, \overline{\FF_q})$ for $f,g \in \overline{\FF_q} [ x_1, \ldots , x_n]$  requires $f \in I(X, \overline{\FF_q})$ or $g \in I(X, \overline{\FF_q})$.
 A prominent property of the irreducible affine varieties $X$ is the Zariski density of an arbitrary non-empty Zariski open subset $U \subseteq X$.
This is equivalent to $U \cap W \neq \emptyset$ for any non-empty Zariski open subsets $U \subseteq X$ and $W \subseteq X$.

For an arbitrary irreducible affine variety $X / \FF_q \subset \overline{\FF_q}^n$, defined over $\FF_q$ and an arbitrary constant field
 $\FF_q \subseteq F \subseteq \overline{\FF_q}$, the affine coordinate ring
 \[
 F[X] := F [ x_1, \ldots , x_n] / I(X,F)
 \]
 of $X$over $F$ is an integral domain.
 The fraction field
 \[
 F(X) := \left \{  \frac{\varphi _1}{\varphi _2} \ \ \Big \vert \ \ \varphi _1, \varphi _2 \in F[X], \ \ \varphi _2 \neq 0 \in F[X] \right \}
 \]
 of $F[X]$ is called the functional field of $X$ over $F$.
 The points $a \in X$ correspond to the maximal ideals $I(a, \overline{\FF_q}) \triangleleft \overline{\FF_q} [ x_1, \ldots , x_n]$, containing
  $I(X, \overline{\FF_q})$.
  For any $F$-rational point $a \in X(F) := X \cap F^n$ the localization
  \[
  \cO _a (X,F) := \left \{ \frac{\varphi _1}{\varphi _2} \ \ \Big \vert \ \ \varphi _1, \varphi _2 \in F[X], \ \ \varphi _2 (a) \neq 0 \right \}
  \]
  of $F[X]$ at $F[X] \setminus \left( I(a,F) / I(X,F) \right)$ is the local ring of $a$ in $X$ over $F$.
 An $F$-linear derivation $D_a : \mathcal{O}_a (X,F) \rightarrow F$ at $a \in X(F)$ is an $F$-linear map, subject to Leibnitz-Newton rule
   $D_a ( \psi _1 \psi _2) = D_a ( \psi _1) \psi _2 (a) + \psi _1 (a) D_a ( \psi _2)$  for   $\forall \psi _1, \psi _2 \in \mathcal{O}_a (X,F)$.
 The $F$-linear space
 \[
 T_a (X,F) := {\rm Der} _a ( \mathcal{O}_a (X,F), F)
  \]
  of the $F$-linear derivations $D_a : \mathcal{O}_a (X,F) \rightarrow F$ at $a \in X(F)$ is called the Zariski tangent space to $X$ at $a$ over $F$.

  In order to derive a coordinate description of $T_a (X,F)$, note that any $F$-linear derivation $D_a : \cO _a (X,F) \rightarrow F$  at $a \in X(F)$ restricts to an $F$-linear derivation $D_a : F[X] \rightarrow F$ at $a$.
  According to
  \[
  D_a ( \varphi _1) = D_a \left( \frac{\varphi _1}{\varphi _2} \right) \varphi _2 (a) + \frac{\varphi _1(a)}{\varphi _2 (a)} D_a ( \varphi _2) \ \ \mbox{  for  } \ \  \forall \varphi _1, \varphi _2 \in F[X] \ \ \mbox{  with  } \ \  \varphi _2 (a) \neq 0,
  \]
  any $F$-linear derivation $D_a :F [X] \rightarrow F$ at $a  \in X(F)$ has unique extension to an $F$-linear  derivation $D_a : \cO_a (X,F) \rightarrow F$ at $a$.
In such a way, there arises an $F$-linear isomorphism
\[
T_a (X,F) \simeq {\rm Der} _a (F[X], F).
\]
Any $F$-linear  derivation $D_a : F [X] \rightarrow F$ of the affine ring $F[X]$ of $X$ at $a \in X(F)$ lifts to an $F$-linear  derivation
 $D_a : F [ x_1, \ldots , x_n] \rightarrow F$ of the polynomial ring at $a$, vanishing on the ideal $I(X, F)$ of $X$ over $F$.
If $I(X,F) = \langle f_1, \ldots , f_m \rangle _F \triangleleft F [ x_1, \ldots , x_n]$ is generated by $f_1, \ldots , f_m \in F [ x_1, \ldots , x_n]$
 then for arbitrary $g_1, \ldots , g_m \in F[ x_1, \ldots , x_n]$ one has
\[
D_a \left( \sum\limits _{i=1} ^m f_i g_i \right) = \sum\limits _{i=1} ^m D_a ( f_i) g_i (a)
\]
 and the Zariski tangent space
 \[
 T_a (X,F) \simeq \{ D_a \in {\rm Der} _a (F[ x_1, \ldots , x_n],F) \ \ \vert \ \  D_a (f_1) = \ldots =  D_a (f_m) =0 \}
 \]
 to $X$ at $a$ consists of the derivations $D_a : F [ x_1, \ldots , x_n] \rightarrow F$ at $a$, vanishing on $f_1, \ldots , f_m$.
 In such a way, the coordinate description  of  $T_a (X,F)$ reduces to the coordinate description of
 \[
 {\rm Der} _a (F[ x_1, \ldots , x_n], F) = {\rm Der} _a (F [ \overline{\FF_q} ^n], F) = T_a ( \overline{\FF_q}^n,F).
 \]
  In order to endow $T_a ( \overline{\FF_q} ^n, F)$ with a basis  over $F$,  let us note that the polynomial ring
  \[
  F [ x_1, \ldots , x_n] = F [ x_1 - a_1, \ldots , x_n - a_n] = \oplus _{i=0} ^{\infty} F [ x_1 - a_1, \ldots , x_n - a_n] ^{(i)}
  \]
  has a natural grading by the $F$-linear spaces $F [ x_1 - a_1, \ldots , x_n - a_n] ^{(i)}$ of the
    homogeneous polynomials  on $x_1 - a_1, \ldots , x_n - a_n$ of degree $i \geq 0$.
  An arbitrary $F$-linear derivation $D_a : F [ x_1, \ldots , x_n] \rightarrow F$ at $a \in F^n$ vanishes on
  $F [ x_1 - a_1, \ldots , x_n - a_n] ^{(0)} = F$  and on the homogeneous polynomials $F [ x_1 - a_1, \ldots , x_n - a_n] ^{(i)}$ of degree $i \geq 2$.
  Thus, $D_a$ is uniquely determined by its restriction to the $n$-dimensional space
  \[
  F [ x_1 - a_1, \ldots , x_n - a_n] ^{(1)} = {\rm Span} _F ( x_1 - a_1, \ldots , x_n - a_n)
  \]
  over $F$.
  That enables to  identify  the Zariski tangent space
  \[
  T_a ( \overline{\FF_q} ^n, F) \simeq {\rm Der} _a (F [ x_1, \ldots , x_n], F) \simeq {\rm Hom} _F ( F[ x_1 - a_1, \ldots , x_n - a_n] ^{(1)}, F)
  \]
  to $\overline{\FF_q} ^n$ at $a$ with the space of the $F$-linear  functionals on $F [ x_1 - a_1, \ldots , x_n - a_n] ^{(1)}$.
  Note that $x_1 - a_1, \ldots , x_n - a_n$ is a basis of $F [ x_1 - a_1, \ldots , x_n - a_n] ^{(1)}$ over $F$ and denote by
  $\left( \frac{\partial }{\partial x_1} \right) _a, \ldots , \left( \frac{\partial}{\partial x_n} \right) _a$ its dual basis.
In other words, $\left( \frac{\partial}{\partial x_j} \right) _a \in T_a ( \overline{\FF_q}^n, F)$ are the uniquely determined $F$-linear functionals on
 $F [ x_1 - a_1, \ldots , x_n - a_n] ^{(1)}$ with
 \[
 \left( \frac{\partial}{\partial x_j} \right) _a (x_i - a_i) = \delta _{ij} =
 \begin{cases}
 1  &  \text{ for $ 1 \leq i = j \leq n$,  }  \\
 0  &  \text{ for $1 \leq i \neq j \leq n$.   }
 \end{cases}
 \]
 As a result, the Zariski tangent space to $X$ at $a \in X(F)$ over $F$ can be described as
 \[
 T_a (X,F)  = \left \{ v = \sum\limits _{j=1} ^n v_j \left( \frac{\partial}{\partial x_j} \right) _a \ \ \Big \vert \ \
  \sum\limits _{j=1} ^n v_j \frac{\partial f_i}{\partial x_j} (a) =0, \ \ 1 \leq   i \leq m \right \}
\]
for any generating set $f_1, \ldots , f_m$ of $I(X,F) = \langle f_1, \ldots , f_m \rangle _F$.
If
\[
\frac{\partial  f}{\partial  x} = \frac{\partial ( f_1, \ldots , f_m)}{\partial ( x_1, \ldots , x_n)} =
\left(  \begin{array}{ccc}
\frac{\partial f_1}{\partial x_1}   &  \ldots  &  \frac{\partial f_1}{\partial x_n}  \\
\mbox{   }   &  \mbox{   }   &  \mbox{   }  \\
\frac{\partial f_m}{\partial x_1}  &  \ldots  &  \frac{\partial f_m}{\partial x_n}
\end{array}   \right)
\]
is the Jacobian matrix of $f_1, \ldots , f_m$ and $F = \FF_{q^s}$ is a finite field then $T_a(X, \FF_{q^s}) \subset \FF_{q^s} ^n$ is the $\FF_{q^s}$-linear code with parity check matrix $\frac{\partial f}{\partial x} (a) \in M_{m \times n} ( \FF_{q^s})$.

 Let   $X \subset \overline{\FF _q} ^n$  be  an irreducible affine variety with
$I(X, \overline{\FF _q}) = \langle f_1, \ldots , f_m \rangle _{\overline{\FF _q}}$, $f_1, \ldots , f_m \in \FF _q [ x_1, \ldots , x_n]$.
For  an arbitrary point $a = (a_1, \ldots , a_n) \in X$, let us denote by $\delta (a)$ the minimal natural number,  for which
 $a \in X( \FF _{q^{\delta(a)}}) := X \cap \FF _{q^{\delta (a)}}$ is an $\FF_{q^{\delta (a)}}$-rational point of $X$.
We say that $\FF _{q^{\delta (a)}}$ is the definition field of $a$ over $\FF_q$.
If $\FF _{q^{\delta (a_i)}}$ are the definition fields of $a_i \in \overline{\FF _q}$ over $\FF_q$  then $\delta (a)$ is the least common multiple of
 $\delta (a_1), \ldots , \delta (a_n)$.
 Note that $a \in X( \FF _{q^m})$ is an $\FF _{q^m}$-rational point if and only if $\delta (a)$ divides $m$.
 For all $l \in \NN$ the Zariski tangent spaces $T_a (X, \FF _{q^{l \delta (a)}})$ have one and a same parity check matrix
 \[
 \frac{\partial f}{\partial x}  (a) := \frac{\partial (f_1, \ldots , f_m)}{\partial (x_1, \ldots , x_n)} (a) \in M_{m \times n} ( \FF _{q^{\delta (a)}})
  \]
  and are uniquely determined by $T_a (X, \FF _{q^{\delta (a)}})$ as the  tensor products
  \[
  T_a (X, \FF _{q^{l \delta (a)}} ) = T_a (X, \FF _{q^{\delta (a)}}) \otimes _{\FF _{q^{\delta (a)}}} \FF ^{l \delta (a)}.
  \]
 In particular, $T_a (X, \FF _{q^{\delta(a)}})$ and $T_a (X, \FF _{q^{l \delta (a)}})$  have equal   dimension
  $n - {\rm rk} _{\FF _{q^{\delta (a)}}} \frac{\partial f}{\partial x} (a)$ over $\FF _{q^{\delta (a)}}$, respectively, over $\FF _{q^{l \delta (a)}}$.
 The minimum distances of $T_a (X, \FF _{q^{\delta (a)}})$ and $T_a (X, \FF _{q^{l \delta (a)}})$ coincide, as far as they equal  the minimal natural number $d$
 for which $\frac{\partial f}{\partial x}(a)$ has $d$ linearly dependent columns.
From now on, we write $\dim T_a (X, \FF_{q^{\delta (a)}})$ for the dimension of $T_a (X, \FF_{q^{\delta (a)}})$ over $\FF_{q^{\delta (a)}}$.

Let $X = X_1 \cup \ldots \cup X_s$ be a reducible affine variety and $a \in X_{i_1} \cap \ldots \cap X_{i_r}$ with $1 \leq i_1 < \ldots < i_r \leq s$ be a common point of $r \geq 2$ irreducible components $X_{i_j}$ of $X$.
In general, $X_{i_j}$ have different Zariski tangent spaces at $a$ and the union
$T_a (X_{i_1}, \FF_{q^{\delta (a)}}) \cup \ldots \cup T_a (X_{i_r}, \FF_{q^{\delta (a)}})$ is not an $\FF_{q^{\delta (a)}}$-linear subspace of
 $\FF_{q^{\delta (a)}} ^n$.
 That is why we define the Zariski tangent space $T_a (X, \FF_{q^{\delta (a)}})$ to a reducible variety $X \subset \overline{\FF_q}^n$ at a point $a \in X$ as the $\FF_{q^{\delta (A)}}$-linear code of length $n$ with parity check matrix
 \[
 \frac{\partial f}{\partial x} (a) = \frac{\partial ( f_1, \ldots , f_m)}{\partial (x_1, \ldots , x_n)} (a) \in M_{m \times n} ( \FF_{q^{\delta (a)}}),
 \]
  for some generators $f_1, \ldots , f_m \in \FF_q [ x_1, \ldots , x_n]$ of the absolute ideal
   $I(X, \overline{\FF_q}) = \langle f_1, \ldots , f_m \rangle _{\overline{\FF_q}}$ of $X$.

For an arbitrary finite set $S$ and an arbitrary natural number $t \leq |S|$ let us denote by $\Sigma _t (S)$ the set of the $t$-tuples with entries from $S$.
In the case of $S = \{ 1, \ldots , n \}$, we write  $\Sigma _t(1, \ldots , n)$ instead of $\Sigma _t( \{ 1, \ldots , n \})$.

For  a  systematic study of the Zariski tangent spaces  to an affine variety see \cite{Sh}, \cite{BCGB}, \cite{OG}, \cite{R} or \cite{H}.

\section{Minimum distance of a tangent code}

\subsection{Typical minimum distance of a tangent code}


The minimum distance of a linear code $C \subset \FF_q^n$ is related to the kernels of the puncturings of $C$.
Note that the puncturing
\[
\Pi _{\gamma} :  T_a (X, \FF_{q^{\delta (a)}}) \longrightarrow \Pi _{\gamma}T_a (X, \FF_{q^{\delta (a)}}) \subseteq  \FF_{q^{\delta (a)}} ^{n-|\gamma|}
\]
of a finite Zariski tangent space to $X$ coincides with the differential
\[
\Pi _{\gamma} = ( d \Pi _{\gamma}) _a : T_a (X, \FF_{q^{\delta (a)}}) \longrightarrow T_{\Pi _{\gamma}(a)}  ( \Pi _{\gamma} (X), \FF_{q^{\delta (a)}})
\]
of the puncturing
\[
\Pi _{\gamma} : X \longrightarrow \Pi _{\gamma} (X)
\]
 of the corresponding irreducible affine variety $X$.
That allows to study the minimum distance of $T_a (X, \FF_{q^{\delta (a)}})$ by the global properties of $\Pi _{\gamma}: X \rightarrow \Pi _{\gamma} (X)$.

In order to formulate precisely, let us recall that a finite morphism $\varphi : X \rightarrow \varphi (X)$ is called separable if the finite extension
$\overline{\FF_q} ( \varphi (X)) \subseteq \overline{\FF_q}(X)$ of the corresponding function fields is separable.
This  means that the minimal polynomial $g_{\xi} (t) \in \overline{\FF_q} ( \varphi (X)) [t]$ of an arbitrary element $\xi  \in \overline{\FF_q}(X)$ over $\overline{\FF_q} ( \varphi (X))$ has no multiple roots.

A morphism $\varphi : X \rightarrow \varphi (X)$ is etale at  some point  $a \in X$, if the differential
$(d \varphi) _a : T_a (X, \overline{\FF_q}) \rightarrow T_{\varphi (a)} (\varphi (X), \overline{\FF_q}) $
of $\varphi$ at $a$ is an $\overline{\FF_q}$-linear embedding.
Let us denote by ${\rm Etale}(\varphi)$ the set of the points $a \in X$, at which the morphism  $\varphi : X \rightarrow \varphi (X)$ is etale.

 \begin{lemma}      \label{FiniteSeparableAndEtalePuncturings}
 Let $X / \FF_q \subset \overline{\FF_q}^n$ be an irreducible affine variety, defined over $\FF_q$ and
 $\Pi _{\gamma} : X \rightarrow \Pi _{\gamma} (X) \subseteq \overline{\FF_q} ^{n - |\gamma|}$ be a puncturing at a subset
 $\gamma \subseteq \{ 1, \ldots , n \}$.

(i) The puncturing $\Pi _{\gamma} : X \rightarrow \Pi _{\gamma} (X)$ is etale at $a \in X$ if and only if the Zariski tangent space
 $T_a (X, \FF_{q^{\delta (a)}})$  does not contain a non-zero word $v(a)$ with support ${\rm Supp} (v(a)) \subseteq \gamma$.

 (ii) If the set $ {\rm Etale} ( \Pi _{\gamma}) \cap \Pi _{\gamma} ^{-1} (\Pi _{\gamma} (X)^{\rm smooth}) \neq \emptyset$ is non-empty then  the puncturing
$\Pi _{\gamma} : X \rightarrow \Pi _{\gamma} (X)$ is a finite morphism,
 ${\rm Etale} ( \Pi _{\gamma}) \cap \Pi _{\gamma} ^{-1} (\Pi _{\gamma} (X)^{\rm smooth}) \subseteq X^{\rm smooth}$ and  the  differentials
\[
(d \Pi _{\gamma}) _a : T_a (X, \FF_{q^{\delta (a)}})  \longrightarrow T_{\Pi _{\gamma} (a)} ( \Pi _{\gamma} (X), \FF_{q^{\delta (a)}})
\]
are  surjective at all  the points $a \in {\rm Etale} ( \Pi _{\gamma}) \cap \Pi _{\gamma} ^{-1} (\Pi _{\gamma} (X)^{\rm smooth})$.

(iii) If the puncturing $\Pi _{\gamma} : X \rightarrow \Pi _{\gamma} (X)$ is a finite separable morphism then
${\rm Etale} ( \Pi _{\gamma}) \cap \Pi _{\gamma} ^{-1} ( \Pi _{\gamma} (X) ^{\rm smooth})$ is a   Zariski dense subset of $X$.
\end{lemma}

 \begin{proof}

(i) The kernel of the differential
$(d \Pi _{\gamma}) _a : T_a (X, \FF_{q^{\delta (a)}})  \rightarrow T_{\Pi _{\gamma}(a)} ( \Pi _{\gamma} (X), \FF_{q^{\delta (a)}})$
consists of the tangent vectors $v(a) \in T_a (X, \FF_{q^{\delta (a)}})$ with support ${\rm Supp} (v(a)) \subseteq \gamma$.

(ii) Note that $\dim T_a (X, \FF_{q^{\delta (a)}}) \geq \dim X =k$ at any point $a \in X$.
If    $a \in {\rm Etale} ( \Pi _{\gamma}) \cap \Pi _{\gamma} ^{-1} ( \Pi _{\gamma} (X) ^{\rm smooth})$   then
$(d \Pi _{\gamma}) _a : T_a (X, \FF_{q^{\delta (a)}})  \rightarrow T_{\Pi _{\gamma} (a)} ( X, \FF_{q^{\delta (a)}})$ is injective  and
$\dim T_{\Pi _{\gamma} (a)} ( X, \FF_{q^{\delta (a)}}) = \dim \Pi _{\gamma}(X)$.
Combining with $\dim \Pi _{\gamma} (X) \leq \dim X$, one obtains
\begin{align*}
\dim X \leq \dim T_a (X,  \FF_{q^{\delta (a)}} ) = \dim  (d \Pi _{\gamma}) _a T_a (X, \FF_{q^{\delta (a)}}) \leq
\dim T_{\Pi _{\gamma}(a)} ( \Pi _{\gamma}(X), \FF_{q^{\delta (a)}}) =  \\
 \dim \Pi _{\gamma} (X)  \leq \dim X.
\end{align*}
Therefore   $(d \Pi _{\gamma}) _a T_a (X, \FF_{q^{\delta (a)}}) = T_{\Pi _{\gamma}(a)} ( \Pi _{\gamma}(X), \FF_{q^{\delta (a)}}) $,
$\dim X =   \dim T_a (X, \overline{\FF_q})$ and the dimensions   $\dim \Pi _{\gamma} (X)  = \dim X$ coincide.
In other words, the differential  $(d \Pi _{\gamma}) _a : T_a (X, \FF_{q^{\delta (a)}})  \rightarrow T_{\Pi _{\gamma} (a)} ( \Pi _{\gamma}(X), \FF_{q^{\delta (a)}})$ is surjective,  $a \in X^{\rm smooth}$ is a smooth point   and   $\Pi _{\gamma} : X \rightarrow \Pi _{\gamma} (X)$ is a finite morphism.

(iii)  In order to show that ${\rm Etale} ( \Pi _{\gamma})$  contains a non-empty Zariski open subset of $X$, let us consider the minimal polynomials
\[
g_i(t, \overline{x_{\neg \gamma}} ) = \sum\limits _{s=0} ^{N_i} \frac{\varphi _{i,s} ( \overline{x_{\neg \gamma}})}{  \psi _{i,s} ( \overline{x_{\neg \gamma}})} t^s \in \overline{\FF_q} ( \overline{x_{\neg \gamma}} ) [t]
\]
 of $\overline{x_i} = x_i + I(X, \overline{\FF_q}) \in \overline{\FF_q}(X)$ over
 $\overline{\FF_q}( \Pi _{\gamma} (X)) = \overline{\FF_q}( \overline{\neg \gamma})$ with
  $\varphi _{i,s} ( \overline{x_{\neg \gamma}}), \psi _{i,s} ( \overline{x_{\neg \gamma}}) \in \overline{\FF_q}  [ \overline{x_{\neg \gamma}} ] :=
   \overline{\FF_q} [ x_{\neg \gamma}] / I(X, \overline{\FF_q})$.
   Denote by $\psi_i ( x_{\neg \gamma}) \in \overline{\FF_q}[ x_{\neg \gamma}] $ the least common multiple  of  the polynomials
    $\psi _{i,s} ( x_{\neg \gamma}) \in \overline{\FF_q}[x _{\neg \gamma}]$, $0 \leq s \leq N_i$.
 The  polynomial $\psi _i ( x_{\neg \gamma})$ is well defined up to a factor from $\overline{\FF_q} [ x_{\neg \gamma}] ^* = \overline{\FF_q}^*$.
 The product
 \[
 f_i(x_i, x_{\neg \gamma}) := \psi _i ( x_{\neg \gamma})  g_i (x_i, x_{\neg \gamma}) \in \overline{\FF_q} [ x_i, x_{\neg \gamma}] \cap I(X, \overline{\FF_q})
 \]
 is a polynomial of minimal degree with respect to $x_i$  from $I(X, \overline{\FF_q})$.
We claim that the Zariski open subset
\[
W :=  X \setminus V \left( \prod\limits _{i \in \gamma} \frac{\partial f_i}{\partial x_i} \right)
\]
 is non-empty and contained in ${\rm Etale} ( \Pi _{\gamma})$.
More precisely,  $f_{\gamma_1}, \ldots , f_{\gamma _d} \in I(X, \overline{\FF_q})$  for $\gamma = \{ \gamma _1, \ldots , \gamma _d \}$ implies that
 $T_a (X, \overline{\FF_q})$ is contained in the solution set $\cS$ of the homogeneous  linear system with coefficient matrix
 $H_{\gamma} = \frac{\partial f_{\gamma}}{\partial x} (a)$.
 If we arrange the variables $x = \{ x_1, \ldots , x_n \}$ in two groups - $x_{\gamma}$ and $x_{\neg \gamma}$, then
  $\frac{\partial f_{\gamma}}{\partial x_{\gamma}} (a)$ is a diagonal matrix with non-zero entries for any  $a \in W$.
As a result, the components   $v_{\gamma} = v_{\gamma} ( v_{\neg \gamma})$ of  $v = (v_{\gamma}, v_{\neg \gamma}) \in \cS$, labeled by $\gamma$ can be expressed
 as homogeneous linear functions  of $v_{\neg \gamma}$ and the puncturing $\Pi _{\gamma} : \cS \rightarrow \Pi _{\gamma} (\cS)$ is injective.
 Then  the  restriction
  $\Pi _{\gamma} = (d \Pi _{\gamma}) _a : T_a (X, \overline{\FF_q}) \rightarrow T_{\Pi _{\gamma}(a)} ( \Pi _{\gamma}(X), \overline{\FF_q})$
 of $\Pi _{\gamma} \vert _ {\cS}$ is injective and $\Pi _{\gamma} : X \rightarrow \Pi _{\gamma} (X)$ is etale at any point
  $a \in W$.
  That justifies ${\rm Etale} ( \Pi _{\gamma}) \supseteq W$.

The assumption  $W = \emptyset$  implies that
\[
X \subseteq V \left( \prod\limits _{i \in \gamma} \frac{\partial f_i}{\partial x_i} \right).
\]
As a result, $\prod\limits _{i \in \gamma} \frac{\partial f_i}{\partial x_i} \in I(X, \overline{\FF_q})$.
The absolute ideal $I(X, \overline{\FF_q}) \triangleleft \overline{\FF_q} [ x_1, \ldots , x_n]$ of the irreducible affine variety
 $X \subseteq \overline{\FF_q}^n$ is prime and there follows $\frac{\partial f_i}{\partial x_i} \in I(X, \overline{\FF_q})$ for some $i  \in \gamma$.
Since $f_i \in I(X, \overline{\FF_q}) \setminus \{ 0 \}$ is of minimum degree with respect to $x_i$, one concludes that
$\frac{\partial f_i}{\partial x_i} \equiv 0 \in \overline{\FF_q} [ x_1, \ldots , x_n]$.
However,  $\frac{\partial f_i}{\partial t} (t, x_{\neg \gamma}) = \psi _i (x_{\neg \gamma}) \frac{\partial g_i}{\partial  t} (t, x_{\neg \gamma})$
implies that $\frac{\partial g_i}{\partial t} (t, x_{\neg \gamma}) \equiv 0$ and $g_i(t, \overline{x_{\neg \gamma}}) \in \overline{\FF_q} ( \overline{x_{\neg \gamma}}) [t]$ has a multiple root.
That contradicts the separability of the finite  extension $\overline{\FF_q} ( \overline{x_{\neg \gamma}}) \subseteq \overline{\FF_q}(X)$ and    proves that
  $ W \neq \emptyset$.

The non-empty Zariski open subset  $\Pi _{\gamma} (X) ^{\rm smooth} \subseteq \Pi _{\gamma} (X)$ pulls back to a non-empty Zariski open subset
$ W_{\gamma} :=\Pi _{\gamma} ^{-1} ( \Pi _{\gamma} (X) ^{\rm smooth}) \subseteq X$.
The  intersection $W \cap W_{\gamma}$ is a non-empty Zariski open and, therefore, Zariski dense subset of the irreducible affine variety $X$.
Thus, the Zariski closures $X \supseteq \overline{{\rm Etale} ( \Pi _{\gamma}) \cap W_{\gamma}} \supseteq \overline{W \cap W_{\gamma}} = X$ coincide with $X$ and ${\rm Etale} ( \Pi _{\gamma}) \cap \Pi _{\gamma} ^{-1} ( \Pi _{\gamma} (X) ^{\rm smooth})$ is Zariski dense in $X$.

 \end{proof}

 Note that Lemma \ref{FiniteSeparableAndEtalePuncturings} (ii) establishes a sort of a generalization of the Implicit Function Theorem, according to which  any
 puncturing $\Pi _{\gamma} : X \rightarrow \Pi _{\gamma} (X)$  with an injective differential at some point
 $ a \in \Pi _{\gamma} ^{-1} ( \Pi _{\gamma} (X))^{\rm smooth}$ is a  finite morphism.


 For an arbitrary irreducible affine variety $X / \FF _q \subset \overline{\FF _q} ^n$, defined over $\FF _q$, let us denote by
 \[
 X^{(\leq d)} := \{ a \in X \ \ \vert \ \  d( T_a (X, \FF _{q^{\delta (a)}}) \leq d \}
 \]
 the set of the points $a \in X$, at which the finite Zariski tangent spaces are of minimum distance $\leq d$.
 Similarly, put
 \[
 X^{(d)} := \{ a  \in X \ \ \vert \ \  d(T_a (X, \FF _{q^{\delta (a)}}) =d \} \ \ \mbox{  and   }
 \]
 \[
 X^{( \geq d)} := \{ a \in X \ \ \vert \ \  d(X, \FF _{q^{\delta (a)}}) \geq d \}.
 \]

The next proposition establishes that  if an irreducible affine variety $X$ admits a  tangent code $T_a (X, \FF_{q^{\delta (a)}})$ of minimum distance $\geq d$  then "almost all" finite Zariski tangent spaces to $X$ are of minimum distance $\geq d$.
If there is  a non-finite puncturing  $\Pi _{\gamma} : X \rightarrow \Pi _{\gamma} (X)$ at $|\gamma| =d$ variables, we show that all the tangent codes to $X$ are   of minimum distance $\leq d$.
When all the  puncturings $\Pi _{\gamma} : X \rightarrow \Pi _{\gamma} (X)$ at $|\gamma| =d$ variables are finite and separable, the minimum distance of a
finite Zariski tangent space to $X$ is bounded below by $d+1$ at "almost all"  the points of $X$.

\begin{proposition}    \label{MinimumDistance}
Let  $X \subset \overline{\FF _q} ^n$  be  an irreducible  affine variety of dimension $k \in \NN$  with
 $I(X, \overline{\FF _q}) = \langle f_1, \ldots , f_m \rangle _{\overline{\FF _q}}$ for some $f_1, \ldots , f_m \in \FF _q [ x_1, \ldots , x_n]$.

(i) For an arbitrary natural number $d \leq n-k+1$ the locus
\[
X^{( \leq d)} = V \left( \prod\limits _{i \in \Sigma _d (1, \ldots , n)} \det \frac{\partial f _{\varphi (i)}}{\partial x_i} \ \ \Big \vert \ \ \forall \varphi : \Sigma _d (1, \ldots , n) \rightarrow \Sigma _d (1, \ldots , m) \right)
\]
is a Zariski closed subset of $X$ and $X^{( \leq 1)} \subseteq X^{( \leq 2)} \subseteq \ldots \subseteq X^{( \leq n-k+1)} = X$.

(ii)   If there is a non-finite puncturing $\Pi _{\gamma} : X \rightarrow \Pi _{\gamma} (X)$ at $|\gamma| = d$ coordinates then $X = X^{( \leq d)}$.
 Moreover, in the case of  $X^{(d)} \neq \emptyset$ the  locus $X^{(d)} = X^{(\geq d)}$ is a Zariski dense, Zariski open  subset of $X$.

(iii) If for any $\gamma \in \Sigma _d (1, \ldots , n)$ the puncturing $\Pi _{\gamma} : X \rightarrow \Pi _{\gamma} (X)$ is finite and separable then
the subset $X^{( \geq d+1)} \subseteq X$ is Zariski dense.
\end{proposition}

\begin{proof}

(i)  For an arbitrary natural number $d \leq n-k$, observe that $a \in X^{ (>d)} = X^{( \geq d+1)}$ exactly when any $d$-tuple of columns of
 $\frac{\partial f}{\partial x} (a)$ is linearly independent.
That amounts to ${\rm rk} \frac{\partial f}{\partial x_i} (a) =
{\rm rk} \frac{\partial (f_1, \ldots , f_m)}{\partial (x_{i _1}, \ldots , x_{i_d})} (a) =d$
 for all $i \in \Sigma _d (1, \ldots , n)$.
By $k = \dim X \geq n-m$ there follows $m \geq  n-k \geq d$ and ${\rm rk} \frac{\partial f}{\partial x_i} (a) =d$ is equivalent to
 $\det \frac{\partial f_{\gamma}}{\partial x_i} (a) \neq 0$ for some $\gamma \in \Sigma _d (1, \ldots , m)$.
Thus,
 \begin{equation}   \label{LocusMinDistAtLeastD}
 \begin{split}
 X^{( \geq d+1)} = \cap _{i \in \Sigma _d (1, \ldots , n)} \left[ \cup _{\gamma \in \Sigma _d (1, \ldots , m)}
  \left( X \setminus V \left( \det \frac{\partial f _{\gamma}}{\partial x_i} \right) \right) \right] =  \\
\cap _{i \in \Sigma _d (1, \ldots , n)} \left[ X \setminus
 V \left( \det \frac{\partial f_{\gamma}}{\partial x_i} \ \ \Big \vert \ \    \gamma \in \Sigma _d (1, \ldots , m) \right) \right] =  \\
X \setminus \cup _{i \in \Sigma _d (1, \ldots , n)}
 V \left( \det \frac{\partial f_{\gamma}}{\partial x_i} \ \ \Big \vert \ \   \gamma \in \Sigma _d (1, \ldots , m) \right) =  \\
X \setminus V \left( \prod\limits _{i \in \Sigma _d (1, \ldots , n)} \det  \frac{\partial f _{\varphi (i)}}{\partial x_i} \ \ \Big \vert \ \
  \varphi : \Sigma _d (1, \ldots , n) \rightarrow \Sigma _d (1, \ldots , m) \right).
 \end{split}
\end{equation}
The last equality follows from $\cup _{i \in \Sigma _d (1, \ldots , n)} V(S_i) = V \left(  \prod\limits _{i \in \Sigma _d (1, \ldots , n)} S_i \right)$ for
\[
\prod\limits _{i \in \Sigma _d (1, \ldots , n)} S_i := \left  \{ \prod\limits _{i \in \Sigma _d (1, \ldots , n)} g_i \ \ \vert \ \  g_i \in S_i  \right \}, \ \
S_i := \left \{ \det \frac{\partial f_{\gamma}}{\partial x_i} \ \ \Big \vert \ \  \gamma \in \Sigma _d (1, \ldots , m) \right \}.
\]
The complement
\begin{align*}
X^{( \leq d)} = X \setminus X^{( \geq d+1)} =  \\
X \cap V \left( \prod\limits _{i \in \Sigma _d (1, \ldots , n)} \det \frac{\partial f_{\varphi (i)}}{\partial x_i} \ \ \Big \vert \ \
 \varphi : \Sigma _d (1, \ldots , n) \rightarrow \Sigma _d (1, \ldots m) \right)
\end{align*}
is a Zariski closed subset of $X$.

For an arbitrary point $a \in X$, note that $\dim   T_a (X, \FF _{q ^{\delta (a)}}) \geq \dim X = k$.
  Singleton bound on the minimum distance requires
\[
d (T_a(X, \FF _{q^{\delta (a)}} )) \leq n+1 - \dim T_a (X, \FF _{q^{\delta (a)}}) \leq n+1 -k.
\]
Thus, $X = X^{( \leq n-k+1)}$ is also a Zariski closed subset of $X$.

(ii)  Note that the non-empty  Zariski open subset $\Pi _{\gamma} (X) ^{\rm smooth}$ of $\Pi _{\gamma} (X)$ pulls back to a non-empty Zariski open subset
  $W_{\gamma} := \Pi _{\gamma} ^{-1} ( \Pi _{\gamma} (X) ^{\rm smooth})$ of $X$.
 We claim that at any $a \in W_{\gamma}$ the Zariski tangent space $T_a (X, \FF _{q^{\delta (a)}})$ contains a non-zero word, supported by $\gamma$.
 To this end, it suffices to establish that the differential
 \[
 (d \Pi _{\gamma} ) _a : T_a (X, \FF _{q^{\delta (a)}}) \longrightarrow T_{\Pi _{\gamma} (a)} ( \Pi _{\gamma} (X), \FF _{q^{\delta (a)}})
 \]
 of $\Pi _{\gamma}$ at $a$ is non-injective.
 Assume the opposite, i.e., that $\ker (d \Pi _{\gamma}) _a =0$.
 Then
 \[
 k \leq \dim  T_a (X, \FF _{q^{\delta (a)}}) \leq  \dim   T_{\Pi _{\gamma} (a) } ( \Pi _{\gamma} (X), \FF _{q^{\delta (a)}} ) = \dim \Pi _{\gamma} (X).
 \]
The morphism $\Pi _{\gamma} : X \rightarrow \Pi _{\gamma} (X)$ is not finite, so that $\dim \Pi _{\gamma} (X) < \dim X = k$.
That leads to a contradiction and implies that $\ker (d \Pi _{\gamma} ) _a \neq 0$ for $\forall a \in W_{\gamma} $.
As a result, $W_{\gamma} \subseteq X^{(\leq d)}$.
According to (i), $X^{( \leq d)}$ is a Zariski closed subset of $X$, so that the inclusion
 $X = \overline{W_{\gamma}} \subseteq \overline{X^{( \leq d)}} = X^{( \leq d)}$ of the corresponding Zariski closures is tantamount to $X = X^{( \leq d)}$.
Now, $X^{(d)} = X^{( \leq d)} \cap X^{( \geq d)} = X \cap X^{( \geq d)} = X^{( \geq d)}$ is a Zariski open subset of $X$, whereas Zariski dense for
$X^{(d)} \neq \emptyset$.

(iii)  Let $U_{\gamma} := {\rm Etale} ( \Pi _{\gamma}) \cap \Pi _{\gamma} ^{-1} ( \Pi _{\gamma} (X) ^{\rm smooth})$ and note that
 $U := \cap _{\gamma \in \Sigma _d (1, \ldots , n)} U_{\gamma}$ is contained in $X^{( \geq d+1)}$.
 Indeed, the presence of $a \in U \setminus X^{( \geq d+1)} = U \cap X^{( \leq d)}$ implies the existence of a tangent vector
  $v(a) \in T_a (X, \FF_{q^{\delta (a)}})$ of weight $\leq d$.
The support of $v(a)$ is contained in some $\gamma \in \Sigma _d (1, \ldots , n)$ and $0^n \neq v(a) \in \ker \Pi _{\gamma} = \ker (d \Pi _{\gamma}) _a$.
That contradicts $a \in {\rm Etale} ( \Pi _{\gamma})$ and justifies that $U \subseteq X^{( \geq d+1)}$.

We claim that $U$ is Zariski dense in $X$.
To this end, observe that $U \subseteq U_{\gamma}$ implies $I(U_{\gamma}, \overline{\FF_q}) \subseteq I(U, \overline{\FF_q})$ for
$\forall \gamma \in \Sigma _d (1, \ldots , n)$.
Therefore
\[
I(U, \overline{\FF_q}) \supseteq \sum\limits _{\gamma \in \Sigma _d (1, \ldots , n)} I(U_{\gamma}, \overline{\FF_q})
\]
 and the Zariski closure
\begin{align*}
\overline{U} = VI(U, \overline{\FF_q}) \subseteq V \left( \sum\limits _{\gamma \in \Sigma _d (1, \ldots , n)} I(U_{\gamma}, \overline{\FF_q}) \right) =  \\
\cap _{\gamma \in \Sigma _d (1, \ldots , n)} VI (U_{\gamma}, \overline{\FF_q}) = \cap _{\gamma \in \Sigma _d (1, \ldots , n)} \overline{U_{\gamma}} =
\cap _{\gamma \in \Sigma _d (1, \ldots , n)} X = X.
\end{align*}
Now $X = \overline{U} \subseteq \overline{X^{( \geq d+1)}} \subseteq X$ reveals the Zariski density of $X^{( \geq d+1)}$ in $X$.

\end{proof}

The proof of Proposition \ref{MinimumDistance} (iii) reveals that  for any point $a \in X ^{(d)}$  there exists a $d$-tuple of indices
$\gamma \in \Sigma _d (1, \ldots , n)$, such that the puncturing $\Pi _{\gamma} : X \rightarrow \Pi _{\gamma} (X)$   is not etale at $a$.

In the light  of Proposition \ref{MinimumDistance},  the finite Zariski tangent spaces  are expected to suit for construction of extremal codes.

\begin{corollary}    \label{MinimizingLength}
Let $X / \FF_q \subset \overline{\FF_q}^n$ be an irreducible $k$-dimensional affine variety, defined over $\FF_q$.
 Suppose that there exist $i \in \{ 1, \ldots , n \}$ and $\beta \subseteq \{ 1, \ldots , n \} \setminus \{ i \}$ with $|\beta| = d$, such that
  $\Pi _{\beta} : X \rightarrow \Pi _{\beta} (X)$ is a non-finite morphism and for any $\delta \in \Sigma _d (1, \ldots , n)$ with $i \in \delta$ the puncturing $\Pi _{\delta} : X \rightarrow \Pi _{\delta} (X)$ is a finite separable morphism.
  Then $\Pi _i : X \rightarrow \Pi _i (X)$ is a finite separable morphism, the generic tangent codes to $X$ are $[n,k,d]$-codes  and the generic finite Zariski tangent spaces  to $\Pi _i (X)$ are $[n-1,k,d]$-codes.
\end{corollary}

\begin{proof}

For any $\gamma \subseteq \{ 1, \ldots , n \} \setminus \{ i \}$ with $|\gamma| = d-1$  note that
 $\Pi _{\gamma \cup \{ i \}} : X \rightarrow  \Pi _{\gamma \cup \{ i \}} (X)$ is a finite and separable  puncturing.
 The factorization
 \[
 \begin{diagram}
 \node{X}  \arrow{s,l}{\Pi _{\gamma \cup \{ i \}}}  \arrow{e,t}{\Pi _i}  \node{\Pi _i (X)}  \arrow{sw,r}{\Pi _{\gamma}}  \\
    \node{\Pi _{\gamma \cup \{ i \}} (X)}   \node{\mbox{   }}
  \end{diagram}
 \]
implies the inclusions $\overline{\FF_q} ( \Pi _{\gamma \cup \{ i \}} (X)) \subseteq \overline{\FF_q} ( \Pi _i (X))   \subseteq \overline{\FF_q} (X)$
of the corresponding function fields.
The extensions $\overline{\FF_q} ( \Pi _{\gamma \cup \{ i \}} (X)) \subseteq \overline{\FF_q} ( \Pi _i (X))$ and
$\overline{\FF_q} ( \Pi _i (X))   \subseteq \overline{\FF_q} (X)$ are finite and separable, as far as
$\overline{\FF_q} ( \Pi _{\gamma \cup \{ i \}} (X))  \subseteq \overline{\FF_q}(X)$ is finite and separable.
Thus, $\Pi _i : X \rightarrow \Pi _i (X)$ is a finite separable morphism and the generic tangent codes to $\Pi _i (X)$ are of minimum distance $\geq d$ by Proposition \ref{MinimumDistance} (iii).

We claim that $\Pi _{\beta} : \Pi _i (X) \rightarrow \Pi _{\beta \cup \{ i \}} (X)$ is a non-finite puncturing at $|\beta| = d$ variables.
To this end, consider the commutative diagram
\[
\begin{diagram}
\node{X}   \arrow{s,l}{\Pi _{\beta}}  \arrow{e,t}{\Pi _i}  \arrow{se,l}{\Pi _{\beta \cup \{ i \}}} \node{\Pi _i (X)}   \arrow{s,r}{\Pi _{\beta}}  \\
\node{\Pi _{\beta} (X)}  \arrow{e,t}{\Pi _i}   \node{\Pi _{\beta \cup \{ i \}} (X)}
\end{diagram}.
\]
By assumption, $\Pi _{\beta} : X \rightarrow  \Pi _{\beta} (X)$ is a non-finite puncturing at $d$ variables, so that
$\Pi _{\beta \cup \{ i \}} : \Pi _i \Pi _{\beta} : X \rightarrow \Pi _{\beta \cup \{ i \}}  (X)$  is a non-finite puncturing at $d+1$ variables.
The presence of a factorization $\Pi _{\beta \cup \{ i \}} = \Pi _{\beta} \Pi _i : X \rightarrow  \Pi _{\beta \cup \{ i \}} (X)$ through a finite morphism
 $\Pi _i : X \rightarrow \Pi _i (X)$ suffices for $\Pi _{\beta} : \Pi _i (X) \rightarrow \Pi _{\beta \cup \{ i \}} (X)$ to be a non-finite morphism.
 Now, Proposition \ref{MinimumDistance}  (ii) applies to provide an upper bound $d$ on the minimum distance of all the  finite Zariski tangent spaces to
 $\Pi _i (X)$.
 As a result, the generic tangent codes to $\Pi _i (X)$ are of length $n-1$, dimension $k$ and minimum distance $d$.

\end{proof}

If $X / \FF_q \subset \overline{\FF_q}^n$ is an irreducible affine variety, defined over $\FF_q$ with a non-finite puncturing
 $\Pi _{\beta} : X \rightarrow \Pi _{\beta} (X)$ at $|\beta| =d$ coordinates and finite separable puncturings $\Pi _{\gamma} : X \rightarrow \Pi _{\gamma} (X)$ for $\forall \gamma \in \Sigma _{d-1} (1, \ldots , n)$ then the finite Zariski tangent spaces to $X$ are of minimum distance $d$ at "almost all" the points of $X$ (i.e., on a Zariski dense subset of $X$).
For fixed length $n$ and minimum distance $d$ one looks for tangent codes of maximal dimension.
These occur at the  singular points  $a \in X$, at which the parity check matrix of $T_a (X, \FF_{q^{\delta (a)}})$ is of minimum rank.

If the length $n$ and the dimension $k$ are fixed then one looks for a $k$-dimensional irreducible affine variety $X / \FF_q \subset \overline{\FF_q}^n$, which is in "most general" position with respect to the coordinate axes.
Namely, we seek such equations of $X$ which maximize the minimal integer $d \in \NN$, for which there is a non-finite puncturing
$\Pi _{\beta} : X \rightarrow \Pi _{\beta} (X)$ of $X$ at $|\beta| =d$ variables.

\subsection{ Reproducing the dimension and the minimum distance of a   code }

The Zariski tangent bundles of the coordinate $k$-dimensional affine subspace  $X_0 :=V(x_{k+1}, \ldots , x_n) \subset \overline{\FF _q} ^n$  are constant, i.e., $T_a(X, F) = F^k \times 0^{n-k}$ for any point $a \in X$ and any field  $\FF _q ^{\delta (a)} \subseteq F \leq \overline{\FF_q}$.
   The next corollary provides    different embeddings $X$ of $\overline{\FF _q}^k$  in $\overline{\FF _q} ^n$ with non-constant Zariski tangent bundles $T(X,F)$.
For any natural number $d \leq n-k$  and any $\sigma \in \Sigma _d (1, \ldots , n)$ we construct  such $X \simeq \overline{\FF _q}^k$  that
 $T_a (X, \FF _{q^{\delta (a)}})$ contains a non-zero word $v(a)$ of  ${\rm Supp} (v(a)) \subseteq \sigma$ for any $a \in X$.
 For an arbitrary $\FF_q$-linear  $[n,k,d]$-code $C$ we provide explicit equations of $\overline{\FF_q}^k \simeq X / \FF_q \subset \overline{\FF_q}^n$, such that "almost all"  tangent codes to $X$ are $[n,k,d]$-codes

\begin{corollary}     \label{StabilizationMinDist}
Let $C$ be an $[n,k,d]$-code and $\sigma \in \Sigma _d (1, \ldots , n)$ be a support of a non-zero word $c \in C \setminus \{ 0^n \}$.
Then there is an irreducible affine variety  $X/ \FF _q \subset \overline{\FF_q} ^n$ through $0^n$ with $T_{0^n} (X, \FF_q) = C$,  isomorphic to $\overline{\FF_q} ^k$ and   such that $T_a (X, \FF_{q^{\delta (a)}}) \simeq \FF _{q^{\delta (a)}} ^k$ contains a word
 $v(a) \in T_a (X, \FF _{q^{\delta (a)}}) \setminus \{ 0^n \}$ of  ${\rm Supp} ( v(a)) \subseteq \sigma$ for all $a \in X$.

In particular, $X = X^{( \leq d)}$ and    $\emptyset \neq X^{(d)} \cap X^{\rm smooth}  = X^{(\geq d)} \cap X^{\rm smooth} $ is such  a non-empty, Zariski open, Zariski dense subset of $X$   that the Zariski tangent spaces $T_a (X, \FF_{q^{\delta (a)}})$ are $[n,k,d]$-codes for all $a \in X^{(d)} \cap X^{\rm smooth}$.
\end{corollary}

\begin{proof}

Let $H \in M _{(n-k) \times n} ( {\mathbb F}_q)$ be a parity check matrix of  $C$  with columns $H_i \in M_{(n-k) \times 1} ( \FF_q)$.
Since ${\rm rk} H = n-k$, there exists $\alpha \in \Sigma _{n-k} (1, \ldots , n)$ with
$\det ( H_{\alpha}) = \det (H_{\alpha _1} \ldots H_{\alpha _{n-k}}) \neq 0$.
Without loss of generality, one can assume that $H_{\alpha} = ( H_{\alpha _1} \ldots H_{\alpha _{n-k}} ) = I_{n-k}$ for the identity matrix $I_{n-k}$ of size $n-k$ or
\[
H_{i, \alpha _j} = \delta _{ij} =
 \begin{cases}
1  &  \text{ for $1 \leq i = j \leq n-k$,  }  \\
0   &  \text{ for $1 \leq i \neq j \leq n-k$.   }
\end{cases}
\]
The presence of a word $c \in C$ with ${\rm Supp} (c) = \sigma \in \Sigma _d (1, \ldots , n)$ implies that
$H_{\sigma _{\nu}} \in {\rm Span} _{\FF _q} (H_{\sigma _1}, \ldots , H_{\sigma _{\nu -1}}, H_{\sigma _{\nu +1}}, \ldots , H_{\sigma _d})$
for all $1 \leq \nu \leq d$.
We claim that the existence of $\sigma _{\nu} \not \in \alpha$, since $\sigma \subseteq \alpha$ contradicts the linear dependence  of the columns
 $H_{\sigma _1}, \ldots , H_{\sigma _d}$.
If $H_{\sigma _{\nu}} = \sum\limits _{s \in \sigma \setminus \{ \sigma _{\nu} \}} \lambda _s H_s$ for $\sigma _{\nu} \not \in \alpha$ and  some
 $\lambda _s \in \FF _q$, then we choose polynomials
\[
f_{i, \alpha _j} ( x_{\alpha _j}) = H_{i, \alpha _j} x_{\alpha _j} = \delta _{ij} x_{\alpha _j} \in \FF _q [ x_{\alpha _j}] \ \ \mbox{  for   } \ \
\forall  1 \leq i, j \leq n-k,
\]
\[
f_{i,s} ( x_s) = H_{is} x_s + \sum\limits _{r \geq 2} a_{i,s,r} x_s ^r \in \FF _q [ x_s] \ \ \mbox{  for  } \ \
\forall 1 \leq i \leq n-k, \ \ \forall s \in \{ 1, \ldots , n \} \setminus ( \alpha \cup \{ \sigma _{\nu} \}
\]
and
\[
f_{i, \sigma _{\nu}} ( x_{\sigma _{\nu}}) = \sum\limits _{s \in \sigma \setminus \{ \sigma _{\nu} \}} \lambda _s f_{i,s}  ( x_{\sigma _{\nu}} )
 \in \FF _q [ x_{\sigma _{\nu}} ] \ \ \mbox{  for  }  \ \ \forall 1 \leq i \leq n-k.
\]
The polynomials
\[
f_i (x_1, \ldots , x_n) := \sum\limits _{s=1} ^n f_{i,s} (x_s) = x_{\alpha _i} + \sum\limits _{s \in \{ 1, \ldots , n \} \setminus \alpha} f_{i,s} ( x_s) \ \
\mbox{  for   } \ \ 1 \leq i \leq n-k
\]
cut out an affine variety $X = V( f_1, \ldots , f_{n-k}) \subset \overline{\FF_q} ^n$ with a biregular puncturing
 $\Pi _{\alpha}  : X \rightarrow \overline{\FF _q} ^k$ at $\alpha \in \Sigma _{n-k}  (1, \ldots , n)$.
 In particular, $X$ is irreducible and the polynomials $f_1, \ldots , f_{n-k} \in \FF _q [ x_1, \ldots , x_n]$ generate the absolute ideal
  $I(X, \overline{\FF _q}) = \langle f_1, \ldots , f_{n-k} \rangle _{\overline{\FF _q}}$ of $X$.
According to
\begin{align*}
\frac{\partial f_i}{\partial x_{\sigma _{\nu}}} = \frac{\partial f_{i, \sigma _{\nu}} }{\partial x_{\sigma _{\nu}}} ( x_{\sigma _{\nu}}) =
\sum\limits _{s \in \sigma \setminus \{ \sigma _{\nu} \}} \lambda _s \frac{\partial f_{i,s}}{\partial x_{\sigma _{\nu}}} ( x_{\sigma _{\nu}}) =
\sum\limits _{s \in \sigma \setminus \{ \sigma _{\nu} \}} \lambda _s \frac{\partial f_{i,s}}{\partial x_s} \Big \vert _{x_s = x_{\sigma _{\nu}}}   \\
= \sum\limits _{s \in \sigma \setminus \{ \sigma _{\nu} \}} \lambda _s \frac{\partial f_i}{\partial x_s} \ \
\mbox{ for  } \ \ \forall 1 \leq i \leq n-k,
\end{align*}
one has
\[
\frac{\partial f}{\partial x_{\sigma _{\nu}}} (a) = \sum\limits _{s \in \sigma  \setminus \{ \sigma _{\nu} \}} \lambda _s \frac{\partial f}{\partial x_s} (a)
\ \ \mbox{   for  } \ \ \forall a \in X,
\]
whereas ${\rm rk} \frac{\partial f}{\partial x_{\sigma}} (a) <d$.
In other words, $v(a) \in \FF _{q^{\delta (a)}} ^n$ with $v(a) _s := \lambda _s$ for $s \in \sigma \setminus \{ \sigma _{\nu} \}$, $v(a) _{\sigma _{\nu}} := -1$
and $v(a) _s :=0$ for $s \in \{ 1, \ldots , n \} \setminus \sigma$ belongs to $T_a (X, \FF _{q^{\delta (a)}}) \setminus \{ 0^n \}$ and has
${\rm Supp} ( v(a)) \subseteq \sigma$ for $\forall a \in X$.
Straightforwardly,
\[
\frac{\partial f_i}{\partial x_s} ( 0^n) = \frac{\partial f_{i,s}}{\partial x_s} (0) = H_{i,s} \ \ \mbox{  for  } \ \ \forall 1 \leq i \leq n-k, \ \ \forall 1 \leq s \leq n
\]
reveals that $\frac{\partial f}{\partial x} (0^n) = H$ and $T_{0^n} (X, \FF _q) = C$.

The  equality  $X = X^{( \leq d)}$ is an immediate consequence of Proposition \ref{MinimumDistance} (iii).
By Proposition \ref{MinimumDistance} (ii) and $T_{0^n} (X, \FF_q) = C$, $X^{(d)} = X^{(\geq d)}$is a non-empty, Zariski open, Zariski dense subset of $X$.
Due to the irreducibility of $X$, it intersects the smooth locus in a non-empty, Zariski open, Zariski dense subset $X^{(d)} \cap X^{\rm smooth}$ of $X$.

\end{proof}

\subsection{Tangent bundle interpolation of a finite family of codes }

The next proposition illustrates how linear codes of arbitrary minimum distance can be interpolated by a finite Zariski tangent bundle.

\begin{proposition}     \label{DestabilizationMinDist}
Let $\cC \rightarrow S$ be a family of ${\mathbb F}_q$-linear codes $\cC(a) \subset {\mathbb F}_q ^n$, $a \in S$ of length $n$,
  dimension  $k = \dim _{\FF_q} \cC(a)$ and arbitrary minimum distance $d (\cC (a)) \leq n+1-k$,
    parameterized by a subset $S \subseteq  \FF _q^n$.
Then there exist irreducible $k$-dimensional   affine varieties $Y_1 / \FF _q, \ldots , Y_s / \FF _q  \subset \overline{\FF _q}^n$ with
  \[
  \FF _q ^n \subset Y_1 \cup \ldots \cup Y_s, \ \ S \subseteq Y_1^{\rm smooth}( \FF _q) \cup \ldots \cup Y_s^{\rm smooth}( \FF _q),
  \]
   such that $T_a (Y_i, \FF_q) = \cC(a)$ for all $a \in S$ and all $Y_i \ni a$.
\end{proposition}

\begin{proof}

Let us choose  a family $\cH \rightarrow S$ of parity-check matrices $\mathcal{H}(a) \in M _{(n-k) \times n} ( \FF_q)$ of
 $\cC (a) \subset \FF_q ^n$ for all $a \in S$ and denote by $\cH (a) _{ij} \in \FF_q$ the entries of these matrices.
For an arbitrary $\beta \in \FF_q$, consider the Lagrange basis polynomial
\[
L^{\beta} _{\FF_q} (t) := \prod\limits _{\alpha \in \FF_q \setminus \{ \beta \}} \frac{t - \alpha}{\beta - \alpha}
\]
with $L^{\beta} _{\FF_q} (t) ( \beta ) =1$ and $L^{\beta} _{\FF_q} (t) \vert _{\FF_q \setminus \{ \beta \}} =0$.
Straightforwardly,
\begin{align*}
L^{0}_{\FF_q} (t) := \left[ \prod\limits _{\alpha \in \FF_q^*} (t - \alpha) \right]
 \left\{   \left[ \prod\limits _{\alpha \in \FF_q^*} (t - \alpha) \right] \Big \vert _{t=0}  \right \} ^{-1}   = \\
 (t^{q-1} -1) (-1)^{-1} = - (t^{q-1} -1).
\end{align*}
Towards an explicit calculation of $L^{\beta} _{\FF_q}(t)$ for $\beta \in \FF_q^*$, let us denote by $\sigma _1, \ldots , \sigma _{q-1}$ the elementary symmetric polynomials of the roots of $f(t) := \prod\limits _{\alpha \in \FF_q^*} (t - \alpha) = t^{q-1} -1$.
 Put $\tau _1, \ldots , \tau _{q-2}$ for the elementary symmetric polynomials of the roots of  the monic polynomial
\[
f_{\beta} (t) := \prod\limits _{\alpha \in \FF_q^* \setminus \{ \beta \}} (t - \alpha) = \frac{f(t)}{t - \beta} = \frac{t^{q-1} -1}{t  - \beta} =
t^{q-2}  + \sum\limits _{s=0} ^{q-3} (-1) ^{q-2-s} \tau _{q-2-s} t^s.
\]
 Then the relations
\begin{align*}
\tau _1 + \beta = \sigma _1 =0,  \\
\tau _s + \beta \tau _{s-1} = \sigma _s = 0 \ \ \mbox{  for  } \ \ \forall 2 \leq s \leq q-2 \ \ \mbox{  and  }   \\
\beta \tau _{q-2} = \sigma _{q-1} = (-1)^q,
\end{align*}
 reveal that  $\tau _s = ( - \beta) \tau _{s-1}$ for $\forall 1 \leq s \leq q-2$,  $\tau _0 := 1$  form a geometric progression
  $\{ \tau _s \} _{s=1} ^{q-2}$   with quotient $( - \beta)$.
As a result,
\[
\tau _s = ( - \beta) ^{s} \ \ \mbox{ for  } \ \ \forall 1 \leq s \leq q-2
\]
and
\[
f_{\beta} (t) =
t^{q-2} + \sum\limits _{s=0} ^{q-3} \beta ^{q-2-s} t^s =
t^{q-2} + \sum\limits _{s=0} ^{q-3} \beta ^{-s-1} t^s,
\]
according to $\beta ^{q-2} = \beta ^{-1}$ for $\forall \beta \in \FF_q^*$.
Now,
\begin{align*}
L^{\beta} _{\FF_q} (t) := \frac{t f_{\beta}(t)}{\beta f_{\beta} ( \beta)} =
\frac{t^{q-1} + \sum\limits _{s=1} ^{q-2} \beta ^{-s} t^s}{\beta ^{q-1} + \sum\limits _{s=1} ^{q-2} 1}  \\
= (q-1)^{-1} \left[ t^{q-1} + \sum\limits _{s=1} ^{q-2} \beta ^{-s} t^s \right] =
- \left[ t^{q-1} + \sum\limits _{s=1} ^{q-2} \beta ^{-s} t^s \right]
\end{align*}
for arbitrary $\beta \in \FF_q^*$.

Let us denote by
\begin{align*}
\Phi _p : \overline{\FF_q}^n \longrightarrow \overline{\FF_q}^n,  \\
\Phi _p (a_1, \ldots , a_n) = ( a_1 ^p, \ldots , a_n ^p) \ \ \mbox{ for } \ \ \forall a = (a_1, \ldots , a_n) \in \overline{\FF_q}^n
\end{align*}
the Frobenius automorphism of degree $p = {\rm char} \FF_q$ and consider the polynomials
\begin{align*}
f_i (x_1, \ldots , x_n) :=  \\
 \sum\limits _{b \in \Phi _p (S)} \left[ \sum\limits _{j=1} ^n \cH ( \Phi _p ^{-1} (b)) _{ij} (x_j - x_j ^q) \right]
L ^{b_1} _{\FF _q} (x_1^p) \ldots L^{b_n} _{\FF_q} ( x_n ^p) \in \FF_q [ x_1, \ldots , x_n]
\end{align*}
for $1 \leq i \leq n-k$.
Suppose that the  $X := V( f_1, \ldots , f_{n-k}) \subset \overline{\FF_q}^n$  decomposes into a union $X = Y_1 \cup \ldots \cup Y_s$ of its irreducible components $Y_i$.
We claim that $Y_1, \ldots , Y_s$ satisfy the announced conditions.
First of all, $\FF_q ^n \subset X$, as far as an arbitrary point $a = (a_1, \ldots , a_n) \in \FF_q ^n$ has components $a_j = a_j ^q$ and
 $f_i (a_1, \ldots , a_n) =0$ for $\forall 1 \leq i \leq n-k$.
The partial derivatives  are
\[
\frac{\partial f_i}{\partial x_j} = \sum\limits _{b \in \Phi _p (S)} \cH ( \Phi _p ^{-1} (b)) _{ij} L^{b_1} _{\FF_q} (x_1 ^p) \ldots L^{b_n} _{\FF_q} ( x_n ^p)
\]
and their values at $a \in S$ equal
\[
\frac{\partial f_i}{\partial x_j} (a) =  \cH ( \Phi _p ^{-1} \Phi _p (a)) _{ij} = \cH(a) _{ij}.
\]
Note that the composition of Lagrange interpolation polynomials with the Frobenius automorphism $\Phi _p$ is designed in such a way that to adjust
\[
\frac{\partial( f_1, \ldots , f_{n-k})}{\partial ( x_1, \ldots , x_n)} (a) = \cH (a)
\]
 at all the points $a \in S$.
 For an arbitrary point $a \in S \subset X( \FF_q) := X \cap \FF_q ^n$ let $Y_i$ be  an irreducible component of $X$ through $a$.
  The Zariski tangent space $T_a (Y_i, \FF_q)$ is contained in the right null-space $\cC(a)$ of the Jacobian matrix
  $\frac{\partial ( f_1, \ldots , f_{n-k})}{\partial ( x_1, \ldots , x_n)} (a) = \cH (a)$.
  On the other hand, $\dim Y_i \geq n - (n-k) =k$, as far as $Y_i$ is subject to at least $n-k$ polynomial equations $f_1 = \ldots = f_{n-k} =0$.
Thus,
 \[
 k \leq \dim Y_i \leq \dim _{\FF_q} T_a (Y_i, \FF_q)  \leq \dim_ {\FF_q} \cC (a) =k
 \]
 implies that $\cC (a) = T_a (Y_i, \FF _q)$, $\dim Y_i = k$ and any  $a \in Y_i ^{\rm smooth} (\FF_q)$.

\end{proof}

Algorithms for primary decomposition of polynomial ideals and decomposition of an affine variety into a union of irreducible components  are provided by  \cite{A}, \cite{GTZ}, \cite{EHV}, \cite{W}, \cite{S}, \cite{M}, \cite{PSS} and other sources.

\section{Decoding  and deforming  tangent   codes. Gradient codes. }

\subsection{Simultaneous  decoding  tangent codes}

After organizing linear codes in families, constituting tangent bundles to affine varieties, we propose an algorithm for simultaneous decoding of tangent codes, after  recognizing the support of the error of the received word.

Let $C \subset \FF_q ^n$ be an $\FF_q$-linear code of minimum distance $d$.
A word $w \in \FF_q ^n$ has a $C$-error of weight $\leq t$ if there exists $e \in \FF_q ^n$ of weight ${\rm wt} (e) \leq t$ with $w-e \in C$.
Denote by
\[
{\rm Err} (C,t) := \{ w \in \FF_q ^n \ \ \vert \ \  \exists e \in \FF_q ^n, \, {\rm wt} (e) \leq t, \, w-e \in C \}
 \]
 the set of the words with  a $C$-error of weight $\leq t$.
 Note that ${\rm Err} (C,t) \supset C$ and   ${\rm Err}(C, t) / C := \{ w +C \ \ \vert \ \ w \in {\rm Err}(C,t) \}$  consists  of the cosets, whose leaders are of weight $\leq t$.
We say that $w \in \FF_q^n$ has a $C$-error, supported by $i \in \Sigma _t (1, \ldots , n)$ if there is $e \in \FF_q ^n$ with ${\rm Supp} (e) \subseteq i$
 and $w - e \in C$.
 It is well known that if $t$ does not exceed the integral part $\left[ \frac{d-1}{2} \right]$ of $\frac{d-1}{2}$ and $w \in \FF_q ^n$ has a $C$-error of weight $\leq t$ then $e \in \FF_q ^n$ with ${\rm wt} (e) \leq t$ and $w-e \in C$ is unique.

 Let $X / \FF_q \subset \overline{\FF_q} ^n$ be an irreducible affine variety, defined over $\FF_q$ with
  $X^{( \geq 2t+1)} := \{ a \in X \ \ \vert \ \  d( T_a (X, \FF_{q^{\delta (a)}}) ) \geq 2t+1 \} \neq \emptyset$ for some $t \in \NN$.
  The disjoint  union
  \[
  {\rm Err} (TX, t) := \coprod\limits _{a \in X^{(2t+1)}}  {\rm Err} (T_a (X, \FF_{q^{\delta (a)}}),t)
  \]
  of the words ${\rm Err} (T_a (X, \FF_{q^{\delta (a)}}),t)$ with $T_a (X, \FF_{q^{\delta (a)}})$-error of weight $\leq t$ will be called the bundle of the words with $TX$-error of weight $\leq t$.
  Denote by
  \[
  \pi : {\rm Err} (TX,t) \longrightarrow X^{(2t+1)}
  \]
   the natural projection on the base.
  To any $i \in \Sigma _t (1, \ldots , n)$ we associate a polynomial matrix $A_i (x) \in M_{l_i \times t} ( \FF_q [ x_1, \ldots , x_n])$ for some $l_i \in \NN$, such that $w_a \in \pi ^{-1} (a) = {\rm Err} (T_a (X, \FF_{q^{\delta (a)}}),t)$ has   a $T_a (X, \FF_{q^{\delta (a)}})$-error $e_a \in \FF_{q^{\delta (a)}} ^n$, supported by $i$ exactly when $A_i (a) w_a ^t = 0$.
If so, then $e_a$ can be computed explicitly by the means of the Jacobian matrix  $\frac{\partial f}{\partial x} (a)$ of a generating set
 $f_1, \ldots , f_m \in \FF_q [ x_1, \ldots , x_n]$ of $I(X, \overline{\FF_q}) = \langle f_1, \ldots , f_m \rangle _{\overline{\FF_q}}$ at $a$.
 The aforementioned result  is referred to as a simultaneous decoding of tangent codes, as far as the construction of $A_i (x)$ for
 $\forall i \in \Sigma _t (1, \ldots , n)$ allows to recognize simultaneously the supports of the $T_a (X, \FF_{q^{\delta (a)}})$-errors for all $a \in X$ and to obtain  the corresponding errors  from  ${\rm Err} (T_a (X, \FF_{q^{\delta (a)}}),t)$.

Towards the construction of $A_i(x)$ we need the following

\begin{lemma}   \label{GroebnerBasis}
Let $X \subset \overline{\FF_q}^n$ be an irreducible affine variety with  absolute ideal
$I(X, \overline{\FF_q}) = \langle f_1, \ldots , f_m \rangle _{\overline{\FF_q}}$,  generated by $f_1, \ldots , f_m \in \FF_q [ x_1, \ldots , x_n]$,
$i \in \Sigma _t (1, \ldots , n)$  and  $G_{i, \neg i}$ be a Groebner basis of $I(X, \FF_q) = \langle f_1, \ldots , f_m \rangle _{\FF_q}$ with respect to a lexicographic order of  $\FF_q [ x_1, \ldots , x_n]$ with $x_i > _{\rm lex} x_{\neg i}$, obtained by Buchberger's algorithm.
Then:

\[
\mbox{ (i) } \ \ G_{i, \neg i} \subseteq M_{1 \times m} ( \FF_q [ x_1, \ldots , x_n])
\left( \begin{array}{c}
f_1  \\
\ldots  \\
f_m
\end{array}   \right);
\]

(ii)   the absolute ideal $I(\Pi _i (X), \overline{\FF_q} ) = I( \overline{\Pi _i (X)}, \overline{\FF_q}) \triangleleft \overline{\FF_q} [ x_{\neg i} ]$ of the Zariski  closure $\overline{\Pi _i (X)}$ of $\Pi _i (X)$ is generated by $G_{\neg i} := G_{i, \neg i} \cap \FF_q [ x_{\neg i}]$.
\end{lemma}

\begin{proof}

 (i)  Let $f := (f_1, \ldots , f_m) \in M_{1 \times m} ( \FF_q[ x_1, \ldots , x_n])$.
By an induction on the steps of   Buchberger's algorithm, we show that   any entry of the current generating set $\Delta$  of $I(X, \FF_q)$  is  of the
form  $gf^t $ for some $g = (g^{(1)}, \ldots , g^{(m)}) \in M_{1 \times m} ( \FF_q[ x_1, \ldots , x_n])$.

Let $e_j \in M_{1 \times m} ( \FF _q)$ be the ordered $m$-tuple with unique non-zero entry $1$ at the $j$-th position.
One can view $e_j$ as an  ordered $m$-tuple of polynomials.
At the input, any $f_j$ is expressed as a product $e_j f^t$ with the transposed $f^t$ of  $f = (f_1, \ldots , f_m) \in M_{1 \times m} ( \FF _q [ x_1, \ldots , x_n])$.

For arbitrary $h_1, h_2 \in \FF _q [ x_1, \ldots , x_n]$ let $x^{\gamma}$ be the least common multiple of the leading monomials $LM(h_1), LM(h_2)$ with respect to the fixed  lexicographic  order on $\FF _q [ x_1, \ldots , x_n]$.
If $LT(h_j) = LC(h_j) LM(h_j)$ are the leading terms of $h_j$ then the $S$-polynomial of $h_1, h_2$ (or the syzygy polynomial of $h_1, h_2$) is defined as
\[
S(h_1, h_2) = \frac{x^{\gamma}}{LT(h_1)} h_1 - \frac{x^{\gamma}}{LT(h_2)} h_2.
\]
If there exist $g_j \in M_{1 \times m} ( \FF _q  [ x_1, \ldots , x_n])$ with $h_j = g_j f^t$ then
\[
S(h_1, h_2) =  \left( \frac{x^{\gamma}}{LT(h_1)} g_1 - \frac{x^{\gamma}}{LT(h_2)} g_2 \right) f^t
\]
can also be represented as a product of a row  $m$-tuple of polynomials with the column $f^t = (f_1, \ldots , f_m) ^t$.
At each step, Buchberger's algorithm for  obtaining  a Groebner basis of $I(X, \FF _q)$ adjoins to the current generating set $\Delta$ of $I(X, \FF _q)$ the remainders $\overline{S(h_1, h_2)} ^{\Delta}$ of the $S$-polynomials of $h_1, h_2 \in \Delta$ under the division by $\Delta$.
By its very definition,
\[
\overline{S(h_1, h_2)} ^{\Delta} = S(h_1, h_2) - a f^t
 \]
 for some $a \in M_{1 \times m} ( \FF _q [ x_1, \ldots , x_n])$ and no-one monomial of $\overline{S(h_1, h_2)} ^{\Delta}$ with non-zero coefficient is divisible by $LT(h)$ for $h \in \Delta$.
Thus, if all $h \in \Delta$ are represented in the form $h = gf^t$ for some $g \in M_{1 \times m} ( \FF _q [ x_1, \ldots , x_n] )$, then all
 $\overline{S(h_1, h_2)} ^{\Delta}$ with $h_1, h_2 \in \Delta$ are of the same form and adjoining them to $\Delta$, one gets a set of polynomials
  $\Delta ' \in M_{1 \times m} ( \FF _q [ x_1, \ldots , x_n]) f^t$.

Buchberger's algorithm terminates with the Groebner basis $G_{i, \neg i} = \Delta$, once  $\overline{S(h_1, h_2)} ^{\Delta} =0$ for
$\forall h_1, h_2 \in \Delta$.

(ii)   Combining the Closure Theorem 3 from Chapter 3, § 2 \cite{CLSh} with the Elimination Theorem 2 from Chapter 3, § 1 \cite{CLSh}, one concludes that the Zariski closure of $\Pi _i (X)$ is the affine variety $\overline{\Pi _i (X)} = V( G_{\neg i})$.
Therefore, the absolute ideal
\[
I( \Pi _i (X), \overline{\FF_q}) = IVI ( \Pi _i (X), \overline{\FF_q}) = I ( \overline{\Pi _i (X)}, \overline{\FF_q}) = IV( G_{\neg i}, \overline{\FF_q} ) =
 r \langle G_{\neg i} \rangle _{\overline{\FF_q}}
\]
equals the radical $ r \langle G_{\neg i} \rangle _{\overline{\FF_q}}   \triangleleft  \overline{\FF_q} [ x_1, \ldots , x_n]$ of
$\langle G_{\neg i} \rangle _{\overline{\FF_q}}$.
If $h( x_{\neg i} ) \in r \langle G_{\neg i} \rangle _{\overline{\FF_q}}$ then
\[
h( x_{\neg i} )^N \in \langle G_{\neg i} \rangle _{\overline{\FF_q}} \subset \langle G_{\neg i} \rangle _{\overline{\FF_q}} \otimes _{\overline{\FF_q}} \overline{\FF_q} [ x_1, \ldots , x_n] \subseteq \langle G_{i, \neg i} \rangle _{\overline{\FF_q}} = I(X, \overline{\FF_q})
\triangleleft \overline{\FF_q} [ x_1, \ldots , x_n]
\]
for some $N \in \NN$.
The absolute ideal $I(X, \overline{\FF_q})$ of the irreducible affine variety $X$ is prime and therefore radical, i.e.,
 $I(X, \overline{\FF_q}) = r I(X, \overline{\FF_q})$.
 As a result, $h ( x_{\neg i} ) \in I(X, \overline{\FF_q})$ and bearing in mind that $h( x_{\neg i} ) \in \overline{\FF_q} [ x_{\neg i} ]$, one concludes that
 $h( x_{\neg i} ) \in I(X, \overline{\FF_q}) \cap \overline{\FF_q} [ x_{\neg i} ]$.
Buchberger's algorithm for obtaining the Groebner basis $G_{i, \neg i}$ of $I(X, \FF_q) = \langle f_1, \ldots , f_m \rangle _{\FF_q}$ adjoins to
 $\{ f_1, \ldots , f_m \}$ the remainders of appropriate $S$-polynomials and does not depend on the constant field $\FF_q$.
 Bearing in mind that $f_1, \ldots , f_m$ generate the absolute ideal $I(X, \overline{\FF_q}) = \langle f_1, \ldots , f_m \rangle _{\overline{\FF_q}}$ of $X$, one concludes that $G_{i, \neg i}$ is a Groebner basis of $I(X, \overline{\FF_q})$.
 By the Elimination Theorem 2 from Chapter 3, §1 \cite{CLSh},
  $G_{i, \neg i} \cap \overline{\FF_q} [ x_{\neg i} ] = G_{i, \neg i} \cap \FF_q [ x_{\neg i}] = G_{\neg i}$ is a Groebner basis of
   $I(X, \overline{\FF_q}) \cap \overline{\FF_q} [ x_{\neg i} ]$ with respect to the fixed lexicographic order of $\overline{\FF_q} [ x_{\neg i} ]$.
 Thus,
  \[
 h( x_{\neg i}) \in I(X, \overline{\FF_q}) \cap \overline{\FF_q} [ x_{\neg i} ] = \langle G_{\neg i} \rangle _{\overline{\FF_q}}
 \]
  and  $r \langle G_{\neg i} \rangle _{\overline{\FF_q}} \subseteq \langle G_{\neg i} \rangle _{\overline{\FF_q}}$.
 Combining with $\langle G_{\neg i} \rangle _{\overline{\FF_q}} \subseteq r \langle G_{\neg i} \rangle _{\overline{\FF_q}}$ one concludes that
 $I(\Pi _i (X), \overline{\FF_q}) = I( \overline{\Pi _i (X)}, \overline{\FF_q}) = r \langle G_{\neg i} \rangle _{\overline{\FF_q}} =
  \langle G_{\neg i} \rangle _{\overline{\FF_q}}$.

\end{proof}

In order to characterize the words $w \in \FF_{q^{\delta (a)}} ^n$, whose $T_a (X, \FF_q{\delta (a)})$-error  is supported by $i \in \Sigma _t (1, \ldots , n)$, let us note that the puncturing $\Pi _i$   can be viewed, both,  as a morphism
 $\Pi _i : \overline{\FF_q} ^n \rightarrow \overline{\FF_q} ^{n-t}$ of $\overline{\FF_q} ^n$ and  a morphism $\Pi _i : X \rightarrow \Pi _i (X)$ of the irreducible affine variety $X$.
If $\overline{\Pi _i (X)} \subset \overline{\FF_q} ^{n-t}$ is the Zariski closure of $\Pi _i (X)$ in $\overline{\FF_q} ^{n-t}$ then
\[
\Pi _i ^{-1} ( \overline{\Pi _i (X)}) := \{ a \in \overline{\FF_q} ^n \ \ \vert \ \  \Pi _i (a) \in \overline{\Pi _i (X)} \}
\]
is an irreducible affine variety of $\overline{\FF_q}^n$, isomorphic to $\overline{\Pi _i (X)} \times \overline{\FF_q} ^t$ and containing $X$.
We call $\Pi _i ^{-1} ( \overline{\Pi _i (X)})$ the cylinder  over  $\overline{\Pi _i (X)}$.

\begin{proposition}   \label{ErrorSupport}
Let $X \subset \overline{\FF _q} ^n$ be an affine variety, whose absolute ideal
$I(X, \overline{\FF _q}) = \langle f_1, \ldots , f_m \rangle _{\overline{\FF _q}}$    is generated  by some
 $f_1, \ldots , f_m \in \FF _q [ x_1, \ldots , x_n]$,  $i \in \Sigma _t(1, \ldots , n)$ and
  $\Pi _i ^{-1} ( \overline{\Pi _i (X)}) \simeq \overline{\Pi _i (X)} \times \overline{\FF_q}^t$ be the cylinder over  $\overline{\Pi _i (X)}$.
  Suppose that the tangent code $T_a (X, \FF_{q^{\delta (a)}})$ to $X$ at $a \in X$ is of minimum distance $d(T_a (X, \FF_{q^{\delta (a)}})) >t$ and
   $w \in \FF_{q^{\delta (a)}} ^n$ is a word with a $T_a (X, \FF_{q^{\delta (a)}})$-error of weight $\leq t$.
 Then the $T_a (X, \FF_{q^{\delta (a)}})$-error of $w$ is supported by $i$ if and only if
 $w \in T_a ( \Pi _i ^{-1} ( \overline{\Pi _i (X)}), \FF_{q^{\delta (a)}})$ is tangent to the cylinder $\Pi _i ^{-1} ( \overline{\Pi _i (X)})$ at $a$.
 If so, then any $T_a (X, \FF_{q^{\delta (a)}})$-error $e \in \FF_{q^{\delta (a)}} ^n$  of $w$ with $\Pi _i (e) = 0 _{(n-t) \times 1}$ projects onto a solution
  $\Pi _{\neg i}  (e) = (e_{i_1}, \ldots , e_{i_t}) = e_i$ of the homogeneous linear system
 \[
\frac{\partial f }{\partial x_i} (a)
\left(  \begin{array}{c}
x_{i_1}  \\
\ldots  \\
x_{i_t}
\end{array}   \right) =
\frac{\partial f }{\partial x} (a) w^t.
\]
\end{proposition}

\begin{proof}

The Zariski tangent space $T_a (X, \FF_{q^{\delta (a)}})$ does not contain a non-zero word $v(a)$ of ${\rm Supp} (v(a))  \subseteq i \in \Sigma _t (1, \ldots , n)$, according to $t < d(T_a (X, \FF_{q^{\delta (a)}})) $.
By Lemma \ref{FiniteSeparableAndEtalePuncturings} (i), one concludes that $\Pi _i : X \rightarrow \Pi _i (X)$ is etale at
$a \in \Pi _i ^{-1} ( \Pi _i (X))^{\rm smooth} \cap \FF_{q^{\delta (a)}}^n$.
Now Lemma \ref{FiniteSeparableAndEtalePuncturings} (ii)  applies to provide the surjectiveness of  the  differential
\begin{equation}   \label{DifferentialPuncturingAta}
(d \Pi _i) _a : T_a (X,F) \longrightarrow T_{\Pi _i (a)} ( \Pi _i (X), F).
\end{equation}

Let $G_{i, \neg i}$ be a Groebner basis of $I(X, \FF_q) = \langle f_1, \ldots , f_m \rangle _{\FF_q}$ with respect to a lexicographic order with
 $x_i > _{\rm lex} x_{\neg}$, obtained by Buchberger's algorithm.
 If $G_{\neg i} := G_{i, \neg i} \cap \FF_q [ x_{\neg i}] = \{ h_1, \ldots , h_l \}$ and
  $f = (f_1, \ldots , f_m) \in M_{1 \times m} ( \FF_q [ x_1, \ldots , x_n])$ then by Lemma \ref{GroebnerBasis} (i) there exist
   $g_1, \ldots , g_l \in M_{1 \times m} ( \FF_q [ x_1, \ldots , x_n])$ with $h_i = g_i f^t$ for $\forall  1 \leq i \leq l$.
  Consider the matrix $\cG_i \in M_{l \times m} ( \FF_q [ x_1, \ldots , x_n])$ with rows $g_1, \ldots   g_l$.
We claim  that the existence of $e \in \FF_{q^{\delta (a)}}^n$ with ${\rm Supp} (e) \subseteq i$ and $w - e \in T_a (X, \FF_{q^{\delta (a)}})$ is equivalent to
\[
\cG _i (a) \frac{\partial f^t}{\partial x} (a) w^t = \cG _i (a) \frac{\partial f^t}{\partial x} (a) e^t =0.
\]
To this end, note that $\cG _i f^t \in M_{l \times 1} ( \FF_q [ x_1, \ldots , x_n])$ consists of the entries of $G_{\neg i}$, arranged in a column.
Therefore
\[
\cG _i (a) \frac{\partial f^t}{\partial x} (a) =   \frac{\partial  (\cG _i f^t)}{\partial x} (a) =  \frac{\partial G_{\neg i}}{\partial x} (a) =
\left( 0  \ \ \frac{\partial G_{\neg i}}{\partial x_{\neg i} } (a) \right) = \left( 0 \ \ \cG _i (a) \frac{\partial f^t}{\partial x_{\neg i}} (a) \right),
\]
according to $f^t (a) = 0 _{m \times 1}$, whereas
\[
\cG _i (a) \frac{\partial f^t}{\partial x} (a) e^t = \cG _i (a) \frac{\partial f^t}{\partial x_{\neg i}} (a) \Pi _i (e) ^t.
\]
Any $e \in \FF_{q^{\delta (a)}}^n$ with ${\rm Supp} (e) \subseteq i$ has  $\Pi _i (e) = 0^{n-t}$, so that
 $\cG _i  (a) \frac{\partial f^t}{\partial x} (a) e^t =0$.
Conversely, if
\[
0 = \cG_i (a) \frac{\partial f^t}{\partial x} (a) e^t = \frac{\partial G_{\neg i} }{\partial x_{\neg i}} (a) \Pi _i (e)^t
\]
for some $e = ( \Pi _{\neg i} (e), \Pi _i (e)) \in \FF_{q^{\delta (a)}}^n$ then $\Pi _i (e) \in T_{\Pi _i (a)} ( \overline{\Pi _i (X)}, \FF_{q^{\delta (a)}})$.
By the surjectiveness of (\ref{DifferentialPuncturingAta}), there exists $v_o = (\Pi _{\neg i} (v_o), \Pi _i (e)) \in T_a (X, \FF_{q^{\delta (a)}})$.
Now,
\[
e_o := e - v_o = (\Pi _{\neg i} (e-v_o), 0)
\]
 has ${\rm Supp} (e_o) \subseteq i$ and $w - e_o = w - e + v_o \in T_a (X, \FF_{q^{\delta (a)}})$.

The    polynomials $h = \{ h_1, \ldots , h_l \} = G_{\neg i} \subset \FF_q [ x_{\neg i}]$ have Jacobian matrix
\[
\frac{\partial h}{\partial x} (a) = \frac{\partial G_{\neg i}}{\partial x} (a) = \frac{\partial}{\partial x} \left( \cG_i (x) f^t (x) \right) (a) = \cG_i (a) \frac{\partial f}{\partial x} (a)
\]
at $a \in X$.
 As a result,  $w \in \FF_{q^{\delta (a)}} ^n$ has a $T_a (X, \FF_{q^{\delta (a)}})$-error, supported by $i$ exactly when
 $\frac{\partial h}{\partial x} (a) w^t = 0 _{l \times 1}$.
By Lemma \ref{GroebnerBasis}  (ii), the absolute ideal
 $I( \Pi _i (X), \overline{\FF_q}) = I ( \overline{\Pi _i (X)}, \overline{\FF_q}) = \langle G_{\neg i} \rangle _{\overline{\FF_q}} =  \langle h_1, \ldots , h_l \rangle _{\overline{\FF_q}}$ of $\Pi _i (X)$ is generated by $h_1, \ldots , h_l \in \FF_q [ x_1, \ldots , x_n]$.
 The absolute ideal of the cylinder
 $\Pi _i ^{-1} ( \overline{\Pi _i (X)})$ is the extension
\[
 I( \Pi _i ^{-1} ( \overline{\Pi _i (X)})) = I ( \overline{\Pi _i (X)}, \overline{\FF_q}) \otimes _{\overline{\FF_q}} \overline{\FF_q} [ x_1, \ldots , x_n]
 \]
of $  I ( \overline{\Pi _i (X)}, \overline{\FF_q})$  to an ideal  of $\overline{\FF_q} [ x_1, \ldots , x_n]$.
Therefore, the Zariski tangent space to $\Pi _i ^{-1} ( \overline{\Pi _i (X)})$ at $a \in X \subseteq \Pi _i ^{-1} ( \overline{\Pi _i (X)})$ is
\[
T_a ( \Pi _i ^{-1} ( \overline{\Pi _i (X)}), \FF_{q^{\delta (a)}}) = \left  \{ w \in \FF_{q^{\delta (a)}} ^n \ \ \vert \ \
\frac{\partial G_{\neg i}}{\partial x} (a) w^t = \frac{\partial h}{\partial x} (a) w^t =    0 _{|G_{\neg i}| \times 1}  \right \}.
\]
In such a way we have established that  there is $e \in \FF_{q^{\delta (a)}} ^n$ with ${\rm Supp} (e) \subseteq i$ and $w-e \in T_a (X, \FF_{q^{\delta (a)}})$ if and only if
 $w \in T_a ( \Pi _i ^{-1} ( \overline{\Pi _i (X)}), \FF_{q^{\delta (a)}})$.
Making use of $\frac{\partial f}{\partial x} (a) (w^t - e^t) = 0_{m \times 1}$ and $\Pi _i (e) = 0^{n-t}$, one concludes that
\[
\frac{\partial f}{\partial x_i} (a) e_i ^t = \frac{\partial f}{\partial x} (a) e^t = \frac{\partial f}{\partial x} (a) w^t.
\]

\end{proof}

Note that Proposition \ref{ErrorSupport} can be used for computing the   leaders of  those cosets
$ w + T_a (X, \FF_{q^{\delta (a)}}) \in  \FF_{q^{\delta (a)}} ^n / T_a (X, \FF_{q^{\delta (a)}})$,
 which  admit representatives  of weight less that the minimum distance of  $  T_a (X, \FF_{q^{\delta (a)}})$.

Here is an algorithm for simultaneous decoding of tangent codes.

\begin{corollary}  \label{SimultaneousDecoding}
Let $X \subset \overline{\FF_q}^n$ be  an irreducible affine variety, whose absolute ideal
$I(X, \overline{\FF_q}) = \langle f_1, \ldots , f_m \rangle _{\overline{\FF_q}}$ is generated by $f_1, \ldots , f_m \in \FF_q [ x_1, \ldots , x_n]$
and
\[
\pi : {\rm Err} (TX, t) := \coprod\limits _{a \in X^{( \geq 2t+1)}} {\rm Err} (T_a (X, \FF_{q^{\delta (a)}}), t) \longrightarrow X^{( \geq 2t+1)}
\]
be the bundle of the words with $TX$-error of weight $\leq t$.
Then any $w \in {\rm Err} (TX, t)$ can be decoded by the following algorithm:

{\rm Step 1:}  For any $i \in \Sigma _t (1, \ldots , n)$ apply Buchberger's algorithm and obtain a Groebner basis $G_{i \neg i}$ of
 $I(X, \FF_q) = \langle f_1, \ldots , f_m \rangle _{\FF_q}$ with respect to a lexicographic order of $\FF_q [ x_1, \ldots , x_n]$ with
  $x_i > _{\rm lex} x_{\neg i}$.
Single out the polynomials $G_{\neg i} := G_{i, \neg i} \cap \FF_q [ x_{\neg i} ]$ from $G_{i, \neg i}$, which do not depend on $x_{i_1}, \ldots , x_{i_t}$ for $i = \{ i_1, \ldots , i_t \}$.

{\rm Step 2:}  For arbitrary
\[
i, j \in \Sigma _t (1, \ldots , n)
\]
 compute the determinant
\[
\Delta  _{ji} := \det \frac{\partial f_j}{\partial x_i} \in \FF_q [ x_1, \ldots , x_n],
\]
 the inverse matrix
\[\left( \frac{\partial f_j}{\partial x_i} \right) ^{-1} \in M_{t \times t} ( \Delta _{ji} ^{-1} \FF_q [ x_1, \ldots , x_n])
\]
 and  the  product
 \[\left( \frac{\partial f_j}{\partial x_i} \right)  ^{-1} \frac{\partial f_j}{\partial x}
 \in M_{t \times n} \left( \Delta _{ji} ^{-1} \FF_q [ x_1, \ldots , x_n] \right).
 \]

{\rm Step 3:}  If $w \in {\rm Err} (TX,t)$ and $\pi (w) = a \in X^{( \leq 2t+1)}$ then compute the products
\[
\frac{\partial G_{\neg \gamma}}{\partial x_{\neg \gamma} }(a) w^t \ \ \mbox{  for } \ \ \gamma \in  \Sigma _t (1, \ldots , n)
\]
until you recognize the unique $i \in \Sigma _t (1, \ldots , n)$ with
\[
\frac{\partial G_{\neg i}}{\partial x_{\neg i}} (a) w^t = 0 _{|G_{\neg i}| \times 1}.
\]

{\rm Step 4:}  Plug in $a$ in $\Delta _{ji} \in \FF_q [ x_1, \ldots , x_n]$ and choose some $j \in \Sigma _t (1, \ldots , n)$ with $\Delta _{ji} (a) \neq 0$.
(For any $i \in \Sigma _t(1, \ldots , n)$ and $a \in X$, subject to the aforementioned properties, there exists at least one $j \in \Sigma _t (1, \ldots , n)$ with $\Delta _{ji} (a) \neq 0$.)

{\rm Step 5:}  The unique  $T_a (X, \FF_{q^{\delta (a)}})$-error
\[
e = ( \Pi _{\neg i} (e), \Pi _i (e)) = ( \Pi _{\neg i}(e), 0_{1 \times (n-t)} ) \in \FF_{q^{\delta (a)}} ^n
\]
 of $w \in {\rm Err} (T_a (X, \FF_{q^{\delta (a)}}), t) \subset \FF_{q^{\delta (a)}}^n$ has projection
\[
\Pi _{\neg i} (e) = (e_{i_1}, \ldots , e_{i_t}) = e_i :=
\left[ \left( \frac{\partial f_j}{\partial x_i} \right) ^{-1} \frac{\partial f_j}{\partial x} \right] (a) w^t \in \FF_{q^{\delta (a)}} ^t
\]
onto the components, labeled by $i$.
\end{corollary}

\subsection{Lower semi-continuity of the generic  minimum distance }

If $X \subseteq \overline{\FF _q}^n$ is an affine variety with absolute ideal $I(X, \overline{\FF_q}) = \langle f_1, \ldots , f_m \rangle _{\overline{\FF_q}}$ then $k = \dim X \geq n-m$, as far as any polynomial relation $f_j$ on $X$ decreases the dimension at most by $1$.
If $I(X, \overline{\FF_q})$ admits a generating set with minimal cardinality $m = n-k$, then $X$ is called a complete intersection.

Suppose that the generators $f_j$ of  $I(X, \overline{\FF_q}) = \langle f_1, \ldots , f_{n-k} \rangle _{\overline{\FF _q}}$  are of the form
$f_j = \sum\limits _{\nu \in ( \ZZ ^{\geq 0}) ^n} \alpha _{j, \nu} x^{\nu} \in \FF _q[x]$, where $x^{\nu} := x_1 ^{\nu _1} \ldots x_n ^{\nu _n}$.
The set
\[
S(f_j) := \{ \nu \in ( \ZZ ^{\geq 0} ) ^n \setminus \{ 0^n \} \ \ \vert \ \  \alpha _{j, \nu} \neq 0 \}
 \]
 will be referred to as the non-constant support of $f_j$.
For an arbitrary smooth point $a \in X^{\rm smooth}$ let
\[
F_j (y_j,x) := \sum\limits _{\nu \in S(f_j)} y_{j, \nu} (x^{\nu} - a^{\nu}) \in \FF _{q^{\delta (a)}} [ x, y_j]
\]
and note that $F_j (y_j,a) \equiv 0 \in \FF _{q^{\delta (a)}} [y_j] = \FF _{q^{\delta (a)}} [ y_{j, \nu} \, \vert \, \nu \in S(f_j) ]$.
The maps $S(f_j) \rightarrow \overline{\FF_q}$  are in a bijective correspondence with   the collections
$\gamma _j = \{ \gamma _{j, \nu} \in \overline{\FF _q} \ \ \vert \ \  \nu \in S(f_j) \}$ of  their images.
That is why the set $\overline{\FF _q} ^{S(f_j)}$ of the maps $S(f_j) \rightarrow \overline{\FF _q}$ can be identified with the affine space of dimension $|S(f_j)|$ over $\overline{\FF _q}$.
The product
\[
\cA =  \cA (f,a) := \overline{\FF _q} ^{S(f_1)} \times \ldots \times \overline{\FF _q} ^{S(f_{n-k})}
\]
parameterizes the polynomials $F( \gamma, x) = \{ F_1 ( \gamma _1, x), \ldots , F_m ( \gamma _{n-k},x) \}$ and the affine varieties
$X_{\gamma} := V( F( \gamma ,x)) = V( F_1 ( \gamma _1,x), \ldots , F_{n-k} ( \gamma _{n-k},x)) \subset \overline{\FF _q} ^n$ through $a$.
We are going to show that $\dim X_{\gamma} = \dim X = k$  for "almost all"  $ \gamma \in \cA$ and the minimum distance
 $d(T_a (X_{\gamma}, \FF_{q^{\delta (a)}}) \geq d$   at "almost all" the points of
 $ \{ (\gamma, a) \in \cA \times \overline{\FF_q}^n \ \ \vert \ \  a \in X_{\gamma} \}$.

For an arbitrary polynomial
\[
g(y,x) \in \FF _{q^{\delta (a)}} [y,x] = \FF _{q^{\delta (a)}} [ x_1, \ldots , x_n, y_{j, \nu _j} \ \ \vert \ \ \nu _j \in S(f_j), 1 \leq j \leq n-k],
\]
  let us denote by  $\cW _{\cA} (g(y,a)) := \{ \gamma \in \cA \ \ \vert \ \  g( \gamma , a) =0 \} \subseteq \cA$  the hypersurface,  cut  by the polynomial $g(y,a) \in \FF _{q^{\delta (a)}} [ y_{j, \nu _j} \ \ \vert \ \ \nu _j \in S(f_j), \ \ 1 \leq j \leq n-k]$ for some $a \in \overline{\FF_q}^n$.
Note also that for  any $\gamma \in \cA$ the polynomial $g( \gamma ,x) \in \FF _{q^{\delta (\gamma ,a)}} [ x_1, \ldots , x_n]$ determines a hypersurface
$V(g(\gamma ,x)) := \{ b \in \overline{\FF _q}^n \ \ \vert \ \  g( \gamma, b) =0 \} \subseteq \overline{\FF _q} ^n$.

\begin{proposition}    \label{LowerSemicontinuityMinDist}
Let $X / \FF _q \subset \overline{\FF _q} ^n$ be an   complete intersection, defined over $\FF _q$ and $a \in X^{\rm smooth}$ be a smooth point of $X$, at which the Zariski tangent space $T_a (X, \FF _{q^{\delta (a)}} )$ is of minimum distance $d(T_a (X, \FF _{q^{\delta (a)}} )) \geq d$.
Then there exist non-zero polynomials  $g(y,x), h(y,x) \in \FF _{q^{\delta (a)}} [y,x] \setminus \{ 0 \}$, such that  $\dim X_{\gamma} = \dim X$,
 $a \in X_{\gamma } ^{\rm smooth}$    and
\[
\emptyset \neq  X_{\gamma} \setminus V(g(\gamma, x)) \subseteq X_{\gamma} ^{( \geq d)} \ \ \mbox{   for all } \ \
\gamma \in \cA \setminus \cW _{\cA} (g(y,a)h(y,a))
 \]
\end{proposition}

\begin{proof}

Let $I(X, \overline{\FF_q}) = \langle f_1, \ldots , f_{n-k} \rangle _{\overline{\FF_q}}$ for some $f_1, \ldots , f_{n-k} \in \FF_q[x_1, \ldots , x_n]$
and $F_j ( y_j,x) := \sum\limits _{\nu \in S(f_j)} y_{j, \nu} (x^{\nu} - a^{\nu}) \in \FF_{q^{\delta (a)}} [x, y_j]$ for the non-constant supports
 $S(f_j) \subset ( \ZZ ^{\geq 0} ) ^n \setminus \{ 0^n \}$ of $f_j$ with $1 \leq j \leq n-k$.
By its very definition,  the affine variety $X_{\alpha} := V(F_1 (\alpha, x), \ldots , F_{n-k} (\alpha, x))$ coincides with $X$.
The point $a \in X_{\alpha}$ is smooth exactly when the Jacobian matrix
$\frac{\partial (F_1(\alpha, x), \ldots , F_{n-k} (\alpha,x))}{\partial x}$ is of maximal rank ${\rm rk} \frac{\partial F( \alpha,x)}{\partial x} (a) = n-k$ at $a$.
That implies the existence of $j  \in \Sigma _{n-k} (1, \ldots , n)$ with  $\det \frac{\partial F( \alpha,x)}{\partial x_j}(a) \neq 0$.
The polynomial
\[
h(y,x) := \det \frac{\partial F(y,x)}{\partial x_j} \in \FF _{q^{\delta (a)}} [y,x]
\]
cuts a proper hypersurface $\cW _{\cA} (h(y,a)) \varsubsetneq \cA$, as far as $\alpha \in \cA \setminus \cW _{\cA} (h(y,a))$.
By the very construction, $a \in X_{\gamma}$ for all $\gamma \in \cA$.
Note that an arbitrary  affine variety $X_{\gamma} = V( F_1 (\gamma,x), \ldots , F_{n-k} (\gamma,x))$ is of dimension $\dim X_{\gamma} \geq k$.
 According to  $\dim   T_a (X_{\gamma}, \FF_{q^{\delta (\gamma,a)}}) \geq \dim X_{\gamma}$,
 it suffices to show that $\dim   T_a (X_{\gamma}, \FF _{q^{\delta (\gamma,a)}} ) \leq k$ for
$\forall \gamma \in \cA \setminus \cW _{\cA} (h(y,a))$ in order to conclude that $\dim X_{\gamma} =k$ and $a \in X^{\rm smooth} _{\gamma}$ is a smooth point of $X_{\gamma}$ for all $\gamma \in \cA \setminus \cW _{\cA} (h(y,a))$.
Indeed, $\gamma \in \cA \setminus \cW _{\cA} (h(y,a))$ amounts to
\[
\left( \det \frac{\partial F( \gamma ,x)}{\partial x_j} \right) (a) = h(\gamma ,a) \neq 0.
\]
According to $F_1 (\gamma ,x), \ldots , F_{n-k} (\gamma, x) \in I(X_{\gamma}, \overline{\FF _q})$, the   tangent  code
 $T_a (X_{\gamma}, \FF _{q^{\delta (\gamma,a)}} )$ is contained in the $\FF _{q^{\delta (\gamma,a)}}$-linear code $C_{\gamma,a}$ with parity check matrix $\frac{\partial F(\gamma,x)}{\partial x} (a) \in M_{(n-k) \times n} ( \FF _{q^{\delta (\gamma,a)}})$.
 Bearing in mind that $ n-k \geq {\rm rk} \frac{\partial F ( \gamma,x)}{\partial x} (a)  \geq  {\rm rk} \frac{\partial F( \gamma,x)}{\partial x_j} (a) = n-k$, one concludes that $\dim   C_{\gamma,a} =k$ and
 \[
 \dim   T_a (X_{\gamma}, \FF _{q^{\delta (\gamma,a)}} ) \leq \dim   C_{\gamma,a} =k.
 \]
Thus, $\dim X_{\gamma} = k$ for all $\gamma \in \cA \setminus \cW _{\cA} (h(y,a))$.

It suffices to construct a polynomial $g(y,x) \in \FF _{q^{\delta (a)}} [y,x]$ with $g(\alpha,a) \neq 0$, such that for any $\gamma \in \cA$ the points
 $b \in X_{\gamma}$ with $d( T_b (X_{\gamma}, \FF _{q^{\delta (\gamma,a,b)}} )) <d$ belong to the hypersurface $V( g(\gamma,x)) \subset \overline{\FF _q}^n$.
Then $\cA \setminus \cW _{\cA} (g(y,a) h(y,a) ) \ni \alpha$ is non-empty and $X_{\gamma} \setminus V( g(\gamma,x)) \ni a$ is non-empty for any
$\gamma \in \cA \setminus \cW _{\cA} (g(y,a) h(y,a))$.
Moreover, $d( T_b (X_{\gamma}, \FF _{q^{\delta ( \gamma ,a,b)}})) \geq d$ for any  $\gamma \in \cA \setminus \cW _{\cA} (g(y,a) h(y,a))$  and
  $b \in X_{\gamma} \setminus V( g( \gamma ,x))$.
Towards the construction of $g(y,x) \in \FF _{a^{\delta (a)}} [y,x]$ with the desired properties, note that
\[
X_{\gamma} ^{( <d)} \subseteq  Z_{\gamma, d} := \cup _{i \in \Sigma _{d-1} (1, \ldots ,n)}
V \left( \det \frac{\partial F_{\sigma} ( \gamma _{\sigma}, x)}{\partial x_i} \ \ \Big \vert  \sigma \in \Sigma _{d-1} (1, \ldots , n-k) \right),
\]
 according to  $F_1(\gamma _1,x), \ldots , F_{n-k} (\gamma _{n-k},x) \in I(X_{\gamma}, \overline{\FF_q})$.
 In particular, the point $a$ does not belong to  $ X_{\alpha} ^{( <d)} = Z_{\alpha, d}$ and
  for any $i \in \Sigma _{d-1}(1, \ldots , n)$ there exists $\rho (i) \in \Sigma _{d-1} (1, \ldots , n-k)$ with
 $\det \frac{\partial F_{\rho (i)} ( \alpha _{\rho (i)},x)}{\partial x_i} (a) \neq 0$.
 As a result,
 \[
 a \not \in \cup _{i \in \Sigma _{d-1} (1, \ldots n)} V \left( \det \frac{\partial F _{\rho (i)} ( \alpha _{\rho (i)}, x)}{\partial x_i} \right) =
 V \left( \prod\limits _{i \in \Sigma _{d-1} (1, \ldots n) } \det \frac{\partial F_{\rho (i)} ( \alpha _{\rho (i)},x)}{\partial x_i} \right).
 \]
If
\[
g(y,x) := \prod\limits _{i \in \Sigma _{d-1} (1, \ldots , n)} \det \frac{\partial F_{\rho (i)} ( y_{\rho (i)}, x)}{\partial x_i} \in \FF _{q^{\delta (a)}} [y,x]
\]
then $g( \alpha,a) \neq 0$.
Straightforwardly,
\begin{align*}
X_{\gamma} ^{( <d)} := \{ b \in X_{\gamma} \ \ \vert \ \  d( T_b ( X_{\gamma}, \FF _{q^{\delta (b)}} )) <d \}  \subseteq  \\
Z_{\gamma ,d} := \cup _{i \in \Sigma _{d-1}(1, \ldots , n)} V \left( \det \frac{\partial F _{\sigma} ( \gamma _{\sigma}, x)}{\partial x_i} \ \ \Big \vert \ \ \sigma \in \Sigma _{d-1} (1, \ldots , n-k) \right) \subseteq  \\
\cup _{i \in \Sigma _{d-1} (1, \ldots , n)} V \left( \det \frac{\partial F _{\rho  (i)} ( \gamma _{\rho (i)},x)}{\partial x_i} \right) =
 V \left( g( \gamma , x) \right).
\end{align*}

\end{proof}

\subsection{ Gradient codes }

Let $X \subset \overline{\FF_q}^n$ be an affine variety, whose absolute ideal $I(X, \overline{\FF_q}) = \langle f_1, \ldots , f_m \rangle _{\overline{\FF_q}}$ is generated by $f_1, \ldots , f_m \in \FF _q [ x_1, \ldots , x_n]$.
For an arbitrary  constant  field $\FF _q \subseteq F \subseteq \overline{\FF_q}$, the Zariski tangent bundle of $X$  over $F$ is defined as the disjoint union
\[
T(X,F)  := \coprod\limits _{a \in X(F)} T_a (X,F)
\]
of the Zariski tangent spaces to $X$ over $F$ at the $F$-rational points $a \in X(F) := X \cap F^n$.
Note that the fibres of $T(X,F)$ are not supposed to be of one and a same dimension over $F$.
The union
\[
T(X,F)^{\perp} := \coprod\limits _{a \in X(F)} T_a (X, F)^{\perp}
\]
of the dual codes $T_a (X, F) ^{\perp}$  of $T_a (X, F)$ is referred to as the dual of the Zariski tangent bundle to $X$ over $F$.

As far as the absolute ideal
\[
I(X, \overline{\FF_q}) := \{ g \in \overline{\FF_q} [ x_1, \ldots , x_n] \ \ \vert \ \  g(a) =0, \ \ \forall a \in X \}
\]
of $X$ is generated by polynomials $f_1, \ldots , f_m \in \FF _q [ x_1, \ldots , x_n]$ with coefficients from $\FF_q$, for an arbitrary constant  field
 $\FF _q \subseteq F \subseteq \overline{\FF_q}$ the ideal
\[
I(X, F) := \{ g \in F[x_1, \ldots x_n] \ \ \vert \ \  g(a) =0, \ \ \forall a \in X \}
\]
of $X$ over $F$ is generated by $f_1, \ldots , f_m \in F[ x_1, \ldots , x_n]$, i.e.,  $I(X, F) = \langle f_1, \ldots , f_m \rangle _F$.
 For any   $g \in I(X, F)$   the ordered $n$-tuple  of polynomials
\[
{\rm grad} (g) := \left( \frac{\partial g}{\partial x_1}, \ldots , \frac{\partial g}{\partial x_n} \right) \in F[ x_1, \ldots , x_n] ^n
\]
 is called the gradient of $g$.
We consider the $F$-linear space
\[
{\rm Grad}  I(X, F) := \{ {\rm grad} (g) \ \ \vert \ \  g \in I(X, F) \} \subset F[ x_1, \ldots , x_n] ^n
\]
of the gradients of the polynomials from $I(X, F)$  and   its evaluations
\[
{\rm grad} _a I(X, F) := \{ {\rm grad} (g) (a)   \ \ \vert \ \  g \in I(X, F) \} \subseteq F^n
\]
   at the $F$-rational points  $a \in X(F) := X \cap F^n$ of $X$.
That allows to form  the  vector bundle
\[
{\rm grad} I(X, F) = \coprod\limits _{a \in X(F)} {\rm grad} _a I(X, F) \subset X(F) \times F^n
\]
over $X(F)$, contained in the trivial bundle $X(F) \times F^n$.
The fibres ${\rm grad}_a I(X,F)$ of ${\rm grad} I(X,F)$ are not supposed to be of one and a  same dimension.
Nevertheless, we   say that
\[
{\rm grad} I(X, F) \longrightarrow X(F)
\]
the gradient bundle of $X$ (or of $I(X,F)$) over $F$.

\begin{lemma}    \label{GradientBundleIsDualToTangentBundle}
Let $X \subset \overline{\FF_q}^n$ be an affine variety with $I(X, \overline{\FF_q}) = \langle f_1, \ldots , f_m \rangle _{\overline{\FF_q}}$ for some
 $f_1, \ldots , f_m \in \FF_q [ x_1, \ldots , x_n]$ and $\FF_q \subseteq F \subseteq \overline{\FF_q}$ be a constant field.
Then the dual
\begin{equation}    \label{DualTangentGradient}
T(X, F) ^{\perp} = {\rm grad} I(X,F)
\end{equation}
of the Zariski  tangent bundle $T(X,F)$ of $X$ over $F$ is the gradient bundle of $X$ (or of $I(X,F)$) over $X(F)$.
\end{lemma}

\begin{proof}

The equality (\ref{DualTangentGradient}) is meant as a coincidence $T_a (X,F) ^{\perp} = {\rm grad} _a I(X,F)$ of the fibres over all the points
 $  a \in X(F)$.
By its very definition, $T_a (X,F) ^{\perp}$ is the linear code with a generator matrix
\[
\frac{\partial f}{\partial x} (a) = \left(  \begin{array}{c}
{\rm grad} (f_1)(a)  \\
\ldots   \\
{\rm grad} (f_m) (a)
\end{array}  \right).
\]
Therefore
\[
T_a (X,F) ^{\perp} = \left \{ \sum\limits _{j=1} ^m \lambda _j {\rm grad} (f_j) (a) = {\rm grad} \left( \sum\limits _{j=1} ^m \lambda _j f_j \right) (a) \ \
\Big \vert \ \ \lambda _j \in F \right \}
\]
is a subspace of ${\rm grad} _a I(X,F)$,  as far as $\sum\limits _{j=1} ^m  \lambda _j f_j \in I(X,F)$ for $\forall \lambda _j \in F$.
Conversely, any element of $I(X,F)$ is of the form $g = \sum\limits _{j=1} ^m g_j f_j$ for some $g_j \in F[ x_1, \ldots , x_n]$.
Then ${\rm grad} (g) = \sum\limits _{j=1} ^m f_j {\rm grad} (g_j) + g_j {\rm grad} (f_j)$ and
\begin{align*}
{\rm grad} (g) (a) = \sum\limits _{j=1} ^m g_j (a) {\rm grad} (f_j)  (a) \in {\rm Span} _F ( {\rm grad} (f_1) (a), \ldots , {\rm grad} (f_m)(a) )  \\
 = T_a (X,F)^{\perp}.
 \end{align*}
Thus,  ${\rm grad} _a I(X,F) \subseteq T_a (X,F) ^{\perp}$ and $T_a (X,F) ^{\perp} = {\rm grad} _a I(X,F)$.

\end{proof}

Note that
\[
{\rm Grad} I(X,F) := \{ {\rm grad} (g) \ \ \vert \ \  g \in I(X,F) \} \subset F[x_1, \ldots , x_n] ^n
\]
and
\begin{align*}
\overline{\rm Grad}  I(X,F) := \left \{ \left( \frac{\partial g}{\partial x_1} + I(X,F), \ldots , \frac{\partial g}{\partial x_n} + I(X,F) \right) \ \ \Big \vert \ \
g \in I(X,F) \right \} \subset  \\
 \left[ F[ x_1, \ldots , x_n] / I(X,F) \right] ^n = F[X]^n
\end{align*}
can be viewed as sheaves of sections of ${\rm grad} I(X,F) = T(X,F)^{\perp} \rightarrow X(F)$.
Thus, the gradient codes consist of values of  global sections of vector bundles and appear to be of a similar nature with Goppa codes.
In order to specify, let $Y / \FF _q \subset \PP ^N ( \overline{\FF_q})$ be a smooth irreducible projective curve, defined over $\FF_q$, $D = P_1 + \ldots + P_n$ be a sum of $n$ different $\FF_q$-rational points $P_1, \ldots , P_n \in X( \FF_q)$ and $G$ be a divisor on $Y$ with ${\rm Supp} (G) \cap {\rm Supp} (D) = \emptyset$.
Denote by $\cO _Y ([G]) \rightarrow Y$ the line bundle, associated with $G$ and consider the evaluation map
\[
\cE _D : H^0 (Y, \cO_Y ([G])  \longrightarrow \FF_q ^n,
\]
\[
\cE _D (s) := ( s(P_1), \ldots , s(P_n))
\]
of the global sections $s \in H^0 (Y, \cO_Y ([G]))$ of $\cO_Y ([G])$.
The image $\cE_D H^0 (Y, \cO_Y ([G]))$ of $\cE_D$ is called a Goppa code.
It consists of the values $(s(P_1), \ldots , s(P_n))$ of the global sections $s \in H^0 (Y, \cO_Y ([G]))$ of the line bundle $\cO_Y ([G]) \rightarrow Y$ at  the ordered  $n$-tuple of points $(P_1, \ldots , P_n)$, while ${\rm grad} _a I(X,F)$ is constituted by the  values
\[
{\rm grad} (g) (a) = \left( \frac{\partial g}{\partial x_1}, \ldots , \frac{\partial g}{\partial x_n}  \right) (a) =
 \left( \frac{\partial g}{\partial x_1} + I(X,F), \ldots , \frac{\partial g}{\partial x_n} + I(X,F) \right) (a)
\]
of the global sections
\[
{\rm grad} (g) = \left( \frac{\partial g}{\partial x_1}, \ldots , \frac{\partial g}{\partial x_n} \right) \ \ \mbox{   or  } \ \
\left( \frac{\partial g}{\partial x_1} + I(X,F), \ldots , \frac{\partial g}{\partial x_n} + I(X,F) \right)
\]
 of ${\rm grad} I(X,F) \rightarrow X(F)$.

For an arbitrary integer $0 \leq s \leq n$, let us consider the loci
\begin{align*}
X_{\rm grad} ^{( \geq s)} := \{ a \in X \ \ \vert \ \  d ( {\rm grad} _a I(X, \FF_{q^{\delta (a)}} ) ) \geq s \},  \\
X_{\rm grad} ^{( \leq s)} := \{ a \in X \ \ \vert \ \  d ( {\rm grad} _a I(X, \FF_{q^{\delta (a)}} ) ) \leq s \},
\end{align*}
at which the gradient codes  to $X$ are of minimum distance $\geq s$, respectively, $\leq s$.
The next proposition shows that the presence of a non-zero polynomial from $I(X, \overline{\FF_q})$ in $d$ variables imposes an upper bound $d$ on the minimum distance of a gradient code to $X$ at "almost all" the points of $X$.

\begin{proposition}    \label{DominantPuncturing}
Let $X/ \FF_q \subset \overline{\FF_q} ^n$ be an irreducible affine variety, defined over $\FF_q$, for which there exists a non-zero  polynomial
 $h \in I(X, \overline{\FF_q}) \cap \overline{\FF_q} [ x_{\beta}]$ of $|\beta| = d$ variables and $\Pi _{\neg \beta} : X \rightarrow \overline{\FF_q} ^{d}$ be the puncturing at the complement $\neg \beta$ of $\beta$.
 Then
 \[
 X_{\rm grad} ^{( \geq d+1)} \subseteq \Pi _{\neg \beta} ^{-1} ( \Pi _{\neg \beta} (X) ^{\rm sing} ) \varsubsetneq X,
 \]
  so that   $X_{\rm grad} ^{ (\leq d)}$ is Zariski dense in $X$.
\end{proposition}

\begin{proof}

By assumption, $I(X, \overline{\FF_q}) = \langle f_1, \ldots , f_m \rangle _{\overline{\FF_q}}$  is generated by some polynomials
$f_1, \ldots , f_m \in \FF_q [ x_1, \ldots , x_n]$.
Let $G_{\neg \beta, \beta} \subset \FF_q [ x_1, \ldots , x_n]$ be a Groebner basis of  $\langle f_1, \ldots , f_m \rangle _{\FF_q}$
 with respect to a lexicographic order with $x_{\neg \beta} > _{\rm lex} x_{\beta}$ and
\[
G_{\beta} :=  G_{\neg \beta, \beta} \cap \FF_q [ x_{\beta} ] =  \{ h_1, \ldots , h_l \}.
\]
 As in the proof of Proposition \ref{ErrorSupport}, $h_1, \ldots , h_l \in \FF_q [ x_1, \ldots , x_n]$ generate the absolute ideal
 \[
 I(  \overline{\Pi _{\neg \beta}(X)}, \overline{\FF_q}) = I( \Pi _{\neg \beta}(X), \overline{\FF_q}) =
  I(X, \overline{\FF_q}) \cap \overline{\FF_q}[ x_{\beta} ] =
 \langle G_{\beta} \rangle _{\overline{\FF_q}} = \langle h_1, \ldots , h_l   \rangle _{\overline{\FF_q}}
 \]
  of $\Pi _{\neg \beta} (X)$.
  By the Elimination Theorem 2 from Chapter 2, §1, \cite{CLSh}, $G_{\beta}$ is a Groebner basis of
  $I (X, \overline{\FF_q}) \cap \overline{\FF _q} [ x_{\beta}]  \triangleleft \overline{\FF_q}[x_{\beta}]$, so that the
  presence of a non-zero polynomial  $h \in I (X, \overline{\FF_q}) \cap \overline{\FF _q} [ x_{\beta}]$ implies that the set $G_{\beta} \neq \emptyset$  is non-empty.
  As a result,  $\overline{\Pi _{\neg \beta}(X)} \varsubsetneq \overline{\FF_q}^d$ is an irreducible affine variety of dimension
  $\dim  \Pi _{\neg \beta}(X) <d$.

 For any $h_i \in I(\Pi _{\neg \beta } (X), \FF_{q^{\delta (a)}}) \subseteq I(X, \FF_{q^{\delta (a)}})$  and $a \in X$ note that
 ${\rm grad} (h_i) (a) \in {\rm grad} _a I(X, \FF_{q^{\delta (a)}})$   is a word  of weight $\leq d$, as far as $h_i \in \FF_{q^{\delta (a)}} [ x_{\beta} ]$ depends on at most $|\beta| = d$ variables.
 In particular, for  $a \in X_{\rm grad} ^{( \geq d+1)}$  there follows ${\rm grad} (h_i) ( \Pi _{\neg \beta} (a)) = {\rm grad} (h_i) (a) =0$.
 Thus,
\[
\frac{\partial (h_1, \ldots , h_l)}{\partial x_{\beta} } ( \Pi _{\neg \beta} (a)) =
\left(  \begin{array}{c}
{\rm grad} (h_1)   \\
\ldots   \\
{\rm grad} (h_l)
\end{array}  \right) ( \Pi _{\neg \beta} (a) ) = 0 \ \ \mbox{  at  } \ \ \forall a \in X_{\rm grad} ^{(d+1)}
\]
and
\[
X_{\rm grad} ^{(\geq d+1)} \subseteq V \left( \frac{\partial h_i}{\partial x_j} \ \ \Big \vert \ \   1 \leq i \leq l, \ \   1 \leq j \leq n \right).
\]
We claim that
\[
V \left( \frac{\partial h_i}{\partial x_j} \ \ \Big \vert \ \  1 \leq i \leq l,  \ \  1 \leq j \leq n \right) \subseteq
\Pi _{\neg \beta} ^{-1} ( \Pi _{\neg \beta} (X) ^{\rm sing} ).
\]
Indeed, if $\frac{\partial h_i}{\partial x_j} (a) =0$ for $\forall 1 \leq i \leq l$, $\forall 1 \leq j \leq n$ then
\[
T_{\Pi _{\neg \beta} (a)} ( \Pi _{\neg \beta} (X), \FF_{q^{\delta (a)}}) = \FF_{q^{\delta (a)}}.
\]
According  to $\dim \Pi _{\neg \beta} (X) <d$ there follows $\Pi _{\neg \beta} (a) \in \Pi _{\neg \beta} (X) ^{\rm sing}$,  which is equivalent to
 $a \in \Pi _{\neg \beta} ^{-1} ( \Pi _{\neg \beta} (X) ^{\rm sing} )$.
Thus,
\[
X_{\rm grad} ^{(\geq d+1)} \subseteq V \left( \frac{\partial h_i}{\partial x_j} \ \ \Big \vert \ \   1 \leq i \leq l, \ \   1 \leq j \leq n \right) \subseteq \Pi _{\neg \beta} ^{-1} ( \Pi _{\neg \beta} (X) ^{\rm sing}).
\]
The assumption $\Pi _{\neg \beta} ^{-1} ( \Pi _{\neg \beta} (X) ^{\rm sing}) =X$ leads to $\Pi _{\neg \beta} (X) = \Pi _{\neg \beta} (X) ^{\rm sing}$ and contradicts $\Pi _{\neg \beta} (X)^{\rm smooth} \neq \emptyset$.

Note that $X_{\rm grad} ^{( \geq d+1)} \subseteq \Pi _{\neg \beta} ^{-1} ( \Pi _{\neg \beta} (X)^{\rm sing})$ implies
\[
\Pi _{\neg \beta} ^{-1} ( \Pi _{\neg \beta} (X)^{\rm smooth}) = X \setminus \Pi _{\neg \beta}  ^{-1} ( \Pi _{\neg \beta} (X) ^{\rm sing}) \subseteq
X \setminus X_{\rm grad} ^{( \geq d+1)} = X_{\rm grad} ^{( \leq d)}.
\]
The non-empty subset $\Pi _{\neg \beta} ^{-1} ( \Pi _{\neg \beta} (X) ^{\rm smooth}) \subseteq X$ is Zariski open, as far as the smooth locus
 $\Pi _{\neg \beta} (X) ^{\rm smooth}$ is an open subset of $\Pi _{\neg \beta} (X)$ and $\Pi _{\beta}$ is continuous with respect to the Zariski topology.
 Therefore $\Pi _{\neg \beta} ^{-1} ( \Pi _{\neg \beta} (X)^{\rm smooth})$ is Zariski dense in the irreducible affine variety $X$ and the Zariski closure $\overline{X_{\rm grad} ^{( \leq d)}}$ of $X_{\rm grad} ^{( \leq d)}$ coincides with $X$, according to
 \[
 X = \overline{ \Pi _{\neg \beta} ^{-1} ( \Pi _{\neg \beta} (X)^{\rm smooth})} \subseteq \overline{X_{\rm grad}^{(\leq d)}} \subseteq X.
 \]

\end{proof}

Note that $X_{\rm grad} ^{(d)}$ is not claimed to be Zariski open in $X$.
If there exists a polynomial global tangent frame on $X$ then there is a polynomial family of parity check matrices of the gradient codes to $X$ and
$X_{\rm grad} ^{(d)}$ is Zariski closed in $X$.
As a result, $X = X^{(d)}_{\rm grad}$ and $\Pi _{\neg \beta} (X) \subset \overline{\FF_q}^d$ is smooth.

\section{Tangent codes of special type}

\subsection{  Near MDS tangent and gradient codes     }

According to Proposition \ref{MinimumDistance} (i), if an irreducible affine variety $X / \FF_q \subset \overline{\FF_q}$, defined over $\FF_q$ has at least one MDS tangent code $T_a (X, \FF_{q^{\delta (a)}})$ then the locus
\[
X^{( n-k+1)} := \{ a \in X \ \ \vert \ \ d( T_a (X, \FF_{q^{\delta (a)}})) = n-k+1 \}
\]
 of the MDS tangent codes to $X$ is Zariski open and Zariski dense in $X$.

The present subsection is devoted to the near  MDS-codes,  introduced   by Dodunekov and Landgev in \cite{DL}.
These can be defined as the linear codes $C \subset \FF_q ^n$ of $\dim C =k$ and minimum distance $d(C) = n-k$, whose duals $C^{\perp} \subset \FF_q^n$ are of minimum distance $d(C^{\perp}) = k$.
If $X / \FF_q \subset \overline{\FF_q} ^n$ is an irreducible  $k$-dimensional  affine variety, defined over $\FF_q$ and $X^{\rm NMDS}$ is the set of the points  $a \in X$, at which the tangent code $T_a (X, \FF_{q^{\delta (a)}})$ is near MDS then  $X^{\rm NMDS} \subseteq X^{(n-k)} $.
We show that $X^{\rm NMDS}$ is a Zariski open subset of $X$.
If $X^{\rm NMDS} \neq \emptyset$ is non-empty, then $X^{\rm NMDS}$ is Zariski dense in $X^{(n-k)}$ and in $X$.

In the statement and the proof of the  next proposition we abbreviate
  $\Sigma _s (n) := \Sigma _s (1, \ldots , n)$, respectively,   $\Sigma _r (m) := \Sigma _r (1, \ldots , m)$, in order to simplify the notation.

\begin{proposition}    \label{NearMDS}
Let $X \subset \overline{\FF_q} ^n$ be an irreducible $k$-dimensional affine variety with
$I(X, \overline{\FF_q}) = \langle f_1, \ldots , f_m \rangle _{\overline{\FF_q}}$ for some $f_1, \ldots , f_m \in \FF_q [ x_1, \ldots , x_n]$ and
 $\Pi _{\alpha} : X \rightarrow \Pi _{\alpha} (X)$ be a non-finite puncturing at $|\alpha| = n-k$ coordinates.
Then the subset $X^{\rm NMDS} \subseteq X$ of the points $a \in X$ at which $T_a (X, \FF_{q^{\delta (a)}})$ is a near MDS code is Zariski open,
\begin{align*}
X^{\rm NMDS} =  \\
 X^{(n-k)} \setminus
V \left( \prod\limits _{\beta \in \Sigma _{n-k+1} ( n)} \det \frac{\partial f _{\psi ( \beta)}}{\partial x_{\theta ( \beta)}} \, \Big \vert \,
\begin{array}{c}
  \psi : \Sigma _{n-k+1} ( n) \rightarrow \Sigma _{n-k} ( m),   \\
   \theta : \Sigma _{n-k+1} ( n) \rightarrow \Sigma _{n-k} ( n)  \\
  \theta ( \beta) \in \Sigma _{n-k} ( \beta)
  \end{array}  \right) =
  \end{align*}
  \begin{align*}
X \setminus V \left( \prod\limits _{i \in \Sigma _{n-k-1} ( n)} \det \frac{\partial f_{\varphi (i)}}{\partial x_i} \prod\limits _{\beta \in \Sigma _{n-k+1}( n)} \det \frac{\partial f_{\psi ( \beta)}}{\partial x_{\theta( \beta)}} \, \Big \vert \,
 \begin{array}{c}
   \varphi : \Sigma _{n-k-1} ( n) \rightarrow \Sigma _{n-k-1} ( m), \\
   \psi : \Sigma _{n-k+1}(  n)  \rightarrow \Sigma _{n-k} (  m),  \\
  \theta : \Sigma _{n-k+1} (  n) \rightarrow \Sigma _{n-k} (  n),  \\
 \theta (\beta) \in \Sigma _{n-k} ( \beta)
 \end{array}  \right)
\end{align*}
both in $X^{(n-k)}$ and in $X$.
Thus, if some Zariski tangent space $T_a (X, \FF _{q^{\delta (a)}})$ at a smooth point $a \in X$ is a near MDS-code then $X^{\rm NMDS}$ is Zariski dense in $X^{(n-k)}$, $X$ and $I(X, \overline{\FF_q}) \cap \overline{\FF_q} [ x_{\gamma}] = \{ 0 \}$ for $\forall \gamma \in \Sigma _{k-1} (n)$.
\end{proposition}

\begin{proof}

By Lemma 3.1 from \cite{DL}, an $[n,k,n-k]_q$-code $T_a (X, \FF_{q^{\delta (a)}})$ with a parity check matrix $\frac{\partial f}{\partial x} (a)$ is near MDS exactly when for $\forall \beta \in \Sigma _{n-k+1}(n)$ the matrix $\frac{\partial f}{\partial x_{\beta}} (a) \in M_{m \times (n-k+1)}( \FF_{q^{\delta (a)}})$ is of maximal rank ${\rm rk} \frac{\partial f}{\partial x_{\beta}} (a) = n-k$.
Thus,
\begin{align*}
X^{(n-k)} \setminus X^{\rm NMDS} = X^{(n-k)} \cap \left[ \cup _{\beta \in \Sigma _{n-k+1} (n)} \left  \{ a \in X \ \ \Big \vert \ \
 {\rm rk} \frac{\partial f}{\partial x_{\beta}} (a) < n-k \right \} \right].
\end{align*}
The minors of $\frac{\partial f}{\partial x_{\beta}} (a)$ of order $n-k$ are labeled by $\lambda \in \Sigma _{n-k} (m)$ and $\mu \in \Sigma _{n-k} ( \beta)$, so that
\begin{align*}
X^{(n-k)} \setminus X^{\rm NMDS} =  \\
 X^{(n-k)} \cap \left[ \cup _{\beta \in \Sigma _{n-k+1} (n)} V \left( \det \frac{\partial f_{\lambda}}{\partial x_{\mu}} \ \ \Big \vert \ \
   \lambda \in \Sigma _{n-k} (m),  \mu \in \Sigma _{n-k} ( \beta) \right) \right].
\end{align*}
Bearing in mind that $\cup _{\nu \in N} V( S_{\nu}) = V \left( \prod\limits _{\nu \in N} S_{\nu} \right)$ for an arbitrary finite set $N$,
\[
\prod\limits _{\nu \in N} S_{\nu}  := \left \{ \prod\limits _{\nu \in N} g_{\nu} \ \ \Big  \vert  \ \  g_{\nu} \in S_{\nu} \right \}
\]
and arbitrary subsets  $S_{\nu} \subset \overline{\FF_q} [ x_1, \ldots , x_n]$, one concludes that
\begin{align*}
X^{(n-k)} \setminus X^{\rm NMDS} =  \\
X^{(n-k)} \cap V \left( \prod\limits _{\beta \in \Sigma _{n-k+1} (n)} \det \frac{\partial f_{\psi ( \beta)}}{\partial x_{\theta ( \beta)}} \ \ \Big \vert \ \
\begin{array}{c}
  \psi : \Sigma _{n-k+1} (n) \rightarrow \Sigma _{n-k} (m),  \\
  \theta : \Sigma _{n-k+1} (n) \rightarrow \Sigma _{n-k} (n),  \\
\theta ( \beta) \in \Sigma _{n-k} ( \beta)
\end{array} \right)
\end{align*}
Therefore, the subset $X^{\rm NMDS}$ of $X^{(n-k)}$ equals
\[
X^{\rm NMDS} = X^{(n-k)} \setminus  (X^{(n-k)} \setminus X^{\rm NMDS}) =  X^{(n-k)} \setminus V(A)
\]
for
 \begin{align*}
 A := \left \{  \prod\limits _{\beta \in \Sigma _{n-k+1} (n)} \det \frac{\partial f_{\psi ( \beta)}}{\partial x_{\theta ( \beta)}} \ \ \Big \vert \ \
\begin{array}{c}
  \psi : \Sigma _{n-k+1} (n) \rightarrow \Sigma _{n-k} (m),  \\
  \theta : \Sigma _{n-k+1} (n) \rightarrow \Sigma _{n-k} (n),  \\
\theta ( \beta) \in \Sigma _{n-k} ( \beta)
\end{array}  \right \}.
\end{align*}
If
\begin{align*}
B := \left \{ \prod\limits _{i \in \Sigma _{n-k-1}(n)} \det \frac{\partial f_{\varphi (i)}}{\partial x_i} \ \ \Big \vert \ \
  \varphi : \Sigma _{n-k-1} (n) \rightarrow \Sigma _{n-k-1}(m) \right \}
\end{align*}
the by Proposition \ref{MinimumDistance} (ii)  and (\ref{LocusMinDistAtLeastD}) there follows
\[
X^{(n-k)} = X^{( \geq n-k)} = X \setminus X^{( \leq n-k-1)} = X \setminus V(B).
\]
Thus,
\[
X^{\rm NMDS} = [ X \setminus V(B)] \setminus V(A) = X \setminus [ V(B) \cup V(A)] = X \setminus V(BA)
\]
with
\[
BA = \left \{ \prod\limits _{i \in \Sigma _{n-k-1} ( n)} \det \frac{\partial f_{\varphi (i)}}{\partial x_i}
 \prod\limits _{\beta \in \Sigma _{n-k+1}( n)} \det \frac{\partial f_{\psi ( \beta)}}{\partial x_{\theta( \beta)}}    \Big \vert
 \begin{array}{c}
  \varphi : \Sigma _{n-k-1} ( n) \rightarrow \Sigma _{n-k-1} ( m), \\
    \psi : \Sigma _{n-k+1}(  n)  \rightarrow \Sigma _{n-k} (  m),  \\
  \theta : \Sigma _{n-k+1} (  n) \rightarrow \Sigma _{n-k} (  n),  \\
 \theta (\beta) \in \Sigma _{n-k} ( \beta)
 \end{array}  \right \}.
\]

By Lemma 3.2 from \cite{DL}, a linear code $C$ is near MDS if and only if its dual $C^{\perp}$ is near MDS.
Thus, at any $a \in X^{\rm NMDS}$  the gradient code ${\rm grad} _a I(X, \FF_{q^{\delta (a)}})$ is near MDS and, in particular,
 $d( {\rm grad} _a I(X, \FF _{q^{\delta (a)}} )) =k$.
 In such a way, $X^{\rm NMDS} \subseteq X_{\rm grad} ^{( \geq k)}$ and $X_{\rm grad} ^{(\geq k)}$ is Zariski dense in $X$.
 By Proposition \ref{DominantPuncturing}, the presence of a non-zero polynomial $h \in I(X, \overline{\FF_q}) \cap \overline{\FF_q} [ x_{\gamma}]$ requires the Zariski closure $\overline{X_{\rm grad}} ^{(\geq k)} \subseteq \Pi _{\neg \gamma} ^{-1} ( \Pi _{\neg \gamma} (X) ^{\rm sing}) \varsubsetneq X$ to be a proper affine subvariety of $X$.
 The contradiction justifies  the vanishing of the intersections    $I(X, \overline{\FF_q}) \cap \overline{\FF_q} [ x_{\gamma}] = \{ 0 \}$ for
  $\forall \gamma \in \Sigma _{k-1} (n)$.

\end{proof}

By Theorem 3.5 from \cite{DL}, the existence of an $[n,k,n-k] _q$-code with $k \geq 2$ requires $n-k \leq 2q$.
Therefore $X^{(n-k)} \subseteq \cup _{q^l \geq \frac{n-k}{2}}  X( \FF_{q^l})$ in the case of $\dim X = k \geq 2$.
Theorem 3.4 from \cite{DL} specifies  that for $q < n-k$ any $[n,k,n-k]_q$-code is near MDS.
Thus,  $X^{(n-k)} \setminus X^{\rm NMDS} \subseteq \cup _{q^l \geq n-k} X^{(n-k)} ( \FF_{q^l})$.
Note that
\begin{align*}
Y = \cup _{q^l < n-k} X^{(n-k)} ( \FF_{q^l}) = X^{(n-k)} \cap \left[ \cup _{q^l < n-k} V ( x_i ^{q^l} - x_i \ \ \vert \ \ 1 \leq i \leq n ) \right] = \\
X^{(n-k)} \setminus \cup _{q^l \geq n-k} X^{(n-k)} ( \FF_{q^l})
\end{align*}
is a finite explicitly given affine subvariety of $X^{(n-k)}$, contained in $X^{\rm NMDS}$.
Therefore $\dim Y =0$, while $\dim X^{\rm NMDS} = k$ and $X^{\rm NMDS}$ is "considerably larger" than $Y$.

\subsection{   Cyclic tangent  codes    }

Let us recall that a linear code $C \subset \FF_q ^n$ is cyclic if invariant under the cyclic shift of components
\[
\zeta : \FF_q ^n \longrightarrow \FF_q ^n,
\]
\[
\zeta (x_1, x_2, \ldots ,  x_n) = (x_2, \ldots , x_n, x_1).
\]
The cyclic codes $C$  are in a bijective correspondence  with the  principal ideals $\langle \overline{g} \rangle _{\FF_q} \triangleleft \FF_q [t] / \langle t^n -1 \rangle$ in the quotient ring of $\FF_q [t]$ by the ideal $\langle t^n -1 \rangle \triangleleft \FF_q [t]$ with generator  $t^n -1 \in \FF_q [t]$.
The polynomial $g(t) \in \FF_q [t]$ is a divisor of $t^n-1$ of degree $\deg g = n - \dim _{\FF_q} C$ with leading coefficient $1$.

In order to construct an affine variety,  at which  all Zariski tangent spaces extend to cyclic codes over sufficiently large extensions of the definition   fields, we need the following

 \begin{lemma}   \label{ConstantLocus}
 Let $X \subset \overline{\FF_q} ^n$ be an irreducible $k$-dimensional affine variety, whose absolute ideal
 $I(X, \overline{\FF_q}) = \langle f_1, \ldots , f_m \rangle _{\overline{\FF_q}}$  is generated by  some $f_1, \ldots , f_m \in \FF_q [ x_1, \ldots , x_n]$ and  \[
 X_{\leq r} := \left  \{ a \in X \ \ \Big  \vert \ \ {\rm rk} \frac{\partial f}{\partial x} (a) \leq r \right \} \ \ \mbox{  for  } \ \ 0 \leq r \leq n-k
 \]
be the loci at which the Zariski tangent spaces are of dimension $\geq n-r$.
 For an arbitrary $\FF_{q^s}$-linear code $C \subset \FF_{q^s} ^n$  consider the locus
 \begin{align*}
 X(C) :=  \\
  \{ a \in X \,  \vert \,   T_a (X, \FF_{q^{\delta (a)}}) \otimes _{\FF_{q^{\delta (a)}}} \FF_{q^{l(a)}} = C \otimes _{\FF_{q^s}} \otimes \FF_{q^{l(a)}} \,  \mbox{  for  } \,  l(a) = LCM ( \delta (a),s) \in \NN \}
 \end{align*}
  of $X$ at which  the tangent codes  coincide with $C$ after an  appropriate extension  of the   constant   field.

 (i) The  loci
 \[
 X_{\leq r} = V \left( \det \frac{\partial f_{\lambda}}{\partial x_{\mu}} \ \ \Big \vert \ \
   \lambda \in \Sigma _{r+1} (1, \ldots m), \ \   \mu \in \Sigma _{r+1} (1, \ldots , n) \right) \cap X
 \]
 are Zariski closed subsets of $X$.

 In particular, $X_{\leq n-k-1} = X^{\rm sing}$ is the set of the  singular  points of $X$.

 (ii)   If $C \subset \FF_{q^s} ^n$ has parity check matrix   $H \in M_{r \times n} ( \FF_{q^s})$ of rank $0 \leq r \leq n-k$  and $H_{\lambda}$ is the matrix, formed by the columns of $H$, labeled by $\lambda \subseteq \{ 1, \ldots , n \}$, then   $X(C) = Z(C) \cap X_{ >r-1}$ is the  intersection of the Zariski closed subset
 \[
 Z(C) := V \left( \det \left(  \begin{array}{c}
 \frac{\partial f_j}{\partial x_{\lambda}}   \\
 H_{\lambda}
 \end{array}  \right) \ \  \Big \vert \ \   1 \leq j \leq m, \ \  \lambda \in \Sigma _{r+1} (1, \ldots , n) \right) \cap X
 \]
 of $X$ with the Zariski open subset $X_{>r-1} := X \setminus X_{\leq r-1}$.
 In particular, if $\dim  Z(C) \geq 1$ then $X(C)$ is an infinite set.

 (iii) For an $\FF_q$-linear code  $C \subset \FF_q^n$ of $\dim _{\FF_q} C = k$ with  parity check matrix $H \in M_{(n-k) \times n} ( \FF_q)$    and  for arbitrary polynomials $g_1, \ldots , g_{n-k} \in \FF_q [ x_1, \ldots , x_n]$, the affine variety  $X_C := V( f_1, \ldots , f_{n-k})$ cut by
 \[
 f_i (x_1, \ldots , x_n) = \sum\limits _{j=1} ^n H_{ij} x_j + g_i ( x_1 ^p, \ldots , x_n ^p) \ \ \mbox{    with   } \ \ \forall 1 \leq i \leq n-k
 \]
 is smooth and has constant tangent codes  $T_a (X, \FF_{q^{\delta (a)}}) = C \otimes _{\FF_q} \FF_{q^{\delta (a)}}$ for $\forall a \in X$, i.e., $X_C = X(C)$.
 \end{lemma}

 \begin{proof}

(i) The condition ${\rm rk} \frac{\partial f}{\partial x} (a) \leq r$ is equivalent to the vanishing of all the minors of $\frac{\partial f}{\partial x} (a)$ of order $r+1$.

(ii) Note that $a \in X \setminus X_{\leq r-1}$ exactly when ${\rm rk} \frac{\partial f}{\partial x} (a) \geq r$, which is tantamount to
\[
\dim T_a (X, \FF_{q^{\delta (a)}}) \otimes _{\FF_{q^{\delta (a)}}} \FF_{q^{l(a)}} \leq n-r = \dim _{\FF_{q^{l(a)}}} C \otimes _{\FF_{q^s}} \FF_{q^{l(a)}}
\]
 for $l(a) = LCM ( \delta (a),s) \in \NN$.
 Thus, a point $a \in X \setminus X_{\leq r-1}$ belongs to $X(C)$ if and only if $T_a (X, \FF_{q^{\delta (a)}}) \otimes _{\FF_{q^{\delta (a)}}} \FF_{q^{l(a)}} \supseteq C \otimes _{\FF_{q^s}} \FF_{q^{l(a)}}$.
This is equivalent to the opposite inclusion of the corresponding dual codes,
 \[
 {\rm Span}  _{\FF_{q^{l(a)}}} \left(    \frac{\partial f_j}{\partial x} (a) \ \ \Big \vert \ \ 1 \leq j \leq m \right)  \subseteq
  {\rm Row Span} _{\FF_{q^{l(a)}}} H.
  \]
 Note that $\frac{\partial f_j}{\partial x} (a) \in {\rm Row Span} _{\FF_{q^{l(a)}}} H$ if and only if
 \[
 {\rm rk} \left( \begin{array}{c}
 \frac{\partial f_j}{\partial x}(a)  \\
 H
 \end{array}  \right) = {\rm rk} H = r
 \]
  and this is equivalent to the vanishing of all the minors of
\[
\left(  \begin{array}{c}
 \frac{\partial f_j}{\partial x} (a)  \\
H
\end{array}  \right) \in M_{(r+1) \times n} ( \FF_{q^{l(a)}})
\]
of order $r+1$.

(iii)  Note that the Jacobian matrix $\frac{\partial ( f_1, \ldots , f_{n-k})}{\partial (x_1, \ldots , x_n)} \equiv H$ is constant and
\[
T_a (X_C, \FF_{q^{\delta (a)}}) \subseteq C \otimes _{\FF_q} \FF_{q^{\delta (a)}},
\]
 according to $f_1, \ldots , f_{n-k} \in I(X_C, \overline{\FF_q})$.
Making use of
\[
k \leq \dim X_C \leq \dim T_a (X_C, \FF_{q^{\delta (a)}}) \leq \dim _{\FF_q} C = k,
\]
one concludes that $\dim X_C = k$, $T_a (X_C, \FF_{q^{\delta (a)}}) = C \otimes _{\FF_q} \FF_{q^{\delta (a)}}$ and $X_C$ consists of smooth points.

 \end{proof}

With some abuse of notation, if $X = X(C)$ for some linear code $C \subset \FF_q^n$, we say that $X$ has constant Zariski tangent bundles over all finite fields
 $\FF_{q^s} \subset \overline{\FF_q}$.

Denote     $p = {\rm char} \FF_q$   and     note that $t^n -1 \in \FF_p [t]$ has coefficients from the prime field $\FF_p$ of $\FF_q$.
For   any prime integer $p$ and any natural number  $n$, let $\FF_{\sigma}$, $\sigma = \sigma (p,n)$ be the splitting field of
 $t^n-1 \in \FF_p [t]$ over $\FF_p$.
We assume that $p$ and $n$ are relatively prime, so that $t^n-1$ has no multiple roots in the algebraic closure
 $\overline{\FF_p} = \cup _{s=1} ^{\infty} \FF_{p^s}$ of $\FF_p$.
For an arbitrary divisor $g(t) \in \FF_{\sigma}  [t]$ of $t^n-1$ with leading coefficient $1$ and an arbitrary $s \in \NN$, let us denote by
\[
\langle \overline{g} \rangle _{\FF_{\sigma^s}} :=
 \langle g(t) + \langle t^n-1 \rangle _{\FF_{\sigma^s}} \rangle \triangleleft \FF_{\sigma^s} [t] / \langle t^n-1 \rangle _{\FF_{\sigma^s}}
\]
the cyclic code over $\FF_{\sigma^s}$ with generator polynomial $g(t)$.
We are interested in the cyclic codes $\langle \overline{g} \rangle _{\FF_{\sigma^s}} \subset \FF_{\sigma^s} ^n$ over $\FF_{\sigma^s}$ of length $n$ and dimension
 \[
 k \leq \dim _{\FF_{\sigma ^s}} \langle \overline{g} \rangle _{\FF_{\sigma ^s}} = n - \deg g \leq n.
 \]
 These correspond to $g(t) \in \FF_\sigma  [t]$ of degree $0 \leq \deg g \leq n-k$.

 Denote by $\cD _{p,n}$ the set of the divisors $g(t) \in \overline{\FF_p} [t]$ of $t^n-1 \in \FF_p [t]$ with leading coefficient $1$ and put
 $\cD _{p,n} ( \nu) := \{ g \in \cD _{p,n} \ \ \vert \ \  \deg g = \nu \}$.
 Note that $\cD _{p,n}$ and therefore $\cD_{p,n} ( \nu)$ are finite sets.
 We say that all cyclic codes of length $n$ and dimension $k$  over $\overline{\FF_p}$ are realized as finite Zariski tangent spaces to an affine variety
  $X \subset \overline{\FF_p} ^n$ of $\dim X =k$, if for any $g \in \cD _{p,n} (n-k)$ there exists $a \in X$ with
 \[
 T_a (X, \FF_{q^{\delta (a)}}) \otimes _{\FF_{q^{\delta (a)}}} \FF_{q^{\delta (a,g)}} =
  \langle \overline{g} \rangle _{\FF_{q^{\delta (g)}}} \otimes _{\FF_{q^{\delta (g)}}} \FF_{q^{\delta (a,g)}},
 \]
 where $\FF_{q^{\delta (g)}}$ is the common definition field of all the coefficients of $g(t) \in \overline{\FF_p}[t]$ and $\FF_{q^{\delta (a, g)}}$ is the common definition field of $a$ and the coefficients of $g(t)$.
 All  cyclic codes of length $n$  over $\overline{\FF_p}$ are realized as finite Zariski tangent spaces to an affine variety
  $X \subset \overline{\FF_p} ^n$ if  this holds for all cyclic codes of length $n$ and an arbitrary  dimension $0 \leq k \leq n$ over $\overline{\FF_p}$.

  The next corollary  constructs an affine variety, whose all finite Zariski tangent spaces are cyclic codes and an affine variety with an arbitrary number of non-cyclic tangent codes.

\begin{corollary}    \label{CyclicTangentCodes}
(i) Let $X / \FF_q \subset \overline{\FF_q} ^n$ be an irreducible $k$-dimensional affine variety, defined over $\FF_q$,
\[
XT^{\zeta} := \{ a \in X \ \ \vert \ \ \zeta T_a (X, \FF_{q^{\delta (a)}}) = T_a (X, \FF_{q^{\delta (a)}}) \}
\]
be the subset of $X$, at which the finite Zariski tangent spaces are cyclic codes and
$X^{\rm smooth} = X_{\geq n-k} = X_{n-k} := \{ a \in X \ \ \vert \ \ {\rm rk} \frac{\partial f}{\partial x} (a) = n-k \}$ be the smooth locus of $X$.
Then $XT^{\zeta} \cap X^{\rm smooth}$ is a Zariski closed subset of $X^{\rm smooth}$.

(ii) For an arbitrary natural number $n$ and  an arbitrary prime integer $p$, relatively prime to $n$,  there exists a smooth affine variety
\[
X^{\zeta}_{p,n} := \coprod\limits _{g \in \cD_{p,n}} X_{\langle \overline{g} \rangle _{\FF_q}} \subset \overline{\FF_p} ^n,
\]
such that all cyclic codes of length $n$ over $\overline{\FF_p}$ are realized as Zariski tangent spaces to $X^{\zeta}_{p,n}$ and all    finite Zariski tangent spaces  to $ X^{\zeta}_{p,n}$ are cyclic codes.

(iii) For arbitrary $k,n, M \in \NN$ with  $k < n$  and an arbitrary prime integer $p$ with   $GCD(p,n) =1$, there is an affine variety $X_{p,n,k} (M) \subset \overline{\FF_p} ^n$ of $\dim X_{p,n.k} (M) =k$, such that all cyclic codes of length $n$ and dimension $k$ over $\overline{\FF_p}$ are realized as finite Zariski tangent spaces to $X_{p,n,k} (M)$ and there are at least $M$ non-cyclic finite Zariski tangent spaces to $X_{p,n,k} (M)$.
\end{corollary}

\begin{proof}

(i)   For an arbitrary $ 0 \leq r \leq n-k$ let us consider the subset
 \[
XT^{\zeta} _r := \{ a \in XT^{\zeta} \ \ \vert \ \ \dim  T_a (X, \FF_{q^{\delta (a)}}) = n-r \}
\]
of $XT^{\zeta}$.
It suffices to show that
\[
XT^{\zeta} _r = Z^{\zeta} _r \cap X_{\geq r}
\]
 is  the  intersection of a Zariski closed subset $ Z^{\zeta} _r \subseteq X$
   with the Zariski open subset $X_{>r-1} := X \setminus X_{\leq r-1}$ of $X$.
 In the notations of Lemma \ref{ConstantLocus} (ii), one has
\begin{align*}
XT^{\zeta} _r = \cup _{g \in \cD _{p,n} (n-r)}  \, X( \langle \overline{g} \rangle _{\FF_q} ) =  \\
\cup _{g \in \cD _{p,n} (n-r)} \,  [ Z( \langle \overline{g} \rangle _{\FF_q}) \cap X_{>r-1}] =   \\
\left[ \cup _{g \in \cD_{p,n} (n-r)} Z( \langle \overline{g} \rangle _{\FF_q} ) \right] \cap X_{ >r-1}=  \\
Z^{\zeta} _r \cap X_{\geq r}
\end{align*}
for the Zariski closed subset
\[
Z^{\zeta} _r := \cup _{g \in \cD _{p,n} (n-r)} Z( \langle \overline{g} \rangle _{\FF_q} ) \subseteq X.
\]
In the case of $r = n-k$, one concludes that
\[
XT^{\zeta} _{n-k}  = Z^{\zeta} _{n-k} \cap X_{\geq n-k} = Z^{\zeta} _{n-k} \cap X_{n-k} = XT^{\zeta} \cap X^{\rm smooth}
 \]
 is a Zariski closed subset of $X^{\rm smooth}$.

(ii) In the notations from Lemma \ref{ConstantLocus} (iii), $X_{\langle \overline{g} \rangle _{\FF_q}} $ are smooth affine varieties  whose all finite Zariski tangent spaces
\[
T_a (X _{\langle \overline{g} \rangle _{\FF_q}}, \FF_{q^{\delta (a)}})  = \langle \overline{g} \rangle _{\FF_q} \otimes _{\FF_q} \FF_{q^{\delta (a)}} =
\langle \overline{g} \rangle _{\FF_q^{\delta (a)}} \subset \FF_{q^{\delta (a)}}^n
\]
are cyclic codes.

(iii)  For an arbitrary natural number $n$ and an arbitrary prime integer $p$, which is relatively prime to $n$, there are
$|\cD _{p,n} (n-k)| = \binom{n}{n-k} = \binom{n}{k}$ generator polynomials $g(t)$ of cyclic codes of dimension $k$ over $\overline{\FF_p}$.
These $g(t)$ are uniquely determined by their $n-k$ roots, contained in the set of the  $n$ distinct roots of $t^n-1$ in $\overline{\FF_p}$.
For any $M \in \NN$ there exists $s \in \NN$, such that $p^s - \binom{n}{k} \geq M$ and $\FF_{p^s}$ contains the splitting fields of all
 $g \in \cD _{p,n} (n-k)$.
Consider a family $\cC \rightarrow \FF_{p^s}^n$ of $\FF_{p^s}$-linear codes $\cC(a) \subset \FF _{p^s} ^n$ of length $n$ and dimension $k$, which contains all cyclic codes $\langle \overline{g} \rangle _{\FF_{p^s}}$ with generator polynomials $g \in \cD _{p,n} (n-k)$ as fibres over some points of $\FF _{p^s} ^n$, and has at least $M$ non-cyclic fibres $\cC (b)$, $b \in \FF_{p^s} ^n$.
By Proposition \ref{DestabilizationMinDist}, there exist irreducible $k$-dimensional affine  varieties
$Y_1 / \FF_{p^s} ,    \ldots , Y_m / \FF_{p^s}   \subset \overline{\FF_p} ^n$, defined over $\FF_{p^s}$, such that
\begin{align*}
\FF_{p^s} \subseteq Y_1^{\rm smooth} ( \FF_{p^s})  \cup \ldots \cup Y_m ^{\rm smooth} ( \FF_{p^s}) \ \ \mbox{  and } \\
 T_a( Y_i, \FF_{p^s}) = \cC(a) \ \ \mbox{  for all } \ \ a \in \FF_{p^s}  \ \ \mbox{ and all } \ \ Y_i \ni a.
 \end{align*}
 The affine variety $ X_{p,n,k} (M) := Y_1 \cup \ldots\cup Y_m $  of dimension $k$ realizes all cyclic codes $\langle \overline{g} \rangle _{\FF_{p^s}}$ of length $n$ and dimension $k$  as Zariski tangent spaces and has at least $M$ non-cyclic Zariski tangent  codes over $\FF_{p^s}$.

\end{proof}

\subsection{   Hamming tangent codes    }

Let us fix a finite field $\FF_q$, a natural number $r \in \NN$ and denote $n = | \PP ^{r-1} ( \FF_q)| = \frac{q^r-1}{q-1}$.
Choose a complete set of liftings $H_1, \ldots , H_n \in M_{r \times n} ( \FF_q)$ of the points of the projective space
 $\PP ^{r-1} ( \FF_q) = \PP ( \FF_q ^r) \simeq \FF _q ^r \setminus \{ 0^r \} / \FF_q ^*$ or a complete set of $\FF_q^*$-orbit representatives on $\FF_q^r$.
The $\FF_q$-linear code $\cH _{q,r} \subset \FF_q^n$ with parity check matrix $H_{q,r} := (H_1 \ldots H_n)$ has minimum distance $3$ and is known as the Hamming code.
The present subsection constructs an embedding of $\overline{\FF_q} ^{n-r}$ in $\overline{\FF_q} ^n$, whose image has $q (q-1) ^{n-r-1}$ Hamming tangent codes in the case of an odd characteristic   ${\rm char} \FF_q$ and $q^{n-r}$ Hamming tangent codes  for ${\rm char} \FF_q =2$.

\begin{proposition}     \label{HammingTangentCodes}
Let $H_{q,r} = (H_1 \ldots H_n) \in M_{r \times n} ( \FF_q)$ with  $n = \frac{q^r-1}{a-1}$ be a parity check matrix of a Hamming code with
${\rm rk} (H_1 \ldots H_r) =r$ and $H_{r-1} + H_r = H_{r+1}$.
Then the puncturing $\Pi _{\rho} : \overline{\FF_q}^n \rightarrow \overline{\FF_q} ^{n-r}$ at $\rho = \{ 1, \ldots , r \}$ restricts to a biregular morphism
 $\Pi _{\rho} : X = V( f_1, \ldots , f_r) \rightarrow \overline{\FF_q} ^{n-r}$ of the affine variety $X$, cut by the polynomials
\[
f_i (x_1, \ldots , x_n) = \sum\limits _{j=1} ^n H_{ij} x_j + \sum\limits _{j = r+2} ^n H_{ij} x_j ^2 \ \ \mbox{   for  } \ \ 1 \leq i \leq r.
\]
If $p = {\rm char} \FF_q \geq 3$ then $T_a (X, \FF_q) \subset \FF_q ^n$ are Hamming codes for all
\[
a \in X( \FF_q) \setminus V \left( \prod\limits _{j=r+2} ^n (2x_j+1) \right)
\]
 and
  \[
 \left| X( \FF_q) \setminus  V \left( \prod\limits _{j=r+2} ^n (2x_j+1) \right) \right| = q(q-1) ^{n-r-1}.
 \]

In the case of ${\rm char} \FF_q =2$, the Zariski tangent spaces $T_a (X, \FF_q) \subset \FF_q ^n$ are Hamming codes for all the $\FF_q$-rational points
 $a \in X( \FF_q)$ and $|X( \FF_q)| = q^{n-r}$.
\end{proposition}

\begin{proof}

Let us denote
\[
A  := (H_1 \ldots H_r) ^{-1} \in M_{r \times r} ( \FF_q),  \ \ H'' := (H_{r+2} \ldots H_n) \in M_{r \times (n-r-1)} ( \FF_q)
\]
and observe that
\begin{align*}
\left(  \begin{array}{c}
f_1  \\
\ldots  \\
f_r
\end{array}   \right) =
\left( \begin{array}{ccc}
A^{-1}  &  H_{r+1}  &  H''
\end{array}  \right)
\left(  \begin{array}{c}
x_1  \\
\ldots \\
x_n
\end{array}  \right) +
H'' \left(  \begin{array}{c}
x_{r+2} ^2  \\
\ldots  \\
x_n ^2
\end{array}  \right).
\end{align*}
Therefore
\begin{align*}
A \left(  \begin{array}{c}
f_1  \\
\ldots  \\
f_r
\end{array}  \right) =
\left(  \begin{array}{ccc}
I_r &  (A H_{r+1})  & (AH'')
\end{array}  \right)
\left(  \begin{array}{c}
x_1  \\
\ldots  \\
x_n
\end{array}  \right) +
(AH'') \left(   \begin{array}{c}
x_{r+2}^2  \\
\ldots  \\
x_n ^2
\end{array}  \right)
\end{align*}
and
\begin{equation}    \label{SimplifiedEquations}
A \left(  \begin{array}{c}
f_1  \\
\ldots  \\
f_r
\end{array}  \right) = x_i + (AH_{r+1}) _i x_{r+1} + \sum\limits _{j=r+2} ^n (AH'') _{ij} x_j(x_j+1) \ \ \mbox{  for  } \ \ \forall 1 \leq i \leq r.
\end{equation}
Denote
\[
g_i (x_{r+1}, \ldots , x_n) := - (AH_{r+1}) _i x_{r+1} - \sum\limits _{j=r+2}^n (AH'') _{ij} x_j (x_j +1) \in \FF_q [ x_{r+1}, \ldots , x_n]
\]
for $\forall r+1 \leq i \leq n$.
The fiber of the puncturing $\Pi _{\rho} : X \rightarrow \overline{\FF_q} ^{n-r}$ over an  arbitrary point
$a'= (a_{r+1}, \ldots , a_n) \in \overline{\FF_q} ^{n-r}$  is  claimed to  consist of a single point
\begin{equation}   \label{FibreOfPuncturing}
\Pi _{\rho} ^{-1} (a') \cap X = \{ (g_1(a'), \ldots , g_r (a'), a') \}.
\end{equation}
Indeed, if $a = (a_1, \ldots , a_r, a') \in \Pi _{\rho} ^{-1} (a') \cap X$ then $f_1 (a) = \ldots = f_r (a) =0$ and (\ref{SimplifiedEquations}) implies that
 $a_i - g_i (a') =0$ for $\forall 1 \leq i \leq r$.
 Conversely, if $a_i = g_i (a')$ for $\forall 1 \leq i \leq r$ then (\ref{SimplifiedEquations}) requires
 \[
 A \left(  \begin{array}{c}
 f_1 (a)  \\
 \ldots  \\
 F_r (a)
 \end{array}  \right) = 0_{r \times 1},
 \]
 whereas
 \[
 \left(  \begin{array}{c}
 f_1 (a)  \\
 \ldots  \\
 f_r (a)
 \end{array}  \right) = A^{-1} 0_{r \times 1} = (H_1 \ldots H_r) 0_{r \times 1} = 0_{r \times 1}
 \]
 and $(g_1(a'), \ldots , g_r (a'), a') \in \Pi _{\rho} ^{-1} (a') \cap X$.
 In such a way,
 \[
 \Pi _{\rho} ^{-1} : \overline{\FF_q} ^{n-r} \longrightarrow X,
 \]
 \[
 \Pi _{\rho} ^{-1} (x_{r+1}, \ldots , x_n) = (g_1(x_{r+1}, \ldots , x_n), \ldots , g_r (x_{r+1}, \ldots , x_n), x_{r+1}, \ldots , x_n)
 \]
 is a correctly defined morphism of affine varieties and $\Pi _{\rho} : X \rightarrow \overline{\FF_q} ^{n-r}$ is a biregular map.

According to
 \[
 \frac{\partial f_i}{\partial x_j} = H_{ij} \ \ \mbox{  for  } \ \ 1 \leq j \leq r+1, \ \ 1 \leq i \leq r \ \ \mbox{  and  }
 \]
 \[
 \frac{\partial f_i}{\partial x_j} = H_{ij} (1 + 2 x_j) \ \ \mbox{  for  } \ \  r+2 \leq j \leq n, \ \ 1 \leq i \leq r,
 \]
 the Jacobian matrix of the defining equations of $X$ is
 \[
 \frac{\partial f}{\partial x} = (H_1 \ldots H_r [(1 + 2 x_{r+2}) H_{r+2}] \ldots [(1 + 2 x_n) H_n]).
 \]
Due  to $H_{r-1} + H_r - H_{r+1} = 0_{r \times 1}$, the word $c := (0^{r-2}, 1,1,-1, 0^{n-r-1})$ belongs to all tangent codes
  $T_b (X, \FF_{q^{\delta (b)}})$, $b \in X$.
 Therefore $d(T_b (X, \FF_q ^{\delta (b)})) \leq 3$ and $X = X^{( \leq 3)}$.

 If ${\rm char} \FF_q =2$ then  $\frac{\partial f}{\partial x} (a) = H_{q,r}$ for all $a \in X$ and $T_a (X, \FF_q)$ are Hamming codes at all the $\FF_q$-rational points $a \in X( \FF_q)$.
 Note that the polynomials
 \[
 g_1 (x_{r+1}, \ldots , x_n), \ldots , g_r (x_{r+1}, \ldots , x_n)
 \]
  have coefficients from $\FF_q$, so that for any
 $a'\in \FF_q ^{n-r}$ one has $\Pi _{\rho} ^{-1} (a') \cap X \subseteq X( \FF_q )$.
 Thus, $\Pi _{\rho} : X \rightarrow \overline{\FF_q} ^{n-r}$ restricts to a bijective map $\Pi _{\rho} : X( \FF_q) \rightarrow \FF_q ^{n-r}$ and the cardinality
 $|X( \FF_q)| = | \FF_q ^{n-r}| = q^{n-r}$.

 If ${\rm char} \FF_q = p \geq 3$ is an odd prime then for any $a \in X( \FF_q ) \setminus V \left( \prod\limits _{j = r+2} ^n (2 x_j +1) \right)$  the columns of $\frac{\partial f}{\partial x} (a)$ are pairwise non-proportional.
 Therefore $\frac{\partial f}{\partial x} (a)$  is a  parity check matrix of the Hamming code  $\cH _{q,r}  = T_a (X, \FF_q)$   for
  $\forall a \in X( \FF_q ) \setminus V \left( \prod\limits _{j = r+2} ^n (2 x_j +1) \right)$.
  The bijective map $\Pi _{\rho} : X( \FF_q) \rightarrow \FF_q ^{n-r}$ restricts to an isomorphism
  \[
  \Pi _{\rho} : X( \FF_q) \setminus V \left( \prod\limits _{j = r+2} ^n (2 x_j +1) \right) \longrightarrow
   \FF_q ^{n-r} \setminus V \left(  \prod\limits _{j = r+2} ^n ( 2 x_j +1) \right).
   \]
   Making use of
   \[
   \FF_q ^{n-r} \setminus V \left( \prod\limits _{j = r+2} ^n ( 2 x_j +1 ) \right) \simeq
   \FF_q  \times \left( \FF_q \setminus \{ - [ 2 ( {\rm mod} \, p) ] ^{-1} \} \right) ^{n-r-1},
   \]
   one concludes that $\left| X( \FF_q) \setminus V \left( \prod\limits _{j = r+2} ^n (2 x_j +1) \right) \right| = q(q-1) ^{n-r-1}$.

\end{proof}

\section{ Operations on tangent codes and affine varieties}

\subsection{  Tangent  codes to punctured varieties  }

The present section provides a sufficient condition for a Zariski tangent space to  the  punctured variety $\Pi _{\gamma}(X)$ to be the puncturing at $\gamma$ of the corresponding Zariski tangent space to $X$.
 This is shown to hold on a Zariski dense subset $W_t$ of $X$.
 The tangent vectors to $\Pi _{\gamma}(X)$ of minimum weight at the points of $\Pi _{\gamma} (W_{\gamma}) $ turn to be the puncturings at $\gamma$ of the tangent vectors to $X$ of minimum weight, containing $\gamma$ in its support.

\begin{lemma}   \label{TangentCodesOfPuncturedVariety}
For any  puncturing $\Pi _{\gamma} : X \rightarrow \Pi _{\gamma} (X)$ of a $k$-dimensional affine variety $X / \FF_q \subset \overline{\FF_q} ^N$ with $l$-dimensional generic fibres $\Pi _{\gamma} ^{-1} ( \Pi _{\gamma} (a))$, $a \in X$,  there exists a factorization
\[
\begin{diagram}
\node{X}  \arrow{s,l}{\Pi _{\gamma}}  \arrow{e,t}{\Pi _{\gamma \setminus \beta}}  \node{\Pi _{\gamma \setminus \beta} (X)}  \arrow{sw,r}{\Pi _{\beta}}  \\
\node{\Pi _{\gamma} (X)}
\end{diagram}
\]
through a finite puncturing $\Pi _{\gamma \setminus \beta} : X \rightarrow \Pi _{\gamma \setminus \beta} (X)$, followed by  a   puncturing
$\Pi _{\beta} : \Pi _{\gamma \setminus \beta} (X) \rightarrow \Pi _{\gamma} (X)$ with generic fibres $\overline{\FF_q} ^l$.
Moreover, if  the finite puncturing $\Pi _{\gamma \setminus \beta} : X \rightarrow \Pi _{\gamma \setminus \beta} (X)$ is separable then there is a Zariski dense subset
\[
U_{\gamma \setminus \beta} := {\rm Etale} ( \Pi _{\gamma \setminus \beta}) \cap
\Pi _{\gamma \setminus \beta} ^{-1} ( \Pi _{\gamma \setminus \beta} (X) ^{\rm smooth} ) \subseteq X,
\]
at which the puncturings
\begin{equation}    \label{GenericPuncturedTangentCode}
\Pi _{\gamma} T_a (X, \FF_{q^{\delta (a)}}) = T_{\Pi _{\gamma} (a)} ( \Pi _{\gamma} (X), \FF_{q^{\delta (a)}})
\end{equation}
of the tangent codes to $X$ at $\gamma$ are tangent codes to the puncturing $\Pi _{\gamma} (X)$ of $X$ at $\gamma$.
\end{lemma}

\begin{proof}

The surjective morphism $\Pi _{\gamma}: X \rightarrow \Pi _{\gamma} (X)$ induces  an embedding
$\Pi _{\gamma} ^* : \overline{\FF_q} ( \Pi _{\gamma} (X))  \hookrightarrow \overline{\FF_q} (X)$ of the corresponding function fields.
The transcendence degree ${\rm tr deg} _{\overline{\FF_q} ( \Pi _{\gamma} (X))} \overline{\FF_q} (X)$ of $\overline{\FF_q} (X)$ over
$\overline{\FF_q} ( \Pi _{\gamma} (X))$ coincides with the dimension $l$ of a generic fibre of $\Pi _{\gamma}: X \rightarrow \Pi _{\gamma} (X) $.
If $\overline{x_{\beta}} = \{ \overline{x_{\beta_1}}, \ldots , \overline{x_{\beta _l}} \}$,
$\overline{x_{\beta _i}} := x_{\beta _i} + I(X, \overline{\FF_q} ) \in \overline{\FF_q} [X] \subset \overline{\FF_q} (X)$ is a coordinate transcendence basis of $\overline{\FF_q} (X)$ over $\overline{\FF_q} ( \Pi _{\gamma}  (X))$ for a subset $\beta \subseteq \gamma$  then the field
$\overline{\FF_q}  ( \Pi _{\gamma} (X)) ( \overline{x_{ \beta}}) = \overline{\FF_q} ( \Pi _{\gamma \setminus \beta} (X))$ is a purely trans\-cen\-den\-tal extension of $\overline{\FF_q} ( \Pi _{\gamma} (X))$ of degree $l$ and $\overline{\FF_q} (X)$ is   a finite extension of
 $\overline{\FF_q} ( \Pi _{\gamma \setminus \beta} (X))$.
Thus,  $\Pi _{\gamma \setminus \beta} : X \rightarrow \Pi _{\gamma \setminus \beta} (X)$ is a finite morphism and the generic fibres of
 $\Pi _{\beta} : \Pi _{\gamma \setminus \beta} (X) \rightarrow \Pi _{\gamma} (X)$ are isomorphic to the affine space $\overline{\FF_q} ^l$.

 At an arbitrary point $a \in {\rm Etale} ( \Pi _{\gamma \setminus \beta} )$, one has a commutative diagram
 \[
 \begin{diagram}
 \node{T_a (X, \FF_{q^{\delta (a)}})}  \arrow{e,t}{(d \Pi _{\gamma \setminus \beta}) _a}   \arrow{se,r}{(d \Pi _{\gamma}) _a}
 \node{T_{\Pi _{\gamma \setminus \beta} (a)} ( \Pi _{\gamma \setminus \beta}  (X), \FF_{q^{\delta (a)}})}
 \arrow{s,r}{(d \Pi _{\beta}) _{\Pi _{\gamma \setminus \beta} (a)}}  \\
 \node{\mbox{   }}   \node{T_{\Pi _{\gamma} (a)} ( \Pi _{\gamma} (X), \FF_{q^{\delta (a)}})}
 \end{diagram}
 \]
 of puncturings of tangent codes.
 By the very definition of the etal\'{e} locus ${\rm Etale} (\Pi _{\gamma \setminus \beta})$ of $\Pi _{\gamma \setminus \beta}$, the
  $\FF_{q^{\delta (a)}}$-linear map $(d \Pi _{\gamma \setminus \beta}) _a$ is injective.
  Therefore
  \begin{align*}
  \dim _{\FF_{q^{\delta (a)}}} \ker (d \Pi _{\gamma}) _a =  \\
   \dim _{\FF_{q^{\delta (a)}}} \ker (d \Pi _{\beta}) _{\Pi _{\gamma \setminus \beta} (a)} :
   [ d( \Pi _{\gamma \setminus \beta}) _a T_a (X, \FF_{q^{\delta (a)}})  \longrightarrow \Pi _{\gamma} T_a (X, \FF_{q^{\delta (a)}}) ] \leq  \\
   \dim _{\FF_q^{\delta (a)}} \ker (d \Pi _{\beta}) _{\Pi _{\gamma \setminus \beta} (a)} :
   [ T_{\Pi _{\gamma \setminus \beta} (a)} ( \Pi _{\gamma \setminus \beta} (X), \FF_{q^{\delta (a)}}) \longrightarrow
    T_{\Pi _{\gamma}(a)}  ( \Pi _{\gamma} (X), \FF_{q^{\delta (a)}}) ].
   \end{align*}
 The puncturing $(d \Pi _{\beta}) _{\Pi _{\gamma \setminus \beta} (a)} = \Pi _{\beta}$ at $|\beta| = l$ coordinates has kernel of dimension $\leq l$, so that $\dim _{\FF_{q^{\delta (a)}}} \ker (d \Pi _{\gamma}) _a \leq l$.
 If $a \in {\rm Etale} ( \Pi _{\gamma \setminus \beta} ) \cap \Pi _{\gamma} ^{-1} ( \Pi _{\gamma} (X)^{\rm smooth} )$ then
$\dim T_{\Pi _{\gamma} (a)} ( \Pi _{\gamma} (X), \FF_{q^{\delta (a)}}) = \dim \Pi _{\gamma} (X) = k-l$ and the  code
$(d \Pi _{\gamma}) _a T_a (X, \FF_{q^{\delta (a)}})$, contained in  $T_{\Pi _{\gamma} (a)} ( \Pi _{\gamma} (X), \FF_{q^{\delta (a)}})$
 is of dimension
 \[
 k-l \geq \dim _{\FF_{q^{\delta (a)}}} (d \Pi _{\gamma}) _a T_a (X, \FF_{q^{\delta (a)}}) =
 \dim T_a (X, \FF_{q^{\delta (a)}}) - \dim _{\FF_{q^{\delta (a)}}} \ker (d \Pi _{\gamma}) _a \geq k-l.
 \]
 As a result,
 \[
 \dim _{\FF_{q^{\delta (a)}}} (d \Pi _{\gamma}) _a T_a (X, \FF_{q^{\delta (a)}}) = k-l = \dim T_{\Pi _{\gamma} (a)} ( \Pi _{\gamma} (X), \FF_{q^{\delta (a)}})
 \]
 and there follows ( \ref{GenericPuncturedTangentCode}).

\end{proof}

Let $C \subset \FF_q ^n$ be a linear code of minimum distance $d(C) =d$.
For any $\nu  \in \NN$, $\nu \geq d$ denote by
\[
{\rm Wt} _{\nu} (C) := \{ c \in C \ \ \vert \ \  {\rm wt} (c) = \nu \}
\]
the set of the words of $C$ of weight $\nu$.
The following simple lemma will be used for the description of the words of minimum weight in a generic tangent code to $\Pi _{\gamma} (X)$.

\begin{lemma}   \label{PuncturedCode}
Let $C \subset \FF_q ^n$ be an $\FF_q$-linear code of minimum distance $d(C) =d > s$ and $\Pi _{\gamma} :  C \rightarrow \Pi _{\gamma} (C)$ be the puncturing at some  $\gamma = \{ \gamma _1, \ldots , \gamma _s \} \in \Sigma _s (1, \ldots , n)$.
Then the punctured code  $\Pi _{\gamma}(C) \subset \FF_q ^{n-s}$ is of minimum distance $d( \Pi _{\gamma}(C)) \geq d-s$ and the words
\[
{\rm Wt} _{d-s} ( \Pi _{\gamma} (C)) = \Pi _{\gamma} ( {\rm Wt} _d (C) \setminus V( x_{\gamma _1} \ldots x_{\gamma _s}))
\]
of $\Pi _{\gamma} (C)$ of weight $d-s$ are exactly the punctures of the words $c \in C$ of minimum weight $d$, whose support contains $\gamma$.

In particular, $d ( \Pi _{\gamma} (C)) = d-s$ exactly when $C$ contains a word $c$ with support ${\rm Supp} (c) \supseteq \gamma$, $|{\rm Supp} (c) | =d$.
\end{lemma}

\begin{proof}

By an induction on $s$, if $\gamma = \{ \gamma _1 \}$ and $d>1$ then
\begin{equation}      \label{1StepFormula}
{\rm Wt} _{d-1} ( \Pi _{\gamma _1} (C) ) = \Pi _{\gamma _1}  ( {\rm Wt} _d (C) \setminus V( \gamma _1) )
\end{equation}
and $d( \Pi _{\gamma _1} (C)) \geq d-1$.
For an arbitrary $\gamma = \{ \gamma _1, \ldots , \gamma _{s-1}, \gamma _s \} \in \Sigma _s (1, \ldots , n)$, let us denote
 $\gamma ':= \{ \gamma _1, \ldots , \gamma _{s-1} \} \in \Sigma _{s-1} (1, \ldots , n)$ and assume that
\begin{equation}    \label{IndHypothesis}
{\rm Wt} _{d-s+1} ( \Pi _{\gamma'} (C)) = \Pi _{\gamma'} ( {\rm Wt} _d (C) \setminus V ( x_{\gamma _1} \ldots x_{\gamma _{s-1}} ),
\end{equation}
 $d( \Pi _{\gamma'} (C)) \geq d-s+1 >1$.
The application of (\ref{1StepFormula}) to $\Pi _{\gamma'} (C)$ provides
\[
{\rm Wt} _{d-s} ( \Pi _{\gamma} (C)) = {\rm Wt} _{d-s} ( \Pi _{\gamma _s} \Pi _{\gamma'} (C)) =
\Pi _{\gamma _s} ( {\rm Wt} _{d-s+1} ( \Pi _{\gamma '} (C)) \setminus V( x_{\gamma _s}) )
\]
and $d( \Pi _{\gamma} (C)) \geq d-s$.
By the inductional hypothesis (\ref{IndHypothesis}) one has
\begin{align*}
{\rm Wt} _{d-s+1} ( \Pi _{\gamma'} (C)) \setminus V ( x_{\gamma _s})  =
\Pi _{\gamma'} ( {\rm Wt} _d (C) \setminus V( x_{\gamma _1} \ldots x_{\gamma _{s-1}} )) \setminus V( x_{\gamma _s}) =  \\
\Pi _{\gamma'} ( {\rm Wt} _d (C) \setminus V( x_{\gamma _1} \ldots x_{\gamma _s})).
\end{align*}
Therefore
\begin{align*}
{\rm Wt} _{d-s} ( \Pi _{\gamma} (C)) = \Pi _{\gamma _s} \Pi _{\gamma'} ( {\rm Wt} _d (C) \setminus V( x_{\gamma _1} \ldots x_{\gamma _s}) ) =
\Pi _{\gamma} ( {\rm Wt} _d (C) \setminus V( x_{\gamma _1} \ldots x_{\gamma _s})).
\end{align*}

\end{proof}

\begin{corollary}    \label{PropertiesPuncturedTangentCodes}
Let $X / \FF_q \subset \overline{\FF_q}^n$ be an irreducible affine variety, defined over $\FF_q$ with
$X^{(d)} := \{ a \in X \ \ \vert \ \  d( T_a (X, \FF_q)) = d \} \neq \emptyset$, $\Pi _{\beta} : X \rightarrow \Pi _{\beta} (X)$ be a non-finite puncturing
at $|\beta| = d$ variables and $\Pi _{\gamma} : X \rightarrow \Pi _{\gamma} (X)$ be a finite separable puncturing at  a subset   $\gamma \subset \beta$ of cardinality   $|\gamma| = s$.
Then

(i) $ {\rm Etale} ( \Pi _{\gamma}) \cap \Pi _{\gamma} ^{-1} ( \Pi _{\gamma} (X) ^{\rm smooth}) \cap X^{(d)}$ is a Zariski dense subset of $X$;

(ii) at any point $b \in  {\rm Etale} ( \Pi _{\gamma}) \cap \Pi _{\gamma} ^{-1} ( \Pi _{\gamma} (X) ^{\rm smooth}) \cap X^{(d)}$  the
tangent code $T_{\Pi _{\gamma} (b)} ( \Pi _{\gamma} (X), \FF_{q^{\delta (b)}}) = (d\Pi _{\gamma}) _b T_b (X, \FF_{q^{\delta (b)}})$ is of minimum distance
$\geq d-s$ and the words
\[
{\rm Wt} _{d-s} ( T_{\Pi _{\gamma} (b)} ( \Pi _{\gamma} (X), \FF_{q^{\delta (b)}}) ) =
\Pi _{\gamma} ( {\rm Wt} _d ( T_b (X, \FF_{q^{\delta (b)}}) \setminus V ( x_{\gamma _1} \ldots x_{\gamma _s} ))
\]
of $T_{\Pi _{\gamma} (b)} ( \Pi _{\gamma} (X), \FF_{q^{\delta (b)}})$ of minimum  weight $d-s$ are  the punctures of the words of $T_b (X, \FF_{q^{\delta (b)}})$ of weight $d$, whose support contains $\gamma$.
\end{corollary}

\begin{proof}

(i) By Proposition \ref{MinimumDistance} (ii), the set $X^{( d)}$  is Zariski dense in $X$.
Then   Lemma \ref{FiniteSeparableAndEtalePuncturings} (iii) applies  to $\Pi _{\gamma}$  and  provides the Zariski density of
$U_{\gamma}  := {\rm Etale} ( \Pi _{\gamma}) \cap \Pi _{\gamma} ^{-1} ( \Pi _{\gamma} (X) ^{\rm smooth})$ in $X$.
The inclusions
\[
U_{\gamma} \cap X^{(d)} \subseteq U_{\gamma}, \ \ U_{\gamma} \cap X^{(d)} \subseteq X^{(d)}
 \]
 of subsets of $\overline{\FF_q}^n$ imply the opposite inclusions
 \[
 I(U_{\gamma}, \overline{\FF_q}) \subseteq I(U_{\gamma} \cap X^{(d)}, \overline{\FF_q}), \ \
 I(X^{(d)}, \overline{\FF_q}) \subseteq I(U_{\gamma} \cap X^{(d)}, \overline{\FF_q})
 \]
  of the corresponding absolute ideals.
Therefore
\[
I(U_{\gamma}, \overline{\FF_q}) + I(X^{(d)}, \overline{\FF_q}) \subseteq I( U_{\gamma} \cap X^{(d)}, \overline{\FF_q}),
\]
whereas the Zariski closure
\begin{align*}
\overline{U_{\gamma} \cap X^{(d)}}  = VI( U_{\gamma} \cap X^{(d)}, \overline{\FF_q}) \subseteq
V( I(U_{\gamma}, \overline{\FF_q}) + I(X^{(d)}, \overline{\FF_q})) =  \\
 VI(U_{\gamma}, \overline{\FF_q}) \cap VI(X^{(d)}, \overline{\FF_q}) = \overline{U_{\gamma}} \cap \overline{X^{(d)}} = X \cap X = X.
\end{align*}

(ii)  The claim is an immediate consequence of Lemma \ref{PuncturedCode} and the surjectiveness of   the differential
 $(d \Pi _{\gamma}) _b : T_b (X, \FF_{q^{\delta (b)}})  \rightarrow T_{\Pi _{\gamma} (b)} ( \Pi _{\gamma} (X), \FF_{q^{\delta (b)}})$ of
 $\Pi _{\gamma} : X \rightarrow \Pi _{\gamma} (X)$  at any  point $b \in U_{\gamma}$, established in Lemma  \ref{FiniteSeparableAndEtalePuncturings} (ii).

\end{proof}

\subsection{  Shortening and puncturing tangent and gradient codes  }

The present subsection describes the shortening of a tangent code to $X$ on $\gamma$ as a tangent code to  the shortening
$\Pi _{\gamma}( X \cap V( x_i \, \vert \, i \in \gamma))$ of the variety  $X$ on $\gamma$, the puncturing of a gradient code to $X$  at $\gamma$ as a gradient code to   $X\cap V( x_i \, \vert \, i \in \gamma) $  and the shortening of a gradient code to $X$ on $\gamma$ as a gradient code to the puncturing
 $\Pi _{\gamma}(X)$ of $X$ at $\gamma$.

The shortening of a linear code $C \subset \FF_q^n$ on $\gamma \in \Sigma _s (1, \ldots , n)$ is the puncturing
\[
C_{\gamma} := \Pi _{\gamma} (C \cap V( x_i \, \vert \, \forall i \in \gamma ))
\]
 of the subspace $C \cap V( x_i \, \vert \, \forall i \in \gamma ) = \{ c \in C \ \ \vert \ \  c_i =0, \ \ \forall i \in \gamma \}$ of $C$ at $\gamma$.

Let $X \subset \overline{\FF_q} ^n$ be a $k$-dimensional irreducible  affine variety with absolute ideal
$I(X, \overline{\FF_q}) = \langle f_1, \ldots , f_m \rangle _{\overline{\FF_q}}$ for some $f_1, \ldots , f_m \in \FF_q [ x_1, \ldots , x_n]$.
A point $a \in X$ is smooth exactly when ${\rm rk} \frac{\partial f}{\partial x} (a) = n-k$.
If so, then there exist $\lambda \in \Sigma _{n-k} (1, \ldots , m)$ and $\mu \in \Sigma _{n-k} (1, \ldots , n)$ with
 $\det \frac{\partial f_{\lambda}}{\partial x_{\mu}} (a)  \neq 0$.
Denoting
 \[
X^{\rm smooth} ( \mu ) := X \setminus V \left( \det \frac{\partial f_{\lambda}}{\partial x_{\mu}} \ \ \Big \vert \ \
 \lambda \in \Sigma _{n-k} (1, \ldots , m) \right),
\]
one expresses
\[
X^{\rm smooth} = \cup _{\mu \in \Sigma _{n-k} (1, \ldots , n)} X^{\rm smooth} ( \mu).
\]
Observe that $X^{\rm smooth} ( \mu)$ are Zariski open subsets of $X$.
As far as $X^{\rm smooth}$ is non-empty,  Zariski open and Zariski  dense in $X$, there exists  $\mu \in \Sigma _{n-k} (1, \ldots , n)$ with non-empty and Zariski dense  $X^{\rm smooth} (\mu) \subseteq X$.

\begin{proposition}     \label{ShorteningTangentCodes}
(i) Let $X \subset \overline{\FF_q} ^n$ be an irreducible affine variety with absolute ideal
 $I(X, \overline{\FF_q}) = \langle f_1, \ldots , f_m \rangle _{\overline{\FF_q}}$ for some $f_1, \ldots , f_m \in \FF_q [ x_1, \ldots , x_n]$,
  $X^{\rm smooth} ( \mu ) \neq \emptyset$ for some $\mu \in \Sigma _{n-k} (1, \ldots , n)$ and
$\gamma \in \Sigma _s ( \neg \mu )$.
Then at an arbitrary point $a \in X^{\rm smooth} ( \mu) \cap {\rm Etale} ( \Pi _{\gamma}) \cap
 \Pi _{\gamma} ^{-1} (( \Pi _{\gamma} (X \cap V( x_i \, \vert \, i \in \gamma) )^{\rm smooth} ) $ the shortening
\begin{equation}   \label{ShortenedTangentCode}
\begin{split}
T_a (X, \FF_{q^{\delta (a)}}) _{\gamma} := \\
 \Pi _{\gamma} ( T_a (X, \FF_{q^{\delta (a)}}) \cap V( x_i \, \vert \,  i \in \gamma )) =
T_{\Pi _{\gamma} (a)} ( \Pi _{\gamma} (X \cap V( x_i \,\vert \,  i \in \gamma )), \FF_{q^{\delta (a)}})
\end{split}
\end{equation}
of  a  Zariski tangent space  to $X$ on $\gamma$  coincides with the Zariski tangent space  to the shortening  $ X \cap V( x_i \, \vert \, i \in \gamma)$ of $X$ on $\gamma$ and the puncturing
\begin{equation}   \label{PuncturedGradientCode}
\Pi _{\gamma} {\rm grad} _a I(X, \FF_{q^{\delta (a)}}) = {\rm grad} _{\Pi _{\gamma} (a)} I(X \cap V( x_i \, \vert \,  i \in \gamma ), \FF_{q^{\delta (a)}})
\end{equation}
of the gradient code to $X$ at $\gamma$   coincides with the gradient code to the shortening  $  X \cap V( x_i \, \vert \,  i \gamma )$  of $X$ on $\gamma$.

(ii) Suppose that  $X / \FF_q \subset \overline{\FF_q}^n$  is  an irreducible affine variety, defined over $\FF_q$, $\gamma \in \Sigma _s (1, \ldots , n)$, $\beta \subseteq \gamma$   and $x_{\beta}$  is  such a coordinate transcendence basis of $\overline{\FF_q}(X)$ over $\overline{\FF_q} ( \Pi _{\gamma} (X))$ that the finite puncturing $\Pi _{\gamma \setminus \beta} : X \rightarrow \Pi _{\gamma \setminus \beta} (X)$ is separable.
Then at any point $a \in {\rm Etale} ( \Pi _{\gamma \setminus \beta} ) \cap \Pi _{\gamma} ^{-1} ( \Pi _{\gamma} (X) ^{\rm smooth})$ the shortening
\begin{equation}   \label{ShorteningGradientCode}
\begin{split}
{\rm grad} _a I(X, \FF_{q^{\delta (a)}}) _{\gamma} := \\
 \Pi _{\gamma} ( {\rm grad}_a I(X, \FF_{q^{\delta (a)}}) \cap V( x_i \, \vert \,  i \in \gamma )) =
{\rm grad} _{\Pi _{\gamma} (a)} I( \Pi _{\gamma} (X), \FF_{q^{\delta (a)}})
\end{split}
\end{equation}
of the gradient code to $X$ on $\gamma$ coincides with the gradient code to the puncturing  $\Pi _{\gamma} (X)$  of $X$ at $\gamma$.
\end{proposition}

\begin{proof}

(i) The inclusions $X \cap V( x_i \, \vert \, i \in \gamma) \subseteq X$, $X \cap V( x_i \, \vert \, i \in \gamma ) \subseteq V( x_i \, \vert \, i \in \gamma )$ of affine varieties imply the opposite inclusions $I(X, \overline{\FF_q})  \subseteq I(X \cap V( x_i \, \vert \, i \in \gamma), \overline{\FF_q})$,
$\langle x_i \, \vert \, i \in \gamma \rangle _{\overline{\FF_q}} \subseteq I(X \cap V( x_i \, \vert \, i \in \gamma), \overline{\FF_q})$ of the corresponding absolute ideals.
Thus, at any point $a \in X \cap  V( x_i \, \vert \, i \in \gamma )$ one has
\[
T_a (X \cap V( x_i \, \vert \, i \in \gamma), \FF_{q^{\delta (a)}}) \subseteq T_a (X, \FF_{q^{\delta (a)}}) \cap V( x_i \, \vert \, i \in \gamma),
\]
whereas
\[
\Pi _{\gamma} T_a (X \cap V( x_i \, \vert \, i \in \gamma ), \FF_{q^{\delta (a)}} ) \subseteq
\Pi _{\gamma} ( T_a (X, \FF_{q^{\delta (a)}}) \cap V( x_i \, \vert \, i \in \gamma)).
\]
Moreover, if $a \in {\rm Etale} ( \Pi _{\gamma}) \cap \Pi _{\gamma} ^{-1} ( \Pi _{\gamma} ( X \cap V( x_i \, \vert \, i \in \gamma)    ) ^{\rm smooth}    )$
then
\[
\Pi _{\gamma} T_a (X \cap V( x_i \, \vert \, i \in \gamma ), \FF_{q^{\delta (a)}} ) =
T_{\Pi _{\gamma} (a)} ( \Pi _{\gamma} (X \cap V( x_i \, \vert \, i \in \gamma )), \FF_{q^{\delta (a)}} )
\]
 and  the Zariski tangent space
\[
T_{\Pi _{\gamma} (a)} ( \Pi _{\gamma} (X \cap V( x_i \, \vert \, i \in \gamma )), \FF_{q^{\delta (a)}} ) \subseteq
\Pi _{\gamma} ( T_a (X, \FF_{q^{\delta (a)}}) \cap V( x_i \, \vert \, i \in \gamma))
\]
to the shortening  $\Pi _{\gamma} (X \cap V( x_i \, \vert \, i  \in \gamma ))$ of $X$ on $\gamma$  is contained in the shortening of the Zariski tangent space to $X$ at $a$.
If $a \in X(\mu) ^{\rm smooth}$ then the parity check matrix
\[
\left(  \begin{array}{cc}
\frac{\partial f}{\partial x_{\gamma}} (a)   &  \frac{\partial f}{\partial x_{\neg \gamma}} (a)  \\
\mbox{   }   &  \mbox{  }   \\
I_{\gamma}  &  0
\end{array}  \right)
\]
of $T_a (X, \FF_{q^{\delta (a)}}) \cap V( x_i \, \vert \,  i \in \gamma )$ is of maximal rank $n-k+s$, due to $\mu \subseteq \neg \gamma$.
The non-existence of a non-zero word $c \in T_a (X, \FF_{q^{\delta (a)}}) \cap V( x_i \, \vert \, i \in \gamma )$ with support ${\rm Supp} (c) \subseteq \gamma$ reveals the injectiveness of the puncturing
\[
(d \Pi _{\gamma} ) _a = \Pi _{\gamma}: T_a (X, \FF_{q^{\delta (a)}})  \cap V( x_i \, \vert \, i \in \gamma)  \longrightarrow \FF_{q^{\delta (a)}} ^{n - |\gamma|}.
\]
Therefore
\begin{align*}
 \dim _{\FF_{q^{\delta (a)}}} \Pi _{\gamma} (T_a (X, \FF_{q^{\delta (a)}}) \cap V( x_i \, \vert \, i \in \gamma )) =  \\
  \dim _{\FF_{q^{\delta (a)}}}  T_a (X, \FF_{q^{\delta (a)}}) \cap V( x_i \, \vert \, i \in \gamma ) = k-s.
\end{align*}
On the other hand, at any smooth point $\Pi _{\gamma} (a) \in \Pi _{\gamma} (X \cap V( x_i \, \vert \, i \in \gamma ))^{\rm smooth}$ one has
\[
\dim T_{\Pi _{\gamma} (a)} (\Pi _{\gamma} ( X \cap V (x_i \, \vert \, i \in \gamma )), \FF_{q^{\delta (a)}}) =
\dim \Pi _{\gamma} (X \cap V( x_i \, \vert \, i \in \gamma ) ).
\]
Since $\Pi _{\gamma} : X \cap V( x_i \, \vert \, i \in \gamma) \rightarrow  \Pi _{\gamma} ( X \cap V( x_i \, \vert \, i \in \gamma))$ is biregular,
there holds $\dim \Pi _{\gamma} (X \cap V( x_i \, \vert \, i \in \gamma)) = \dim X \cap V( x_i \, \vert \, i \in \gamma)  \geq k-s$,
whereas
\[
\dim T_{\Pi _{\gamma} (a)} ( \Pi _{\gamma} (X \cap V( x_i \, \vert \, i \in \gamma )), \FF_{q^{\delta(a)}}) =
\dim _{\FF_{q^{\delta (a)}}} \Pi _{\gamma} (T_a (X, \FF_{q^{\delta (a)}} ) \cap V( x_i \, \vert \, i \in \gamma )).
\]
That justifies  (\ref{ShortenedTangentCode}).

For an arbitrary linear code $C \subset \FF_{q^{\delta (a)}} ^n$ with dual code $C^{\perp} \subset \FF_{q^{\delta (a)}} ^n$ and an arbitrary index set
 $\gamma \in \Sigma _s (1, \ldots , n)$ one has $\Pi _{\gamma} (C^{\perp}) = \Pi _{\gamma} (C \cap V( x_i \, \vert \, i \in \gamma ))^{\perp}$ (cf.\cite{HP}).
 The application of this equality to $C = T_a (X, \FF_{q^{\delta (a)}})$ and (\ref{ShortenedTangentCode}) yields (\ref{PuncturedGradientCode}).

 (ii) Note that $\Pi _{\gamma} (C^{\perp} \cap V( x_i \, \vert \. i \in \gamma )) = \Pi _{\gamma} (C)^{\perp}$ for an arbitrary linear code $C \subset \FF_{q^{\delta (a)}} ^n$ with dual code $C^{\perp} \subset \FF_{q^{\delta (a)}} ^n$ and an arbitrary $\gamma \in \Sigma _s (1, \ldots , n)$ (cf.\cite{HP}).
 Plugging in $C = T_a (X, \FF_{q^{\delta (a)}})$ and combining with (\ref{GenericPuncturedTangentCode}), one obtains (\ref{ShorteningGradientCode}).

\end{proof}

\subsection{ Extension,   direct sum  and the $(u \vert u+v)$ construction  }

For an arbitrary field $\FF_q \subseteq F \subseteq \overline{\FF_q}$ let
\[
\varphi _{n+1} : F^n \longrightarrow F^{n+1},
\]
\[
\varphi _{n+1} (x_1, \ldots , x_n) = \left(  x_1, \ldots , x_n, - \sum\limits _{i=1} ^n x_i \right)  \ \ \mbox{  for  } \ \ \forall (x_1, \ldots , x_n) \in F^n
\]
be the embedding of $F^n$ in $F^{n+1}$ as the hyperplane with equation
\[
x_1 + \ldots + x_n + x_{n+1} =0.
\]
The image $\varphi _{n+1} (C) \subset \FF_{q^m} ^{n+1}$ of a linear code $C \subset \FF_{q^m} ^n$  is called the extension of $C$.
If $X \subset \overline{\FF_q}^n$,   is an affine variety then
\[
\varphi _{n+1} : X \longrightarrow \varphi _{n+1} (X),
\]
\[
\varphi _{n+1} \left( x_1, \ldots , x_n) := (x_1, \ldots , x_n, - \sum\limits _{i=1} ^n x_i \right)
\]
 is a biregular morphism and we say that
 $\varphi _{n+1} (X) \subset \overline{\FF_q} ^{n+1}$ is the extension of $X$.

 For arbitrary linear codes $C_1 \subset \FF_q^n$ and $C_2 \subset \FF_q ^m$ the direct sum
 \[
 C_1 \oplus C_2 := \{ (v_1, v_2) \ \ \vert \ \  v_1 \in C_1, \ \  v_2 \in C_2 \} \subset \FF_q ^{n+m}
 \]
 is a linear code of length $n+m$ and dimension
 \[
 \dim _{\FF_q} (C_1 \oplus C_2) = \dim _{\FF_q} (C_1) + \dim _{\FF_q} (C_2).
\]
It is easy to observe that the  corresponding operation on affine varieties $X \subset \overline{\FF_q}^n$ and $Y \subset \overline{\FF_q}^m$ is the direct product
 \[
 X \times Y := \{ (a,b) \ \ \vert \ \  a \in X, \ \ b \in Y \}.
 \]

The $(u \vert u+v)$ construction on linear codes $C_1 \subset \FF_q^n$ and $C_2 \subset \FF_q^n$  is the linear code
\[
(C_1 \vert C_1 +C_2) := \{ (u, u+v) \ \ \vert \ \  u \in C_1, \ \  v \in C_2 \} \subset \FF_q ^{2n}
\]
of length $2n$.
For an affine variety $X \subset \overline{\FF_q}^n$ and a morphism
\[
g = (g_1, \ldots , g_s) : \overline{\FF_q} ^n \longrightarrow \overline{\FF_q} ^s
\]
 with  $g_1, \ldots , g_s \in \FF_q [ x_1, \ldots , x_n]$, the fibered product
 \[
 X \times _{g(X)} \overline{\FF_q} ^n := \{ (a,b) \in X \times \overline{\FF_q} ^n \ \ \vert \ \  g(a) = g(b) \}
 \]
 is uniquely determined by the commutative diagram
 \[
 \begin{diagram}
 \node{X}  \arrow{s,r}{g}   \node{X \times _{g(X)} \overline{\FF_q}^n}  \arrow{w,t}{{\rm pr} _1}  \arrow{s,r}{{\rm pr} _2}   \\
 \node{g(X)}  \node{\overline{\FF_q}^n}  \arrow{w,t}{g}
 \end{diagram}
 \]
in which the parallel arrows correspond to maps with isomorphic fibres.
 Recall that an  affine variety $Y = g^{-1} (0^s) \subset \overline{\FF_q}^n$ with
 $I(Y, \overline{\FF_q}) = \langle g_1, \ldots , g_s \rangle _{\overline{\FF_q}}$ is a complete intersection if $\dim Y = n-s$.

 \begin{proposition}  \label{ExtendingDirectSumUV}
 (i) Let $X / \FF_q \subset \overline{\FF_q}^n$ be an irreducible affine variety, defined over $\FF_q$ and
 \[
 \varphi _{n+1} : \overline{\FF_q} ^n \longrightarrow \overline{\FF_q}^{n+1},
 \]
 \[
 \varphi _{n+1} (x_1, \ldots , x_n) =  \left( x_1, \ldots , x_n, - \sum\limits _{i=1} ^n x_i \right) \ \
  \mbox{   for  }   \ \ \forall (x_1, \ldots , x_n) \in \overline{\FF_q}^n
 \]
be the extending map.
Then the extension
\begin{equation}   \label{ExtendingVarietiesTangentCodes}
\varphi _{n+1} T_a (X, \FF_{q^{\delta (a)}}) = T_{\varphi _{n+1} (a)} ( \varphi _{n+1} (X), \FF_{q^{\delta (a)}})  \ \
 \mbox{  for  } \ \ \forall a \in X^{\rm smooth}
\end{equation}
of a Zariski tangent space to $X$ is a Zariski tangent space to the extension $\varphi _{n+1}(X)$ of $X$.

(ii) Let $X / \FF_q \subset \overline{\FF_q}^n$ and $Y / \FF _q \subseteq \overline{\FF_q}^n$ be irreducible affine varieties, defined over $\FF_q$.
Then the direct sum
\begin{equation}    \label{DirectSumProduct}
T_a (X, \FF_{q^{\delta (a,b)}}) \oplus T_b (Y, \FF_{q^{\delta (a, b)}}) = T_{(a,  b)} (X \times Y, \FF_{q^{\delta (a,b)}})  \ \
\mbox{   at   }  \ \ \forall  (a,b) \in X^{\rm smooth} \times Y^{\rm smooth}
\end{equation}
of tangent spaces to $X$ and $Y$ is a tangent space to $X \times Y$.

(iii) Let  $X / \FF_q \subset \overline{\FF_q}^n$   be   an irreducible affine variety, defined over $\FF_q$ and
 $g = (g_1, \ldots , g_s) : \overline{\FF_q} ^n \rightarrow \overline{\FF_q}^s$    with $g_1, \ldots , g_s \in \FF_q [ x_1, \ldots , x_n]$
be such a morphism that the fibre $Y _a:= g^{-1} (a) / \FF_q \subset \overline{\FF_q}^n$ through $a \in  X^{\rm smooth}$  is an irreducible complete intersection, containing $a$ in its smooth locus.
Then the $(u\vert u +v)$ construction
\begin{equation}    \label{FiberedProductConstruction}
 \left( T_a (X, \FF_{q^{\delta (a)}}) \ \  \vert \ \ T_a (X, \FF_{q^{\delta (a)}}) + T_a (Y_a, \FF_{q^{\delta (a)}}) \right) =
T_{(a,  a)} (X \times _{g(X)} \overline{\FF_q}^n, \FF_{q^{\delta (a)}})
\end{equation}
of Zariski tangent spaces to $X$ and   $Y_a $   is the tangent space to the fibered product $X \times _{g(X)} \overline{\FF_q}^n$ at $(a, a)$.
\end{proposition}

 \begin{proof}

(i) According to the inclusion
$I(X, \overline{\FF_q})  \subset I( \varphi _{n+1} (X), \overline{\FF_q}) \triangleleft \overline{\FF_q} [ x_1, \ldots , x_n, x_{n+1}]$
of the absolute ideals and  $x_1 + \ldots + x_n + x_{n+1} \in I( \varphi _{n+1}(X), \overline{\FF_q})$, the Zariski tangent space
$T_{\varphi _{n+1} (a)} ( \varphi _{n+1} (X), \FF_{q^{\delta (a)}})$ to $\varphi _{n+1} (X)$ at $\varphi _{n+1} (a) \in \varphi _{n+1} (X)$ is contained in the extension
\[
\varphi _{n+1} T_a (X, \FF_{q^{\delta (a)}}) :=  \left \{  \left(  v, - \sum\limits _{i=1} ^n v_i \right) \in \FF_{q^{\delta (a)}} ^{n+1} \ \ \Big \vert \ \  v \in T_a (X, \FF_{q^{\delta (a)}}) \right \}
\]
 of $T_a (X, \FF_{q^{\delta (a)}})$.
If $a \in X^{\rm smooth}$ then
\[
\dim _{\FF_{q^{\delta (a)}}} \varphi _{n+1} T_a (X, \FF_{q^{\delta (a)}}) = \dim T_a (X, \FF_{q^{\delta (a)}}) = \dim X.
\]
On the other hand, the biregular morphism $\varphi _{n+1} : X \rightarrow \varphi _{n+1} (X)$ has image $\varphi _{n+1} (X)$ of $\dim \varphi _{n+1} (X) = \dim X$, so that
\begin{align*}
\dim X = \dim \varphi _{n+1} (X) \leq \dim T_{\varphi _{n+1}(a)}  ( \varphi _{n+1} (X), \FF_{q^{\delta   (a)}}) \leq  \\
\dim _{\FF_{q^{\delta (a)}}} \varphi _{n+1} T_a (X, \FF_{q^{\delta (a)}}) = \dim X
\end{align*}
and there follows  (\ref{ExtendingVarietiesTangentCodes}) at all $a \in X^{\rm smooth}$.

(ii)  By $I(X, \overline{\FF_q}) \subseteq I(X \times Y, \overline{\FF_q}) \triangleleft \overline{\FF_q} [ x_1, \ldots , x_n, y_1, \ldots , y_m]$ and
$I(Y, \overline{\FF_q}) \triangleleft I(X \times Y, \overline{\FF_q})$ one has
\begin{align*}
T_{(a, b)} (X \times Y, \FF_{q^{\delta (a, b)}}) \subseteq T_a (X, \FF_{q^{\delta (a, b)}}) \times T_b (Y, \FF_{q^{\delta (a, b)}})  \simeq  \\
 T_a (X, \FF_{q^{\delta (a,  b)}}) \oplus T_b (Y, \FF_{q^{\delta (a, b)}}) .
\end{align*}
Note that $\dim (X \times Y) = \dim X + \dim Y$.
If $a \in X^{\rm smooth}$ and $b \in Y^{\rm smooth}$ then
\begin{align*}
\dim X + \dim Y = \dim (X \times Y) \leq \dim T_{(a,  b)} (X \times Y, \FF_{q^{\delta (a,  b)}}) \leq  \\
 \dim T_a (X, \FF_{q^{\delta (a, b)}}) \oplus
\dim T_b (Y, \FF_{q^{\delta (a,  b)}}) = \dim X  + \dim Y,
\end{align*}
whereas  (\ref{DirectSumProduct}).

(iii) The inclusions
 $I(X, \overline{\FF_q}) \subseteq I(X \times _{g(X)} \overline{\FF_q} ^n, \overline{\FF_q}) \triangleleft
 \overline{\FF_q} [ x_1, \ldots , x_n, y_1, \ldots , y_n]$ and
  $g_1(y) - g_1 (x), \ldots , g_s (y) - g_s (x) \in I(X \times _{g(X)} \overline{\FF_q} ^n, \overline{\FF_q})$ require  the Zariski tangent space
   $T_{(a,  b)} (X \times _{g(X)} \overline{\FF_q}^n, \FF_{q^{\delta (a, b)}})$  to be  contained in the $\FF_{q^{\delta (a, b)}}$-linear code
   $C_{(a,  b)}$ of length $2n$ with parity check matrix
   \[
   H_{(a, b)} = \left(  \begin{array}{cc}
   \frac{\partial f}{\partial x} (a)  &  0  \\
   \mbox{  }  &  \mbox{  }  \\
   - \frac{\partial g}{\partial x} (a) &  \frac{\partial g}{\partial y} (b)
   \end{array}  \right).
   \]
   Bearing in mind that $\frac{\partial g}{\partial y}(b)$ is the parity check matrix of $T_b (Y_a, \FF_{q^m})$ for all
    $m \in \NN$ with $\FF_{q^{m}} \ni b$, one concludes that
    \[
    C_{(a, a)} = (T_a (X, \FF_{q^{\delta (a)}}) \ \ \vert \ \  T_a (X, \FF_{q^{\delta (a)}}) + T_a (Y_a, \FF_{q^{\delta (a)}}))
    \]
(cf.\cite{HP}).
According to $a \in X^{\rm smooth} \cap  Y_a ^{\rm smooth}$, one has
\[
\dim T_a (X, \FF_{q^{\delta (a)}}) = \dim X,  \ \
\dim T_a ( Y_a, \FF_{q^{\delta (a)}}) = \dim Y_a   = n-s \ \ \mbox{  and  }
\]
\[
\dim C_{a \times a} = \dim X + n-s.
\]
On the other hand, the fibered product $X \times _{g(X)} \overline{\FF_q}^n$ is cut from $X \times \overline{\FF_q}^n$ by $s$ equations $g_i(x) = g_i(y)$, so that  $\dim (X \times _{g(X)} \overline{\FF_q}^n) \geq \dim X + n-s$.
Now,  the inclusion
\[
T_{a \times a} (X \times _{g(X)} \overline{\FF_q}^n, \FF_{q^{\delta (a)}}) \subseteq C_{a \times a}
\]
 implies that
\begin{align*}
\dim X + n-s \leq \dim (X \times _{g(X)} \overline{\FF_q}^n) \leq \dim T_{a \times a} (X \times _{g(X)} \overline{\FF_q}^n, \FF_{q^{\delta (a)}}) \leq  \\
\dim _{\FF_{q^{\delta (a)}}} C_{a \times a} = \dim X + n-s
\end{align*}
 and justifies  (\ref{FiberedProductConstruction}).

 \end{proof}


\section{Families of Hamming isometries}

Recall that the Hamming distance
\[
d(a,b) := | \{ 1 \leq i \leq n \ \ \vert \ \  a_i \neq b_i \}|
\]
between  two words $a, b \in \FF_q^n$ equals the number of their different components.
A map $\mathcal{I} : \FF_q^n \rightarrow \FF_q^n$ is called a Hamming isometry if it preserves the Hamming distance, i.e.,
 $d( \cI a, \cI b) = d(a,b)$ for $\forall a, b \in \FF_q^n$.
All Hamming isometries are bijective maps.
More precisely, if we assume the existence of different $a, b \in \FF_q^n$ with $\cI a = \cI b$ then $0 = d( \cI a, \cI b) = d(a,b)$ contradicts $a \neq b$.

\subsection{ Finite morphisms with isometric differentials}

 The present subsection provides a  pattern for  a  construction of a  global   morphism $\psi : \overline{\FF_q}^n \rightarrow \overline{\FF_q}^n$ and a hypersurface $V( \psi _o) \subset \overline{\FF_q}^n$, depending explicitly  on $\psi$,  such that the  differentials  of $\psi$  restrict to  linear Hamming isometries
 \[
(d \psi ) _a : T_a (X, \FF_{q^{\delta (a)}}) \longrightarrow T_{\psi (a)} ( \psi (X), \FF_{q^ {\delta (a)}})
\]
on the     tangent  codes   to  a generic     affine variety $X    \nsubseteq V( \psi _o)$ at a  generic  points $a \in X \setminus V( \psi _o)$.

 \begin{proposition}   \label{ExistenceMorphismIsometricDifferentials}
For arbitrary polynomials $\psi _1, \ldots , \psi _n \in \FF_q [ x_1, \ldots , x_n]$ and an arbitrary permutation $\sigma \in {\rm Sym} (n)$, let us consider the morphism
\[
\psi := ( x_{\sigma(1)} \psi _{\sigma (1)} (x_1^p, \ldots , x_n^p), \ldots , x_{\sigma (n)} \psi _{\sigma (n)} ( x_1^p, \ldots , x_n^p)) :
\overline{\FF_q}^n  \longrightarrow \overline{\FF_q}^n
 \]
 and the hypersurface $V( \psi _o) \subset \overline{\FF_q}^n$ with equation
\[
 \psi _o ( x_1, \ldots , x_n) :=  \psi _1(x_1^p, \ldots , x_n^p) \ldots \psi _n (x_1^p, \ldots , x_n^p),
 \]
where $p = {\rm char} (\FF_q)$ stands for the characteristic of the basic field $\FF_q$.
Then any irreducible  affine variety $X \subset \overline{\FF_q}^n$, which is not entirely contained in the hypersurface $V( \psi _o)$ has a non-empty Zariski open, Zariski dense subset
\[
W :=  \left[ X^{\rm smooth} \cap \psi ^{-1} ( \psi (X) ^{\rm smooth}) \right] \setminus V( \psi _o),
 \]
such that the differentials of $\psi$ restrict to $\FF_{q^{\delta (a)}}$-linear Hamming isometries
 \[
 (d \psi ) _a : T_a (X, \FF_{q^{\delta (a)}}) \longrightarrow T_{\psi (a)} (\psi (X), \FF_{q^{\delta (a)}})
 \]
 at all the points $a \in W$.
  \end{proposition}

 \begin{proof}

It suffices to prove the proposition for the ${\mathbb F}_q$-morphism
\[
\varphi : = \sigma ^{-1} \psi = ( \varphi _1 = x_1 \psi _1 (x_1^p, \ldots , x_n^p), \ldots , \varphi _n = x_n \psi _n (x_1^p, \ldots , x_n^p)) :
 X \longrightarrow \psi (X),
\]
as far as the permutation $\sigma ^{-1} \in {\rm Sym} (n)$ coincides with its differentials at any point $a \in \overline{\FF_q}^n$ and is a linear Hamming isometry.
Let
\begin{align*}
  \Phi   _p: \FF_q [ x_1, \ldots , x_n] \longrightarrow \FF_q [ x_1, \ldots , x_n],  \\
 \Phi _p   ( f( x_1, \ldots , x_n)) := f( x_1 ^p, \ldots , x_n ^p)
 \end{align*}
be  the Frobenius automorphism of $\FF_q [ x_1, \ldots , x_n]$ of degree $p = {\rm char} \FF_q$.
The differential
\[
(d \varphi) _a : T_a ( \overline{\FF_q}^n, \FF_{q^{\delta (a)}}) \longrightarrow  T_{\varphi (a)} ( \overline{\FF_q}^n, \FF_{q^{\delta (a)}})
\]
of the morphism  $\varphi  := ( x_1 \psi _1 (x_1^p, \ldots , x_n^p), \ldots , x_n \psi _n ( x_1^p, \ldots , x_n^p))$
 has matrix
 \[
\frac{\partial ( \psi _1, \ldots , \psi  _n)}{\partial ( x_1, \ldots , x_n)} (a) =
\left( \begin{array}{cccc}
 \psi _1  (\Phi _p(a))       &  0  &  \ldots  &  0  \\
0  &    \psi _2 (\Phi _p(a))   &  \ldots &  0  \\
\ldots  &  \ldots  &  \ldots  &  \ldots  \\
0  &  0  &  \ldots  &     \psi _n  (\Phi _p (a))
\end{array}  \right)
\]
  with respect to the basis $\left( \frac{\partial}{\partial x_j} \right) _a$, $ 1 \leq i \leq n$ of $T_a ( \overline{\FF_q}^n, \FF_{q^{\delta (a)}})$.
Note that at any point $a \in (X \setminus V( \psi _o)) $ the differential
 $(d \varphi  ) _a : T_a ( \overline{\FF_q}^n, \FF_{q^{\delta (a)}}) \rightarrow T_{\varphi  (a)} ( \overline{\FF_q}^n, \FF_{q^{\delta (a)}})$
is an $\FF_{q^{\delta (a)}}$-linear Hamming isometry and restricts to an $\FF_{q^{\delta (a)}}$-linear Hamming isometry
\[
(d \varphi ) _a : T_a (X, \FF_{q^{\delta (a)}}) \longrightarrow
 (d \varphi) _a T_a (X, \FF_{q^{\delta (a)}}) \subseteq T_{\varphi (a)} ( \varphi (X), \FF_{q^{\delta (a)}})
\]
onto its image.
We claim that $W \neq \emptyset$ is a non-empty Zariski open subset.
Due to the irreducibility of $X$ it suffices to note that $X^{\rm smooth} \neq \emptyset$, $X \setminus V( \psi _o) \neq \emptyset$ and
the non-empty Zariski open subset $\varphi (X)^{\rm smooth} \subseteq \varphi (X)$ pulls back to a non-empty Zariski open subset
 $X  \cap \varphi ^{-1} ( \varphi (X)^{\rm smooth}) \neq \emptyset$ of $X$.

 \end{proof}

\subsection{Interpolation of linear Hamming isometries by  a morphism }

The next proposition realizes  the members $\cI(a) : \FF_q^n \rightarrow \FF_q ^n$   of an arbitrary family $\cI \rightarrow S$ of $\FF_q$-linear Hamming isometries over $S \subseteq \FF_q^n$   by   the differentials  $(d \varphi) _a = \cI(a)$ of an appropriate
$\FF_q$-morphism $\varphi : \overline{\FF_q}^n \rightarrow   \overline{\FF_q}^n$.

\begin{proposition}    \label{DifferentialRealizationOfIsometries}
Let $\cI  \rightarrow S$ be a family of $\FF_q$-linear Hamming isometries  $\cI (a) \in {\rm GL}(n, \FF_q)$, $\cI (a) : \FF_q^n \rightarrow \FF_q^n$, parameterized by a subset $ S \subseteq \FF_q ^n$.
Then there exists an $\FF_q$-morphism $\varphi = ( \varphi _1, \ldots , \varphi _n) : \overline{\FF_q}^n \rightarrow \overline{\FF_q}^n$, whose differentials
 $(d \varphi ) _a = \cI (a)$ at $\forall a \in S$ coincide with the given isometries.
\end{proposition}

\begin{proof}

Let us consider the polynomials
\[
\varphi _i (x_1, \ldots , x_n) :=
\sum\limits _{b \in \Phi _p (S)} \left[ \sum\limits _{j=1} ^n \cI ( \Phi _p ^{-1} (b)) _{ij} (x_j - x_j ^q) \right]
 L^{b_1} _{\FF_q} ( x_1 ^p) \ldots L^{b_n} _{\FF_q} (x_n ^p)
\]
for $1 \leq i \leq n$, where
\begin{align*}
\Phi _p : \overline{\FF_q}^n \longrightarrow \overline{\FF_q}^n,  \\
\Phi _p (a_1, \ldots , a_n) = (a_1^p, \ldots , a_n ^p) \ \ \mbox{  for } \ \ \forall a = (a_1, \ldots , a_n) \in \overline{\FF_q} ^n
\end{align*}
 is  the Frobenius automorphism of degree $p = {\rm char} \FF_q$ and
\[
L_{\FF_q} ^{\beta} (t) := \prod\limits _{\alpha \in \FF_a \setminus \{ \beta \} } \frac{t - \alpha}{\beta - \alpha} =
\begin{cases}
- t^{q-1} +1 & \text{ for $\beta =0$, }  \\
- t^{q-1} - \sum\limits _{s=1} ^{q-2} \beta ^{-s} t^s &  \text{ for $\beta \in \FF_q^* $ }
\end{cases}
\]
    stand for the Lagrange basis polynomials, used in Proposition \ref{DestabilizationMinDist}.
Straightforwardly,
\[
\frac{\partial \varphi _i}{\partial x_j} =
\sum\limits _{b \in \Phi _p (S)} \cI ( \Phi _p ^{-1} (b)) _{ij} L^{b_1} _{\FF_q} ( x_1 ^p) \ldots L ^{b_n} _{\FF_q} ( x_n ^p)
\]
for $\forall 1 \leq i,j \leq n$, whereas
\[
\frac{\partial \varphi _i}{\partial x_j} (a) = \cI (a)_{ij} \ \ \mbox{  at   } \ \ \forall a \in S \subseteq \FF_q^n.
\]
Therefore $\cI (a) \in {\rm GL} (n, \FF_q)$ is the matrix of the differential
\[
(d \varphi) _a : T_a ( \overline{\FF_q}^n, \FF_q) \longrightarrow T_{\varphi (a)} ( \overline{\FF_q}^n, \FF_q)
\]
with respect to the basis $\left( \frac{\partial}{\partial x_j} \right) _a$, $1 \leq j \leq n$ of $T_a ( \overline{\FF_q}^n, {\mathbb F}_q)$.

\end{proof}

Note  that the Frobenius automorphism
\[
\Phi _q : \overline{\FF_q}^n \longrightarrow \overline{\FF_q}^n,
\]
\[
\Phi _q (x_1, \ldots, x_n) = (x_1^q, \ldots , x_n ^q)
\]
 restricts to a bijective morphism $\Phi _q : X \rightarrow X$ on any  affine variety $X / \FF_q \subset \overline{\FF_q}^n$, defined over $\FF_q$.
   The morphism $\Phi _q$ is not an isomorphism,  as far as its inverse map  $\Phi _q ^{-1} : X \rightarrow X$ is not a morphism.
   For any $m \in \NN$ there arises a bijective map  $\Phi _q : X( \FF_{q^m}) \rightarrow X( \FF_{q^m})$ of the  set
$X( \FF_{q^m})$ of the $\FF_{q^m}$-rational points of $X$.
In  particular, the Frobenius automorphism $\Phi _q = {\rm Id} : X( \FF_q) \rightarrow X( \FF_q)$ restricts to the identity  on  the $\FF_q$-rational points of $X$.
One can view
\[
\Phi _q : T(X, \FF_{q^m}) \longrightarrow T(X, \FF_{q^m})
\]
as a non-linear Hamming isometry of the Zariski tangent bundles  $T(X, \FF_{q^m})$ for any $m \in \NN$.
Note that $\Phi _q : T_a (X, \FF_{q^m}) \rightarrow T_{\Phi _q(a)} (X, \FF_{q^m})$ interchanges the fibres over
 $a \in X( \FF _{q^m}) \setminus X( \FF_q)$ and  acts on the fibres $\Phi _q : T_a(X, \FF_{q^m}) \rightarrow T_a(X, \FF_{q^m})$ over $a \in X( \FF_q)$.


 \end{document}